%% file: main.tex
\RequirePackage{amsthm} 
\documentclass[pdflatex,sn-mathphys,Numbered]{sn-jnl}

\usepackage[T1]{fontenc}
\usepackage[utf8]{inputenc}
\usepackage[english]{babel}

\usepackage{amsmath,amssymb,amsfonts}
\usepackage[cmintegrals]{newtxmath}
\usepackage{mathtools}

\usepackage{graphicx}
\usepackage{subcaption}

\usepackage{mathrsfs}
\usepackage[title]{appendix}
\usepackage{manyfoot}

\usepackage{cancel}
\usepackage{keycommand}
\usepackage{listings}
\usepackage{textcomp}
\usepackage{microtype}  
\usepackage{hyperref}
\usepackage{booktabs}   
\usepackage{tabularx}
\usepackage{multicol,multirow}
\usepackage{xcolor}
\usepackage{todonotes}  
\usepackage{floatrow}

\usepackage{algorithm}
\usepackage{algorithmicx}
\usepackage[noend]{algpseudocode}

\usepackage{bussproofs}
\usepackage{mathpartir}
\usepackage{prftree}

\usepackage{paralist}

\usepackage{tikz}
\usetikzlibrary{arrows, arrows.meta, backgrounds, calc, decorations.markings, positioning, shapes.geometric}

\usepackage[cal=zapfc]{mathalfa}

\theoremstyle{thmstyleone}
\newtheorem{theorem}{Theorem}
\newtheorem{corollary}{Corollary}

\theoremstyle{thmstyletwo}
\newtheorem{assumption}{Assumption}
\newtheorem{remark}{Remark}

\theoremstyle{thmstylethree}
\newtheorem{example}{Example}
\newtheorem{definition}{Definition}

\ifdefined\AddToHook
\makeatletter
\AddToHook{env/tabular/begin}{\let\input\@@input}
\makeatother
\fi

\newkeycommand{\texthcancel}[hshiftstart=0pt,vshiftstart=0pt,hshiftend=0pt,vshiftend=0pt,color=black][1]{%
  \tikz[baseline=(tocancel.base)]{
        \node[inner sep=0pt,outer sep=0pt] (tocancel) {#1};
        \draw[\commandkey{color}]
            ($(tocancel.south west)+(\commandkey{hshiftstart},\commandkey{vshiftstart})$) --
            ($(tocancel.north east)+(\commandkey{hshiftend},\commandkey{vshiftend})$);
  }%
}%
\newcommand{\hcancel}[1]{\text{\texthcancel{${#1}$}}}

\def\orcidID#1{\href{http://orcid.org/#1}{\raisebox{-1.25pt}{\includegraphics{orcid_color.eps}}}}

\renewcommand{\paragraph}[1]{\par\smallskip\noindent\textbf{#1}}

\input{latex-definitions.tex}

\DeclareMathOperator{\E}{\mathbb{E}}

\usepackage{array}
\newcommand{\PreserveBackslash}[1]{\let\temp=\\#1\let\\=\temp}
\newcolumntype{C}[1]{>{\PreserveBackslash\centering}p{#1}}

\newcommand{\gray}[1] {
    \protect\leavevmode
    \begingroup
        \color{gray}%
        #1%
    \endgroup
}

\title{SAT Solving for Variants of First-Order Subsumption}

\author*[1]{\fnm{Robin} \sur{Coutelier}}\email{robin.coutelier@tuwien.ac.at}
\equalcont{These authors contributed equally to this work.}

\author[1]{\fnm{Jakob} \sur{Rath}}\email{jakob.rath@tuwien.ac.at}
\equalcont{These authors contributed equally to this work.}

\author[1]{\fnm{Michael} \sur{Rawson}}\email{michael@rawsons.uk}

\author[2]{\fnm{Armin} \sur{Biere}}\email{biere@cs.uni-freiburg.de}

\author[1]{\fnm{Laura} \sur{Kovács}}\email{laura.kovacs@tuwien.ac.at}

\affil[1]{\orgname{TU Wien}, \orgaddress{\state{Vienna}, \country{Austria}}}

\affil[2]{\orgname{University of Freiburg}, \orgaddress{\state{Freiburg}, \country{Germany}}}

\abstract{%
    \input{0-abstract.tex}    \begin{tikzpicture}[remember picture, overlay]%
        \node[anchor=south,align=left,text width=11cm,minimum height=3cm] at (current page.south) {
            This version of the article has been accepted for publication,
            after peer review but is not the Version of
            Record and does not reflect post-acceptance improvements, or any
            corrections. The Version of Record is available online at:
            \url{https://doi.org/10.1007/s10703-024-00454-1}
        };%
    \end{tikzpicture}%
}


\begin{document}

\maketitle

\section{Introduction\label{sec:introduction}}
\input{1-introduction.tex}

\section{Preliminaries\label{sec:preliminaries}}
\input{2-preliminaries.tex}

\newpage
\section{Subsumption and Subsumption Resolution\label{sec:theory}}
\input{3-theory.tex}

\section{Subsumption Constraints\label{sec:defconstraints}}
\input{4-constraints.tex}

\section{SAT Formalization of Subsumption Constraints\label{sec:satconstraints}}
\input{4N-encodings}

\section{SAT Solving for Subsumption Variants\label{sec:implementation}}
\input{5-implementation.tex}

\section{Solving Heuristics for Subsumption Variants\label{sec:heuristics}}
\input{6-heuristics.tex}

\section{Experimental Results\label{sec:experiments}}
\input{7-experiments.tex}

\section{Related Work\label{sec:related}}
\input{8-related.tex}

\section{Conclusion\label{sec:conclusion}}
\input{9-conclusion.tex}

\backmatter

\bmhead{Acknowledgments}
We thank Pascal Fontaine (University of Liège, Belgium) for fruitful discussions.
We acknowledge partial support from the ERC Consolidator Grant ARTIST 101002685, the FWF SFB project SpyCoDe F8504, the Austrian FWF project W1255-N23, the WWTF ICT22-007  Grant ForSmart, and the TU Wien Trustworthy Autonomous Cyber-Physical Systems Doctoral College.
This research was funded in whole or in part by the Austrian Science Fund (FWF) [10.55776/F85, 10.55776/W1255]. For open access purposes, the author has applied a CC BY public copyright license to any author accepted manuscript version arising from this submission.
Initial results on this work have been established during a research internship of Robin Coutelier at TU Wien, while he was a master student at the University of Liège, Belgium.

\bibliography{references}

\end{document}

%% file: latex-definitions.tex
\newcommand{\limpl}{\rightarrow}    

\newcommand{\Union}{\bigcup}

\renewcommand{\paragraph}[1]{\par\smallskip\noindent\textbf{#1}}

\newcommand{\vampire}{\textsc{Vampire}}

\newcommand{\bvars}{\mathcal{B}}
\newcommand{\vars}{\mathcal{V}}
\newcommand{\terms}{\mathcal{T}}

\newcommand{\positive}{\mathit{positive}}
\newcommand{\reverse}{\mathit{reverse}}

\newcommand{\PropF}{F}
\newcommand{\sideP}{S}
\newcommand{\mainP}{M}
\newcommand{\sideL}{s}
\newcommand{\mainL}{m}

\newcommand{\IncSubst}{\tilde{\Sigma}}
\newcommand{\substincompat}{\Cap}

\newcommand{\NoSubsumptionResolution}{\texttt{NoSubsumptionResolution}}
\newcommand{\NoSubsumption}{\texttt{NoSubsumption}}
\newcommand{\Subsumption}{\texttt{Subsumed}}

\newcommand{\CommentedOut}[1]{}
\newcommand{\SR}{\texttt{SR}}
\newcommand{\SEnc}{\mathcal{E}_{\texttt{S}}}
\newcommand{\SREnc}{\mathcal{E}_{\SR}}
\newcommand{\DirEnc}{\mathcal{E}_\SR^d}
\newcommand{\IndEnc}{\mathcal{E}_\SR^i}
\newcommand{\MS}{\Pi}
\newcommand{\AMO}{\mathit{AMO}}

\newcommand\upstrut{\rule{0pt}{10pt}}
\newcommand\downstrut{\rule[-6pt]{0pt}{6pt}}

%% file: 0-abstract.tex
Automated reasoners, such as SAT/SMT solvers and first-order provers, are
becoming the backbones of rigorous systems engineering, being used for example
in applications of system verification, program synthesis, and cybersecurity.
Automation in these domains crucially depends on the efficiency of the
underlying reasoners towards finding proofs and/or counterexamples of the task
to be enforced. In order to gain efficiency, automated reasoners use dedicated
proof rules to keep proof search tractable. To this end, (variants of)
subsumption is one of the most important proof rules used by automated
reasoners, ranging from SAT solvers to first-order theorem provers and beyond.

It is common that millions of subsumption checks are performed during proof
search, necessitating efficient implementations.  However, in contrast to
propositional subsumption as used by SAT solvers and implemented using
sophisticated polynomial algorithms, first-order subsumption in first-order
theorem provers involves NP-complete search queries, turning the efficient use
of first-order subsumption into a huge practical burden.

In this paper we argue that the integration of a dedicated SAT solver opens up
new venues for efficient implementations of first-order subsumption and related
rules. We show that, by using a flexible learning approach to choose between
various SAT encodings of subsumption variants, we greatly improve the
scalability of first-order theorem proving.  Our experimental results
demonstrate that, by using a tailored SAT solver within first-order reasoning,
we gain a large speedup in solving state-of-the-art benchmarks.

%% file: 1-introduction.tex
Most formal verification approaches use automated reasoners in their backend
to, for example, discharge verification conditions~\cite{DBLP:journals/software/Leino17,DBLP:journals/pacmpl/ClochardMP20,DBLP:conf/fmcad/GeorgiouGK20}, produce/block
counter-examples~\cite{SPACER16,IVY16,DBLP:conf/fmcad/AsadiBHFS20,DBLP:conf/cav/GarciaContrerasKSG23}, or enforce security and privacy
properties~\cite{DBLP:conf/fmcad/PickFG20,DBLP:conf/esop/MartinezADGHHNP19,DBLP:conf/sp/VeroneseFBTSM23,DBLP:conf/ccs/BruggerKKR023}. All these approaches crucially
depend on the efficiency of the underlying reasoning procedures,
ranging from SAT/SMT solving~\cite{Picosat,Z3,CVC5} to first-order
proving~\cite{Spass09,kovacs2013first,E19,zipperposition}.
\emph{In this paper, we focus on  effective extensions of first-order theorem proving with SAT-based reasoning, improving the state-of-the-art in  proving first-order (program)
properties.}

\paragraph{Saturation-Based Theorem Proving.} The leading algorithmic approach  in  first-order theorem proving is saturation~\cite{E19,kovacs2013first}.
While the concept of saturation is relatively unknown outside of the theorem
proving community, similar algorithms that are used in other areas, such as
Gr\"obner basis computation~\cite{Buchberger06a}, can be considered examples of
saturation algorithms. %
The key idea in saturation theorem proving is to reduce the problem of
proving the validity of a first-order formula~$A$ to the problem of establishing
unsatisfiability of~$\lnot A$ by using a sound inference system. 
That is, instead of proving~$A$, we refute~$\lnot A$, 
by selecting and applying inferences rules. 
In this paper, \emph{we focus on saturation algorithms
using the superposition calculus}, which is the 
most commonly
used  inference system for first-order logic with equality~\cite{Rubio01}.

\paragraph{Saturation with Redundancy.}
During saturation, the first-order prover keeps a set of
\emph{usable clauses} $C_1,\ldots,C_k$ with $k\geq 0$.
This is the set of clauses that the prover
considers as possible premises for inferences.
After applying an inference with one or more usable clauses
as premises, the consequence $C_{k+1}$ is added to the set of usable clauses.
The number of usable clauses is an important factor for the efficiency of proof search.
A naive saturation algorithm that keeps all derived clauses in the usable set
would not scale in practice.
One reason is that first-order formulas in general yield infinitely
many consequences. 
For example, consider the clause
\begin{equation}\label{ex:motivating}
  \lnot \positive(x) \lor \positive(\reverse(x)),
\end{equation}
where $x$ is a universally quantified variable ranging over the 
algebraic datatype {\tt list}, where list elements are integers;   $\positive$ is a unary predicate over {\tt
  list} such that $\positive(x)$ is valid iff all elements of $x$ are
nonnegative integers; and $\reverse$ is a unary function symbol
reversing a list. As such, clause~\eqref{ex:motivating}
asserts that the reverse of a list $x$ of
nonnegative integers is also a list of nonnegative integers (which
is clearly valid). Note that, when having clause~\eqref{ex:motivating} as a usable clause during proof search,
the clause $\lnot \positive(x) \lor \positive(\reverse^n(x))$ can be derived for any $n \geq
1$ from clause~\eqref{ex:motivating}. Adding $\lnot \positive(x) \lor \positive(\reverse^n(x))$ to the set of usable clauses
would, however, blow up the search space unnecessarily. This is
because $\lnot \positive(x) \lor \positive(\reverse^n(x))$  is a logical consequence of
clause~\eqref{ex:motivating}, and hence, if a formula $A$ can be proved
using $\lnot \positive(x) \lor \positive(\reverse^n(x))$, then $A$ is also provable using
clause~\eqref{ex:motivating}. Yet, storing $\lnot \positive(x) \lor \positive(\reverse^n(x))$ as usable formulas is highly inefficient as $n$ can be
arbitrarily large.

To avoid such and similar cases of unnecessarily increasing the set of
usable formulas during proof search,
first-order theorem provers implement the notion of
{\it redundancy}~\cite{Robinson65}, by extending the standard superposition
calculus with term/clause ordering and literal selection functions.
These orderings and selection functions are used to eliminate
so-called redundant clauses from the search space, where redundant
clauses are logical consequences of smaller clauses w.r.t. the considered
ordering.
In our example above, the clause
$\lnot \positive(x) \lor \positive(\reverse^n(x))$
would be a redundant clause as
it is a logical consequence of clause~\eqref{ex:motivating}, with
clause~\eqref{ex:motivating} being smaller (i.e. using fewer symbols)
than $\lnot \positive(x) \lor \positive(\reverse^n(x))$.
As such,
if clause~\eqref{ex:motivating} is already a usable
clause, saturation
algorithms implementing redundancy should ideally not store $\lnot
\positive(x) \lor \positive(\reverse^n(x))$ as usable clauses.
To detect and reason about redundant clauses, saturation
algorithms with redundancy extend the superposition inference system
with so-called {\it simplification rules}.
Simplification rules do not add new formulas to the set of (usable)
clauses in the search space,
but instead simplify and/or delete redundant formulas from the search
space,
without destroying the refutational completeness of
superposition:
if a formula $A$ is valid, then $\neg A$ can be refuted
using the superposition
calculus extended with simplification rules.
In our example above, this means that if $\neg A$ can be refuted using
$\lnot \positive(x) \lor \positive(\reverse^n(x))$, then $\neg A$ can be refuted in
the superposition
calculus extended with simplification rules, without using
$\lnot \positive(x) \lor \positive(\reverse^n(x))$ but using
clause~\eqref{ex:motivating} instead.

Ensuring that simplification rules are applied efficiently for eliminating
redundant clauses is, however, not trivial.
In this paper, we show that
\textit{%
    SAT-based approaches can effectively 
    identify the application of simplification rules during saturation},
    improving thus the efficiency of saturation algorithms with redundancy.

\paragraph{Subsumption for Effective Saturation.} While redundancy is
a powerful criterion for keeping the set of clauses used in proof
search as small as possible, establishing whether an arbitrary first-order
formula is redundant is as hard as proving whether it is valid.
For example, in order to derive that $\lnot \positive(x) \lor
\positive(\reverse^n(x))$ is redundant in our example above, the prover
should establish (among other conditions) that it is a logical consequence
of~\eqref{ex:motivating}, which essentially requires proving based on
superposition.
To reduce the burden of proving redundancy,
first-order provers implement  sufficient
conditions towards deriving redundancy, so that these conditions can
be efficiently checked (ideally using only syntactic arguments and no
proofs). One
such condition comes with the notion of {\it subsumption}, yielding
one of the most impactful simplification rules in superposition-based
theorem proving~\cite{BG94} and SAT solving~\cite{DBLP:conf/sat/Biere04}.

The intuition behind first-order subsumption is that a (potentially large) instance of a clause~$C$
does not convey any additional information over~$C$,
and thus it should be avoided to have both $C$ and its instance in
the set of usable clauses; to this end, we say that the instance of
$C$ is subsumed by $C$.
More formally, a clause $C$ subsumes another clause $D$ if there is a substitution $\sigma$
such that $\sigma(C)$ is a submultiset of $D$\footnote{we consider a clause $C$ as a multiset of its literals}.
In such a case, subsumption removes the subsumed clause $D$ from the clause set.
To continue our example above, a unit clause $\positive(\reverse^m(x))$, with $m\geq 1$,
would prevent us from deriving $\lnot \positive(x) \lor \positive(\reverse^n(x))$ for any $n \geq m$,
and hence eliminate an infinite branch of clause derivations from the search space.

To detect possible inferences of subsumption and related rules,
state-of-the-art provers use a two-step approach~\cite{handbook-indexing}:
(i)~retrieve a small set of candidate clauses, using literal filtering
methods,
and then
(ii)~check whether any of the candidate clauses represents an actual instance of the rule.
Step~(i) has been well researched over the years, leading to highly efficient
indexing solutions~\cite{voronkov-evaluation-indexing,schulz2013simple,handbook-indexing}.
Interestingly, step~(ii) has not received much attention,
even though it is known that checking subsumption relations between multi-literal clauses
is an NP-complete problem~\cite{matching-np-complete}.
Although indexing in step~(i) allows the first-order prover to skip
step (ii) in many cases,
the application of~(ii) in the remaining cases may remain  problematic
(due to  NP-hardness).
For example, while profiling subsumption in the world-leading
theorem prover \vampire{}~\cite{kovacs2013first}, we observed
subsumption applications, and in particular calls to the literal-matching
algorithm of step~(ii),  that consume more than 20 seconds of running
time.
Given that millions of such matchings are performed during a typical
first-order proof attempt,
we consider such cases highly inefficient, calling for improved
solutions towards step~(ii). In this paper we address this demand and
show that a {\it tailored SAT-based encoding can significantly improve the
literal matching, and thus subsumption}, in first-order theorem
proving.  We also advocate the flexibility of SAT solving for variants of subsumption, in particular when \emph{combining subsumption with resolution.} 

\paragraph{Our Contributions.}
    We bring the
    following main contributions.
    \smallskip

\noindent (1)
    We propose a {\it generic SAT-based encoding for capturing potential 
    applications of both subsumption and subsumption resolution} in first-order theorem proving (Sections~\ref{sec:defconstraints}-\ref{sec:satconstraints}).
    A solution to our SAT-based encoding gives a concrete application
    of subsumption and/or subsumption resolution, allowing the first-order prover to apply that
    instance of subsumption (resolution) as a simplification rule during
    saturation. Moreover, our encoding is complete in the sense that any instance of subsumption (resolution) is  a model of our SAT problem (Theorems~\ref{thm:subsumption-encoding-soundness-and-completeness}, \ref{thm:direct-encoding-soundness-and-completeness} and \ref{thm:indirect-encoding-soundness-and-completeness}). 

    \medskip

\noindent (2)
    We tailor encodings of subsumption and subsumption resolution for effective SAT-based redundancy elimination (Section~\ref{sec:implementation}). Importantly, we adjust unit propagation and
    conflict resolution in SAT solving towards efficient handling of subsumption and subsumption resolution (Section~\ref{sec:sat-solver}). Our resulting SAT-based redundancy approaches are directly integrated in saturation (Section~\ref{sec:loop-optimization}), without changing the underlying design of efficient saturation. 
    \medskip

\noindent (3)
    We establish a flexible learning approach to choose between encodings with different properties. We detail how to train decision trees to obtain the best complexity-efficiency trade-off in choosing encodings for subsumption resolution (Section~\ref{sec:heuristics}). 
    As part of an empirical study, we analyse the utility of solving subsumption and subsumption resolution problems for a large portion of our computation budget. We introduce a method to choose an appropriate cutoff threshold and stop the SAT search prematurely. We empirically show that solely solving simple instances of subsumption and subsumption resolution is not a good solution, even with an educated timeout strategy.

\medskip

    \noindent (4)     We implemented our SAT-based redundancy approach as a new SAT
    solver  in the
    \vampire{} theorem prover. 
    We empirically evaluate our approach on the standard benchmark library TPTP
    (Section~\ref{sec:experiments}).
    Our experiments demonstrate that using SAT solving for deciding and applying
    subsumption and subsumption resolution brings clear improvements in the saturation process
    of first-order proving,
    for example, improving the (time) performance of the prover
    by a factor of $1.36$ when both subsumption and subsumption resolution are enabled.

\medskip

\paragraph{Extension of Previous Works.}
This paper is an extended version of the conference papers
``First-Order Subsumption via SAT Solving''~\cite{rath2022first}
and
``SAT-Based Subsumption Resolution''~\cite{DBLP:conf/cade/CoutelierKRR23}
published at FMCAD 2022 and CADE 2023, respectively.

In Section~\ref{sec:satconstraints}, we extend the SAT-based subsumption framework  of~\cite{rath2022first} to  subsumption resolution and complemented~\cite{DBLP:conf/cade/CoutelierKRR23} with unifying support for both subsumption and subsumption resolution.
In Section~\ref{sec:implementation} we extend the SAT solving algorithms of~\cite{rath2022first,DBLP:conf/cade/CoutelierKRR23} to solve both subsumption and subsumption resolution.
As such, Sections~\ref{sec:defconstraints}-\ref{sec:implementation} unify the approaches of~\cite{rath2022first,DBLP:conf/cade/CoutelierKRR23} into a flexible technique for SAT-based redundancy checking in saturation. Our paper therefore adjusts  the texts of~\cite{rath2022first,DBLP:conf/cade/CoutelierKRR23} and extends their results with formal theoretical arguments and proofs. 

In addition, this paper brings the following new contributions when compared
to our papers~\cite{rath2022first,DBLP:conf/cade/CoutelierKRR23}. 
First, we introduce a symbolic approach to combine SAT-based encodings with learning heuristics to dynamically select the most promising encoding during run-time (Section~\ref{sec:choosing-encoding}). 
Second,  we expand preprocessing via pruning techniques and use our SAT solver only on  hard(er) problems (Section~\ref{sec:pruning}). Here, we provide a faster multilayered filter to detect unsatisfiable instances before they even reach the SAT engine.
Third, we bring in an empirically motivated approach to trade completeness of SAT-based subsumption (resolution) for computation time, by cutting off early the harder instances of subsumption and subsumption resolution (Section~\ref{sec:cutoff}).
We show that subsuming simple clauses is not enough in practice, strengthening our argument for more scalable algorithms in the context of redundancy elimination.

%% file: 2-preliminaries.tex
We consider standard multi-sorted first-order logic, where we support all standard
Boolean connectives  $\land$, $\lor$, $\limpl$, $\lnot$  and quantifiers $\forall$ and $\exists$. 
We assume that
the language contains the logical constants $\top$ and $\bot$ for always true and always false
formulas, respectively.
Let~$\vars$ denote the
set of \emph{first-order} variables.  For the purpose of our work, we also use~$\bvars$ to denote a
set of \emph{Boolean} variables, where  Boolean variables (constants, atoms) are written as~$b$.
Throughout the paper,
we write $x$, $y$, $z$ for \emph{first-order} variables; 
$c$, $d$ for constants; 
$f$, $g$ for function symbols; 
and
$p$, $q$ for predicates.
The set of first-order terms~$\terms$ consists of variables, constants,
and function symbols applied to other terms; we denote terms by $t$.
First-order \emph{atoms}, or simply just atoms, are predicates applied to terms.
Atoms and negated atoms are also called first-order \emph{literals},
and denoted by $\ell$, $\sideL$, $\mainL$.
First-order \emph{clauses}, or simply just clauses, are disjunctions of literals,
denoted by $C$, $D$, $\sideP$, $\mainP$. For convenience, the literals of a clause will often be written with subscripted lower case letters, e.g., $\sideP = \sideL_1 \lor \sideL_2 \lor \dots \lor \sideL_k$.
For simplicity, the notation used
throughout this paper may
possibly use indices.

A clause that consists of a single literal is called a \emph{unit clause}.
Clauses are often viewed as multisets of literals; that is, a clause
$\sideP = \sideL_1\vee \sideL_2\vee \ldots\vee \sideL_k$ is considered to be the multiset
$\{\sideL_1, \sideL_2, \ldots, \sideL_k\}$.

An expression $E$ is a term, literal, or clause.
We denote the set of variables occurring in the expression~$E$ by~$\vars(E)$.
A \emph{substitution} is a partial function~$\sigma\colon\vars\to\terms$;
we occasionally write it as a set of mappings
$\sigma = \{ x_1 \mapsto t_1, \dots, x_n \mapsto t_n \}$.
The function $\sigma$ is extended to arbitrary expressions~$E$
by simultaneously replacing each variable~$x$ in~$E$ by~$\sigma(x)$,
for all variables~$x$ on which $\sigma$ is defined.
We say an expression~$E_1$ can be matched to an expression~$E_2$
if there exists a substitution~$\sigma$ such that $\sigma(E_1) = E_2$.
Additionally, we make the distinction between positive and negative polarity matches.
A positive polarity match $\sigma$ matches two literals $\sideL$ and $\mainL$ such that $\sigma(\sideL) = \mainL$,
whereas a negative (or opposite) polarity match would complement one of the literals (i.e., $\sigma(\sideL) = \neg \mainL$).

\paragraph{Saturation and  Subsumption.}
Most first-order theorem provers, see
e.g.~\cite{Spass09,kovacs2013first,E19}, implement saturation with redundancy, 
using the superposition calculus~\cite{BG94}.
A clause~$S$ subsumes a clause~$M$ iff there exists a substitution $\sigma$
such that $\sigma(S) \sqsubseteq M$, where~$S$ and~$M$ are treated as multisets of literals and $\sqsubseteq$ is the multiset inclusion operator.
\emph{Subsumption} is a simplification rule  that deletes subsumed
clauses from the search space during saturation.  Subsumption gives a
powerful basis for other simplification rules. For example, subsumption
resolution~\cite{kovacs2013first,E19}, also known as contextual
literal cutting  or self-subsuming resolution, is the combination of
subsumption with binary resolution. On the other hand, {subsumption demodulation}~\cite{gleiss-sd} results from combining subsumption with demodulation/rewriting.

\paragraph{SAT Solving.}
Modern SAT solvers, see e.g.~\cite{DBLP:conf/sat/EenS03,DBLP:conf/sat/BiereFW23,DBLP:journals/fmsd/FleuryB22}, are based on conflict-driven clause learning (CDCL)~\cite{Marques-SilvaLynceMalik:2021},
with the core procedures to \emph{decide}, \emph{unit-propagate}, and \emph{resolve-conflict}.
The SAT solver maintains a partial assignment of truth values to the Boolean variables.
Unit propagation (also called Boolean constraint propagation), that is
\emph{unit-propagate} in a SAT solver, 
propagates clauses w.r.t.\ the partial assignment.
If exactly one literal~$\ell$ in a clause remains unassigned in the current assignment
while all other literals are false,
the solver sets~$l$ to true to avoid a conflict.
The two-watched-literal scheme~\cite{MoskewiczMadiganZhaoZhangMalik:2001:Chaff}
is the standard approach for efficient implementation of unit propagation.

If no propagation is possible, the solver may choose a currently unassigned variable~$b$
and set it to true or false; hence, \emph{decide} in SAT solving.
The number of variables in the current assignment that have been assigned by decision
is called the \emph{decision level}.

If all literals in a clause are false in the current assignment,
the solver enters conflict resolution, via the \emph{resolve-conflict}
block of SAT solving.
If the current decision level is zero, the conflict follows unconditionally from the input clauses
and the solver returns ``unsatisfiable'' (UNSAT).
Otherwise, by analysing how the literals in the conflicting clause have been assigned,
the SAT solver may derive and learn a conflict lemma, undo some decisions, and continue solving.

%% file: 3-theory.tex
In this section we formally define subsumption and subsumption resolution. These concepts yield important deletion/simplification rules during saturation. 

\begin{definition}[Subsumption\label{def:subsumption}]
   A clause~$\sideP$ \emph{subsumes} a  clause~$\mainP$
    iff there exists a substitution $\sigma$ such that
    \begin{equation}%
        \label{eq:subs}
        \sigma(\sideP) \sqsubseteq \mainP
        ,
    \end{equation}
    where~$\sqsubseteq$ denotes multiset inclusion.
    We call~$\sideP$ the \emph{side premise} of subsumption,
    and~$\mainP$ the \emph{main premise} of subsumption.
\end{definition}

Subsumed clauses are redundant~\cite{BG94}
and can thus be deleted from the search space
without compromising the completeness of the saturation algorithm.
Removing subsumed clauses $\mainP$ from the search space $F$ is implemented through a simplifying rule, checking condition~\eqref{eq:subs} over pairs of clauses $(\sideP,\mainP)$ from $F$.
To check condition~\eqref{eq:subs} for a clause pair $(\sideP,\mainP)$,
every literal in $\sideP$ is matched to some literal in $\mainP$;
if a compatible set of matches is found and
no literal in $\mainP$ is matched more than once,
then $\mainP$ can be removed from $F$.

\begin{example}\label{ex:subsumption}
    Consider the clause
    $\mainP \coloneqq p(g(c,d)) \lor \lnot p(f(d)) \lor \lnot q(y_1)$.
    \begin{itemize}
    \item
        $\sideP_1 \coloneqq p(g(x_1,x_2)) \lor \lnot q(x_3)$
        subsumes $\mainP$,
        as witnessed by the substitution $\sigma = \{ x_1 \mapsto c, x_2 \mapsto d, x_3 \mapsto y_1 \}$.
    \item
        $\sideP_2 \coloneqq p(g(x_1,x_2)) \lor \lnot q(x_1)$, does \emph{not} subsume $\mainP$.
        This is because the first literal of $S_2$ imposes $x_1 \mapsto c$, 
        while the second literal requires $x_1 \mapsto y_1$ in order to have $M$ subsumed. 
        Note that the substitution is applied only to the side premise; we do \emph{not} unify the clauses.
    \item
        $\sideP_3 \coloneqq p(g(x_1,d)) \lor p(g(c,x_2)) \lor \lnot q(x_3)$
        does not subsume $\mainP$,
        because only set inclusion can be satisfied,
        rather than {multi-set} inclusion.
    \end{itemize}
\end{example}

When subsumption~\eqref{eq:subs} for a clause pair $(S,M)$ fails, it might still be possible to simplify the clause $M$ by deleting one of its literals.
Subsumption resolution, referred to as \SR{} in the sequel, aims exactly to remove one redundant literal from a clause and is defined below.

\begin{definition}[Subsumption Resolution\label{def:subsumptionRes}]
    Clauses $\sideP$ and $\mainP$ are  the 
    \emph{side premise} and \emph{main premise} of \emph{subsumption resolution \SR}, respectively,
    iff there is a substitution $\sigma$,
    a set of literals $\sideP'\subseteq \sideP$, and a literal $\mainL'\in \mainP$ such that
    
    \begin{equation}%
        \label{eq:subsumption-resolution-definition}
        \sigma(\sideP') = \{\neg \mainL'\}
        \quad\text{and}\quad
        \sigma(\sideP\setminus \sideP') \subseteq \mainP\setminus \{\mainL'\}
        ,
    \end{equation}
   implying that $\mainP$ can be replaced by $\mainP\setminus \{\mainL'\}$.
    Subsumption resolution \SR{} is hence the rule
    
    \begin{prooftree}
        \AxiomC{$\sideP$}
        \AxiomC{$\cancel{\mainP}$}
        \RightLabel{(\SR).}
        \BinaryInfC{$\mainP \setminus \{\mainL'\}$}
    \end{prooftree}
\end{definition}

We indicate the deletion of a clause $\mainP$
by drawing a line through it, that is ($\cancel{\mainP}$). We refer to the literal $\mainL'$ of $\mainP$ as the \emph{resolution literal} of \SR.
Intuitively, subsumption resolution is binary resolution followed by subsumption of one of its premises by the conclusion.
However, by combining two inferences into one it can be treated as a simplifying inference,
which is advantageous from the perspective of efficient proof search.

\begin{example}\label{ex:subsumption-resolution}
    Consider  clause
    $\mainP \coloneqq p(g(c,d)) \lor \lnot p(f(d)) \lor \lnot q(y_1)$
    from Example~\ref{ex:subsumption}.
    \begin{itemize}
    \item
        $\sideP_4 \coloneqq \lnot p(g(x_1,x_2)) \lor \lnot q(x_3)$
        allows subsumption resolution with main premise~$\mainP$
        using the substitution~$\sigma = \{ x_1 \mapsto c, x_2 \mapsto d, x_3 \mapsto y_1 \}$. 
        Under this substitution, we have
        $\sigma(\sideP_4) = \lnot p(g(c,d)) \lor \lnot q(y_1)$. We resolve $\sigma(\sideP_4)$ and $\mainP$
        to obtain the conclusion $\lnot p(f(d)) \lor \lnot q(y_1)$, which subsumes $\mainP$. We thus have
        \begin{prooftree}
            \AxiomC{$\lnot p(g(x_1,x_2)) \lor \lnot q(x_3)$}
            \AxiomC{$\hcancel{p(g(c,d)) \lor \lnot p(f(d)) \lor \lnot q(y_1)}$}
            \RightLabel{(\SR).}
            \BinaryInfC{$\lnot p(f(d)) \lor \lnot q(y_1)$}
        \end{prooftree}
    \item
        $\sideP_5 \coloneqq \lnot p(g(x_1,d)) \lor \lnot p(g(c,x_2)) \lor \lnot q(x_3)$
        allows subsumption resolution with $\mainP$ with the same substitution $\sigma$
        and conclusion as used for $\sideP_4$.
        In contrast to $\sideP_4$, two literals of $\sideP_5$ are mapped to the resolution literal.
    \item
        $S_6 \coloneqq p(f(x_1)) \lor q(x_2)$
        does not allow subsumption resolution with $\mainP$,
        because at most one opposite polarity match is permitted.
    \item
        $S_7 \coloneqq p(g(c,x_1)) \lor p(f(x_1)) \lor \lnot p(f(x_2))$
        does not allow subsumption resolution with $\mainP$.
        While we can find a candidate resolution literal
        by matching $p(f(x_1))$ to $\lnot p(f(d))$,
        there is no possible match for $\lnot p(f(x_2))$
        since same-polarity matches to the resolution literal are not permitted.
    \item
        $S_8 \coloneqq p(g(c,x_1)) \lor p(f(x_1)) \lor r(x_2)$
        does not allow subsumption resolution with $\mainP$,
        because there is no possible match for $r(x_2)$.
    \end{itemize}
\end{example}

We note that subsumption and subsumption resolution are NP-complete problems~\cite{handbook-indexing,matching-np-complete}.
In this paper, we advocate the use of state-of-the-art SAT solving and provide tailored SAT encodings for subsumption and subsumption resolution, as follows.
In Section~\ref{sec:defconstraints}, we express subsumption and subsumption resolution through constraints, allowing us to encode subsumption (resolution) as a SAT problem in Section~\ref{sec:satconstraints}.

%% file: 4-constraints.tex
Throughout the remainder of the paper,
we assume that clauses do not have duplicate literals
and do not contain both a literal and its negation,
as expressed by Assumption~\ref{ass:no-duplicates} below.
Only substitutions may collapse several literals into one,
as illustrated in Example~\ref{ex:subsumption}.

\begin{assumption}[No Duplicates\label{ass:no-duplicates}]
    We assume that a clause $C = \ell_1 \lor \ell_2 \lor \dots \lor \ell_k$ does not have duplicate atoms.
    That is, $C$ does not contain duplicate literals, nor a literal and its negation.
    \begin{equation}
        \textbf{no duplicates}\quad \text{for any~} C = \ell_1 \lor \ell_2  \lor \dots \lor \ell_k:~
        \forall i \, i'.\ \bigl( i \neq i' \Rightarrow \ell_i \neq \ell_{i'} \land \ell_i \neq \neg \ell_{i'} \bigr)
        \label{eq:no-duplicates}
    \end{equation}
\end{assumption}

We first show that the application of subsumption (Theorem~\ref{thm:subsumption-constraints})
and subsumption resolution (Theorem~\ref{thm:subsumption-resolution-constraints})
can precisely be captured by substitution constraints, as follows.

\begin{theorem}[Subsumption Constraints]%
  \label{thm:subsumption-constraints}%
  Consider two clauses $\sideP = \sideL_1 \lor \sideL_2 \lor \dots \lor \sideL_k$
  and $\mainP = \mainL_1 \lor \mainL_2 \lor \dots \lor \mainL_n$,
  where~$\mainP$ does not contain duplicate literals.
  $\sideP$ subsumes~$\mainP$
  iff there exists a substitution~$\sigma$ that satisfies the following two properties:
  \begin{align}
    &\textbf{partial completeness} &
    \forall i \ldotp \exists j\ldotp \sigma(\sideL_i) = \mainL_j
    \label{eq:subsumption-completeness}\\
    &\textbf{multiplicity conservation} &
    \forall i \, i' \, j \ldotp \bigl(i \neq i' \land \sigma(\sideL_i) = \mainL_j \Rightarrow \sigma(\sideL_{i'}) \neq \mainL_{j} \bigr)
    \label{eq:subsumption-multiplicity}
  \end{align}
\end{theorem}

\begin{proof}
    Because~$\mainP$ does not contain duplicate literals,
    the subsumption condition~$\sigma(\sideP) \sqsubseteq \mainP$
    amounts to the two statements
    (i) each element of~$\sigma(\sideP)$ is an element of~$\mainP$ and
    (ii) the multiplicity of elements in~$\sigma(\sideP)$ is at most one,
    i.e., there are no duplicates in~$\sigma(\sideP)$.

    Statement (i) is equivalent to \textbf{partial completeness}~\eqref{eq:subsumption-completeness}.

    Given~\eqref{eq:subsumption-completeness},
    \textbf{multiplicity conservation}~\eqref{eq:subsumption-multiplicity}
    can be rewritten into
    \[
        \forall i \, i' \ldotp \bigl(i \neq i' \Rightarrow \sigma(\sideL_i) \neq \sigma(\sideL_{i'}) \bigr)
        ,
    \]
    which is equivalent to statement (ii).
\end{proof}

Note that the \textbf{partial completeness} property~\eqref{eq:subsumption-completeness}
ensures that all literals $\sigma(\sideL_i)$ have a literal $\mainL_j$ to which they match.
\textbf{Partial completeness}~\eqref{eq:subsumption-completeness} alone would,
however, encode a simple subset inclusion.
The \textbf{multiplicity conservation} constraint~\eqref{eq:subsumption-multiplicity}
ensures the preservation of the cardinality of the multi-set.
In fact, due to Assumption~\ref{ass:no-duplicates}, $\mainP$ is a simple set,
and
\textbf{multiplicity conservation}~\eqref{eq:subsumption-multiplicity}
prevents the substition~$\sigma$ from
collapsing several literals into one.
As a result of Theorem~\ref{thm:subsumption-constraints}, only one literal in~$\sideP$
can be matched to any literal of~$\mainP$.

Similarly to Theorem~\ref{thm:subsumption-constraints},
we show that subsumption resolution can be formalised through four constraints, as follows. 

\begin{theorem}[Subsumption Resolution Constraints]
  The clauses~$\mainP = \mainL_1 \lor \mainL_2 \lor \dots \lor\mainL_n$ and~$\sideP = \sideL_1 \lor \sideL_2 \lor \dots\lor\sideL_k$ are respectively the main and side premises of an instance of the subsumption resolution rule \SR{} 
  iff there exists a substitution $\sigma$ that satisfies the following four properties:
  \begin{align}
    &\textbf{existence} &
    \exists i\,j\ldotp\sigma(\sideL_i) = \neg \mainL_j \label{eq:subsumption-resolution-existence} \\
    &\textbf{uniqueness} &
    \exists j'\ldotp \forall i\,j\ldotp \bigl(\sigma(\sideL_{i}) = \neg \mainL_{j} \Rightarrow j = j' \bigr)
    \label{eq:subsumption-resolution-uniqueness} \\
    &\textbf{completeness} &
    \forall i \ldotp \exists j\ldotp \bigl( \sigma(\sideL_i) = \neg \mainL_j \lor \sigma(\sideL_i) = \mainL_j \bigr)
    \label{eq:subsumption-resolution-completeness}\\
    &\textbf{coherence} &
    \forall j \ldotp \bigl(\exists i \ldotp \sigma(\sideL_i) = \mainL_j\Rightarrow \forall i\ldotp \sigma(\sideL_i) \neq \neg \mainL_j\bigr)
    \label{eq:subsumption-resolution-coherence}
  \end{align}
  \label{thm:subsumption-resolution-constraints}
\end{theorem}

\begin{proof}
    It is easy to see that the constraints \eqref{eq:subsumption-resolution-existence}-\eqref{eq:subsumption-resolution-coherence} hold whenever subsumption resolution applies.
    For the other direction, we assume that the four constraints \eqref{eq:subsumption-resolution-existence}-\eqref{eq:subsumption-resolution-coherence} hold, and prove that subsumption resolution applies on $(S,M)$. Let $\sideP, \mainP$ and $\sigma$ such that the four constraints hold. \textbf{Existence}~\eqref{eq:subsumption-resolution-existence} implies that there exists at least one literal $\mainL'\in \mainP$ and a non-empty set $\sideP' \subseteq \sideP$ such that $\neg \mainL' \in \sigma(\sideP')$. \textbf{Uniqueness}~\eqref{eq:subsumption-resolution-uniqueness} asserts that $\mainL'$ is unique, and therefore $\{\neg \mainL'\} = \sigma(\sideP')$. We can now divide the literals of $\sideP$ into two groups: $\sideP'$ and $\sideP^*$ such that $\sigma(\sideP') = \{\neg \mainL'\}$ and $\sideP^* = \sideP \setminus \sideP'$. \textbf{Coherence}~\eqref{eq:subsumption-resolution-coherence} ensures that $\mainL'\notin \sigma(\sideP^*)$.
    From \textbf{completeness}~\eqref{eq:subsumption-resolution-completeness}, we derive  $\sigma(\sideP^*) \subseteq \mainP$. Furthermore,  $\mainL'$ is unique and  $\mainL'\notin \sigma(\sideP^*)$. Therefore, $\sigma(\sideP^*) \subseteq \mainP \setminus \{\mainL'\}$. Putting everything together, we obtain $\sigma(\sideP') = \{\neg \mainL'\} \land \sigma(\sideP^*) \subseteq \mainP \setminus \{\mainL'\}$; hence \SR{} over $(S,M)$ applies.
\end{proof}

%% file: 4N-encodings.tex
Based on the subsumption constraints of Theorems~\ref{thm:subsumption-constraints}
and~\ref{thm:subsumption-resolution-constraints},
we provide tailored SAT encodings for subsumption and subsumption resolution,
allowing us to devise custom SAT solving algorithms in Section~\ref{sec:sat-solver}
and integrate them in saturation Section~\ref{sec:loop-optimization}.
In what follows, we fix two arbitrary clauses $\sideP = \sideL_1 \lor \sideL_2 \lor \dots \lor \sideL_k$
and $\mainP = \mainL_1 \lor \mainL_2 \lor \dots\lor \mainL_n$,
and give all definitions relative to $(\sideP,\mainP)$.
Intuitively, the constraints defined in this section
encode the existence of a substitution~$\sigma$
which witnesses subsumption or subsumption resolution.

\paragraph{Variables and substitutions.}
Given the side premise $\sideP$ and main premise $\mainP$ of subsumption or subsumption resolution, we introduce two Boolean variables $b^+_{i,j}$ and $b^-_{i,j}$ for each literal pair $(\sideL_i, \mainL_j)$,  as follows:

\begin{align}
    b^+_{i,j} \Leftrightarrow \sigma(\sideL_i) = \mainL_j
    \label{eq:positive-match-def}\\
    b^-_{i,j} \Leftrightarrow \sigma(\sideL_i) = \neg \mainL_j
    \label{eq:negative-match-def}
\end{align}

We also define a set of substitutions $\Sigma^+_{i,j}$ and $\Sigma^-_{i,j}$, called \emph{substitution constraints},
such that $\Sigma^+_{i,j}(\sideL_i) = \mainL_j$, and $\Sigma^-_{i,j}(\sideL_i) = \neg \mainL_j$.
In the following, we write $\Sigma^\pm_{i,j}$ to refer to the substitution constraints of $\Sigma^+_{i,j}$ or $\Sigma^-_{i,j}$; when no such substitution exists, we write $\IncSubst$.
For example, let $\sideL_{1} = p(x, y)$ and $\mainL_{1} = \neg p(f(c), d)$. The two variables $b^+_{1,1}$, $b^-_{1,1}$ are associated to the pair $(\sideL_{1}, \mainL_{1})$ and the substitutions $\Sigma^+_{1,1} = \IncSubst$, $\Sigma^-_{1,1} = \{x\mapsto f(c), y\mapsto d\}$.

\begin{definition}[Match set]
    We define a \emph{match set $\MS(\sideP, \mainP)$} associated to clauses $\sideP$ and $\mainP$ to contain a set of Boolean variables and positive/negative polarity matches for each literal pair $(\sideL_i, \mainL_j)$ of $(\sideP,\mainP)$. That is, 
    \begin{equation}
        \MS(\sideP, \mainP) = \Bigl\{
            \left(b^+_{i,j}, \Sigma^+_{i,j}\right), \left(b^-_{i,j}, \Sigma^-_{i,j}\right)
            \Bigm| \sideL_i\in\sideP \land \mainL_j\in\mainP 
            \land \Sigma^+_{i,j}(\sideL_i) = \mainL_j
            \land \Sigma^-_{i,j}(\sideL_i) = \neg \mainL_j
        \Bigr\}
    \end{equation}
\end{definition}

\paragraph{Compatibility constraints.} 
Detecting the application of subsumption and/or subsumption resolution requires finding a substitution~$\sigma$ such that the subsumption constraints of Theorems~\ref{thm:subsumption-constraints}-\ref{thm:subsumption-resolution-constraints} are satisfied. We achieve this by imposing that the substitution constraints $\Sigma^\pm_{i,j} \subseteq \sigma$  are true iff $\Sigma^\pm_{i,j}$ are \emph{compatible} with a global substitution $\sigma$, in the following sense. 

\begin{definition}[Substitution Compatibility\label{def:Substitution:Compatibility}]
    A substitution $\Sigma$ is \emph{compatible} with another substitution $\Sigma'$ if they do not map the same variable to different terms. Formally, $\Sigma$ is compatible with $\Sigma'$ iff
    \begin{equation}
        \forall x. (\Sigma(x)=t \land \Sigma'(x)=t' \land t\neq x \land t'\neq x) \implies t = t'.
    \end{equation}
\end{definition}
\newpage
The compatibility of the substitution constraints  $\Sigma^\pm_{i,j} \subseteq \sigma$ with  $\sigma$ is encoded using the Boolean variables $b^\pm_{i,j}$,  as follows:
  
\begin{align}
  &\textbf{positive compatibility  } &
  \bigwedge_{i}\bigwedge_{j} \left(b_{i,j}^+ \Rightarrow \Sigma^+_{i,j} \subseteq \sigma \right)
  \label{eq:positive-compatibility}\\
  &\textbf{negative compatibility  } & 
  \bigwedge_{i}\bigwedge_{j} \left(b_{i,j}^- \Rightarrow \Sigma^-_{i,j} \subseteq \sigma \right)
  \label{eq:negative-compatibility}
\end{align}

 Note that   $\Sigma^+_{i,j}$ is a substitution constraint between $\sideL_i$ and $\mainL_j$. Further,   $\Sigma^+_{i,j} \subseteq \sigma \Leftrightarrow \sigma(\sideL_i) = \mainL_j$. Using $\Sigma^+_{i,j}$ together with~\eqref{eq:positive-compatibility}, we derive $b^+_{i,j} \Rightarrow \sigma(\sideL_i) = \mainL_j$. A similar result is obtained for compatibility of $\Sigma^-_{i,j}$. 

\subsection{SAT Encoding of Subsumption}
Note that Definition~\ref{def:subsumption} and Theorem~\ref{thm:subsumption-constraints} imply  that subsumption is restricted to only positive matches between literals of $S, M$. As such,  $b^-_{i,j}$ need not to be considered for subsumption. 

Using~\eqref{eq:positive-match-def}-\eqref{eq:negative-match-def}, we rewrite the subsumption constraints of Theorem~\ref{thm:subsumption-constraints} 
by replacing substitution constraints with the Boolean variable $b^+_{i,j}$, yielding:

\begin{align}
  &\textbf{SAT-based partial completeness  } &
  \bigwedge_{i} \bigvee_{j} b_{i,j}^+
  \label{eq:subsumption-completeness-sat}\\
  &\textbf{SAT-based multiplicity conservation  } &
  \bigwedge_j \AMO(\{b_{i,j}^+ \mid i = 1,...,k\})
  \label{eq:subsumption-multiplicity-sat}
\end{align}
where $\AMO(\{b_{i,j}^+ \mid i = 1,...,k\})$ is 
 an at-most-one constraint ensuring that  at most one variable $b_{i,j}^+$ is true at the same time.

\begin{theorem}\label{thm:subsumption-backward-compatibility}
    Assume that clause $\mainP$  does not have duplicate atoms, as in~\eqref{eq:no-duplicates}. Let  $\MS(\sideP, \mainP) = \left\{\left(b^\pm_{i,j}, \Sigma^\pm_{i,j}\right)\right\}$ be  the match set of $\sideP$ and $\mainP$.
   \textbf{Positive compatibility}~\eqref{eq:positive-compatibility} and \textbf{SAT-based partial completeness}~\eqref{eq:subsumption-completeness-sat} imply $\Sigma^+_{i,j}\subseteq \sigma  \Rightarrow  b^+_{i,j}$.
\end{theorem}
\begin{proof}
    Towards a contradiction, assume there exist ~$i, j$ such that
    $\Sigma^+_{i, j} \subseteq \sigma$ and $b^+_{i, j} = \bot$.
    Condition~\eqref{eq:subsumption-completeness-sat} implies that 
    there exists $j'$ such that $b^+_{i, j'} = \top$; then, by constraint~\eqref{eq:positive-compatibility} we have  $\Sigma^+_{i, j'} \subseteq \sigma$,
    that is, $\Sigma^+_{i, j'}(\sideL_{i}) = \mainL_{j'}$.
    Since both $\Sigma^+_{i, j}$ and $\Sigma^+_{i, j'}$ impose a mapping on the same literal $\sideL_{i}$, the mappings are on the same variables.
    Therefore, for $\Sigma^+_{i, j}$ and $\Sigma^+_{i, j'}$ to be compatible with $\sigma$ simultaneously, they must be identical. 
    Hence, $\mainL_{j} = \Sigma^+_{i, j}(\sideL_{i}) = \Sigma^+_{i, j'}(\sideL_{i})  = \mainL_{j'}$,
    which contradicts Assumption~\ref{ass:no-duplicates}.
\end{proof}

We have now all the ingredients to introduce our SAT-based encoding of subsumption. 
 
\begin{definition}[SAT-Based Subsumption Encoding]
    The \emph{SAT-based subsumption encoding} $\SEnc(\sideP, \mainP)$ of the clauses $\sideP$ and $\mainP$ is the conjunction of \textbf{positive compatibility}~\eqref{eq:positive-compatibility}, \textbf{SAT-based partial completeness}~\eqref{eq:subsumption-completeness-sat}, and \textbf{SAT-based multiplicity conservation}~\eqref{eq:subsumption-multiplicity-sat}. 
\end{definition}

As a consequence of Theorem~\ref{thm:subsumption-backward-compatibility}, we obtain the following corollary.

\begin{corollary}\label{cor:subsumption-backward-compatibility}
    A model of the subsumption encoding $\SEnc(\sideP, \mainP)$
    satisfies $\forall i, j.\ b^+_{i,j} \Leftrightarrow \sigma(\sideL_i) = \mainL_j$.
\end{corollary}

Corollary~\ref{cor:subsumption-backward-compatibility}
ensures that $\forall i, j.\  b^+_{i,j} \Leftrightarrow \Sigma^+_{i,j} \subseteq \sigma$, based on which soundness of our SAT-based subsumption encoding is derived.

\begin{theorem}[Soundness]\label{thm:subsumption-encoding-soundness-and-completeness}
Assume $\mainP$ does not contain duplicate literals. 
     Clause $\sideP$ subsumes  $\mainP$ iff the subsumption encoding $\SEnc(\sideP, \mainP)$ is satisfiable.
\end{theorem}

\begin{proof}
     Corollary~\ref{cor:subsumption-backward-compatibility} implies that, if $\SEnc(\sideP, \mainP)$ is satisfied, all  propositional variables $b^+_{i,j}$ can be replaced by $\sigma(\sideL_i) = \mainL_j$.
    \textbf{SAT-based partial completeness} \eqref{eq:subsumption-completeness-sat} yields $\bigwedge_i \bigvee_j \sigma(\sideL_i) = \mainL_j$, that is $\forall i \exists j.\ \sigma(\sideL_i) = \mainL_j$.
    Using the at-most-one constraint of~\eqref{eq:subsumption-multiplicity-sat},  if $b_{i,j}^+\neq b_{i,j'}^+$ then $b_{i,j}^+\implies \neg b_{i,j'}^+$.
    Based on~\eqref{eq:positive-match-def}, we hence obtain that
    if $\SEnc(\sideP, \mainP)$ is satisfiable, then $\sigma(\sideP) \sqsubseteq \mainP$ by Theorem~\ref{thm:subsumption-constraints}.

    For the other direction, assume $\sideP$ subsumes $\mainP$; that is, $\sideP \sqsubseteq \mainP$. Based on Assumption~\ref{ass:no-duplicates}, $\mainP$ has no duplicate literals. For $\sigma(\sideP)$ to be a sub-multiset of $\mainP$, it  should not contain duplicates either. Therefore, there exists a total bijective function $j(i)$ such that $\sigma(\sideL_i) = \mainL_{j(i)}$. From this function, one can build a model such that  $b^+_{i,j(i)} = \top$ for all $i$, and all other variables are false. This model satisfies  \textbf{positive compatibility}~\eqref{eq:positive-compatibility}. Indeed, since $\sigma(\sideL_i) = \mainL_{j(i)}$, we have $\Sigma^+_{i, j(i)} \subseteq \sigma$. \textbf{SAT-based partial completeness} \eqref{eq:subsumption-completeness-sat} is also satisfied since $j(i)$ is a total function. \textbf{SAT-based multiplicity conservation} \eqref{eq:subsumption-multiplicity-sat} is ensured by the bijectivity of $j(i)$. In summary, if $\sideP \sqsubseteq \mainP$, then $\SEnc$ is satisfiable.
\end{proof}

\begin{example}[Subsumption with $\SEnc(\sideP, \mainP)$]
    Consider the following clause pair $\sideP = \sideL_1 \lor \sideL_2 \lor \sideL_3$ and $\mainP = \mainL_1 \lor \mainL_2 \lor \mainL_3$, with
    \begin{center}
        \begin{tabular}{l l}
             $\sideL_1 = q(x_1)$ & $\mainL_1 = q(c)$ \\
             $\sideL_2 = p(x_1, x_2)$ & $\mainL_2 = p(c, d)$ \\
             $\sideL_3 = p(x_2, x_1)$ & $\mainL_3= p(d, c)$
        \end{tabular}
    \end{center}
    We first construct the substitution constraints matching the different literal pairs $(\sideL_i, \mainL_j)$:
    \begin{equation*}
        \Sigma^+_{i,j} = \left(
        \begin{tabular}{C{3.3cm} C{3.3cm} C{3.3cm}}
             $\{x_1 \mapsto c\}$ & $\IncSubst$ & $\IncSubst$ \\
             $\IncSubst$ & $\{x_1 \mapsto c, x_2 \mapsto d\}$ & $\{x_1 \mapsto d, x_2 \mapsto c\}$ \\
             $\IncSubst$ & $\{x_1 \mapsto d, x_2 \mapsto c\}$ & $\{x_1 \mapsto c, x_2 \mapsto d\}$\\
        \end{tabular}
        \right)
    \end{equation*}
    The SAT encoding $\SEnc(\sideP, \mainP)$ of subsumption is given by: 
    \begin{align*}
         & b^+_{1,1} \Rightarrow \{x_1 \mapsto c\} \subseteq \sigma 
         & \textbf{positive compatibility} \\
         & b^+_{2,2} \Rightarrow \{x_1 \mapsto c, x_2 \mapsto d\} \subseteq \sigma 
         & \textbf{positive compatibility} \\
         & b^+_{2,3} \Rightarrow \{x_1 \mapsto d, x_2 \mapsto c\} \subseteq \sigma 
         & \textbf{positive compatibility} \\
         & b^+_{3,2} \Rightarrow \{x_1 \mapsto d, x_2 \mapsto c\} \subseteq \sigma 
         & \textbf{positive compatibility} \\
         & b^+_{3,3} \Rightarrow \{x_1 \mapsto c, x_2 \mapsto d\} \subseteq \sigma 
         & \textbf{positive compatibility} \\
         & b^+_{1,1}
         & \textbf{SAT-based partial completeness}\\
         & b^+_{2,2} \lor b^+_{2,3}
         & \textbf{SAT-based partial completeness}\\
         & b^+_{3,2} \lor b^+_{3,3}
         & \textbf{SAT-based partial completeness}\\
         & \AMO(\{b^+_{1,1}\})
         & \textbf{SAT-based multiplicity conservation}\\
         & \AMO(\{b^+_{2,2}, b^+_{3,2}\})
         & \textbf{SAT-based multiplicity conservation}\\
         & \AMO(\{b^+_{2,3}, b^+_{3,3}\})
         & \textbf{SAT-based multiplicity conservation}\\
    \end{align*}

    Our tailored SAT solver from Section~\ref{sec:implementation} returns a model $\{b^+_{1,1}, b^+_{2,2}, \neg b^+_{2,3}, \neg b^+_{3,2}, b^+_{3,3}\}$ that satisfies $\SEnc(\sideP, \mainP)$.
    We build the final substitution $\sigma$ witnessing that $\sigma(\sideP) \sqsubseteq \mainP$, and hence $S$ subsumes $M$,
    as the union of all the substitutions bound to variables assigned to true. This gives $\sigma = \{x_1 \mapsto c, x_2 \mapsto d\}$.
\end{example}

\subsection{Direct SAT Encoding of Subsumption Resolution }
Similarly to subsumption, we translate the constraints of Theorem~\ref{thm:subsumption-resolution-constraints} into SAT, while also
considering both \textbf{positive compatibility}~\eqref{eq:positive-compatibility} and \textbf{negative compatibility}~\eqref{eq:negative-compatibility}. The following SAT constraints are derived from Theorem~\ref{thm:subsumption-resolution-constraints}:
\begin{align}
  &\textbf{SAT-based existence  } &
  \bigvee_{i} \bigvee_{j} b_{i,j}^-
  \label{eq:subsumption-resolution-direct-existence} \\
  &\textbf{SAT-based uniqueness  } &
  \bigwedge_{j} \bigwedge_{i} \bigwedge_{i' \geq i}\bigwedge_{j'> j}\neg b_{i,j}^- \lor \neg b_{i',j'}^-
  \label{eq:subsumption-resolution-direct-uniqueness} \\
  &\textbf{SAT-based completeness  } &
  \bigwedge_{i} \bigvee_{j} b^+_{i,j} \lor b^-_{i,j}
  \label{eq:subsumption-resolution-direct-completeness}\\
  &\textbf{SAT-based coherence  } &
  \bigwedge_{j}\bigwedge_{i}\bigwedge_{i'}\neg b_{i,j}^+ \lor \neg b_{i',j}^-
  \label{eq:subsumption-resolution-direct-coherence}
\end{align}

\begin{theorem}\label{thm:subsumption-resolution-backward-compatibility}
    Assume that clause $\mainP$  does not have duplicate atoms, as in~\eqref{eq:no-duplicates}. Let $\MS(\sideP, \mainP) = \left\{\left(b^\pm_{i,j}, \Sigma^\pm_{i,j}\right)\right\}$ be the match set of $\sideP$ and $\mainP$.
   \textbf{Positive compatibility}~\eqref{eq:positive-compatibility},
    \textbf{negative compatibility}~\eqref{eq:negative-compatibility},
    and \textbf{completeness}~\eqref{eq:subsumption-resolution-direct-completeness}
     ensures that
    $\Sigma^+_{i,j} \subseteq \sigma \Rightarrow b^+_{i,j}$ and
    $\Sigma^-_{i,j} \subseteq \sigma \Rightarrow b^-_{i,j}$.
\end{theorem}
\begin{proof}
    We use a similar argumentation as in proving Theorem~\ref{thm:subsumption-backward-compatibility}. 
    We only prove the claim for  $\sideL_{i}, \mainL_{j}$ such that $\Sigma^+_{i,j} \subseteq \sigma$ and $b^+_{i,j} = \bot$; the other case is similar.
    \textbf{SAT-based completeness}~\eqref{eq:subsumption-resolution-direct-completeness} ensures that there exists $j'$ such that $b^+_{i,j'} \lor b^-_{i,j'}$.
     Using compatibility~\eqref{eq:positive-compatibility}--\eqref{eq:negative-compatibility}, we have  $b^+_{i,j'} \lor b^-_{i,j'} \Rightarrow (\sigma(\sideL_{i}) = m_{j'} \lor \sigma(\sideL_{i}) = \neg m_{j'})$.
    Similarly as in Theorem~\ref{thm:subsumption-backward-compatibility}, we  obtain  $\Sigma^+_{i, j}(\sideL_{i}) = \Sigma^+_{i, j'}(\sideL_{i}) \lor \Sigma^+_{i, j}(\sideL_{i}) = \neg \Sigma^-_{i, j'}(\sideL_{i})$, which is equivalent to $\mainL_{j} = \mainL_{j'} \lor \mainL_{j} = \neg \mainL_{j'}$.
    Since $\Sigma^+_{i,j} \subseteq \sigma$, we have  $\Sigma^+_{i,j} \neq \IncSubst$. Therefore,  $\Sigma^-_{i,j}$ is the incompatible substitution $\IncSubst$; that is, $\Sigma^-_{i,j} = \IncSubst$.
    From $\Sigma^-_{i,j} = \IncSubst$, we infer $b^-_{i,j} = \bot$.
    In short, we have $\neg b^+_{i,j} \land \neg b^-_{i,j} \land (b^+_{i,j'} \lor b^-_{i,j'})$, therefore $j \neq j'$ and \eqref{eq:no-duplicates} of Assumption~\ref{ass:no-duplicates} is violated.
In conclusion, $\eqref{eq:no-duplicates} \land \eqref{eq:positive-compatibility} \land \eqref{eq:negative-compatibility} \land \eqref{eq:subsumption-resolution-direct-completeness}$ implies $\eqref{eq:positive-match-def} \land \eqref{eq:negative-match-def}$.
\end{proof}

Following upon Theorem~\ref{thm:subsumption-resolution-backward-compatibility}, the (direct)  SAT formalization of  subsumption resolution is given below. 
\begin{definition}[Direct SAT Encoding of Subsumption Resolution]
    The \emph{direct SAT encoding of subsumption resolution $\DirEnc(\sideP, \mainP)$} for the side and main premises $\sideP$ and $\mainP$ is the conjunction of \textbf{positive compatibility}~\eqref{eq:positive-compatibility},  \textbf{negative compatibility}~\eqref{eq:negative-compatibility},  \textbf{existence}~\eqref{eq:subsumption-resolution-direct-existence},  \textbf{uniqueness}~\eqref{eq:subsumption-resolution-direct-uniqueness},  \textbf{completeness}~\eqref{eq:subsumption-resolution-completeness} and \textbf{coherence}~\eqref{eq:subsumption-resolution-direct-coherence}.
\end{definition}

\begin{corollary}\label{cor:subsumption-resolution-backward-compatibility}
    A model of the direct subsumption resolution encoding $\DirEnc(\sideP, \mainP)$  satisfies $\forall i, j.\ b^+_{i,j} \Leftrightarrow \sigma(\sideL_i) = \mainL_j$ and $\forall i, j.\ b^-_{i,j} \Leftrightarrow \sigma(\sideL_i) = \neg \mainL_j$.
\end{corollary}

Towards finding an effective SAT-solving approach,
Theorem~\ref{thm:direct-encoding-soundness-and-completeness} yields a direct algorithmic solution to  subsumption resolution.

\begin{theorem}[Soundness\label{thm:direct-encoding-soundness-and-completeness}]
Assume $\mainP$ does not contain duplicate atoms. 
    Clauses $\sideP$ and $\mainP$ are respectively the side and main premises of subsumption resolution iff $\DirEnc(\sideP, \mainP)$ is satisfiable.
\end{theorem}

\begin{proof}
    Similarly to Theorem~\ref{thm:subsumption-encoding-soundness-and-completeness}, we  use Corollary~\ref{cor:subsumption-resolution-backward-compatibility} together with the definition of $b^\pm_{i,j}$
    to obtain the SAT constraints of $\DirEnc(\sideP, \mainP)$ from Theorem~\ref{thm:subsumption-resolution-constraints}.

    Assume $\sideP$ and $\mainP$ are the side and main premises of subsumption resolution. There exists a substitution $\sigma$, a literal $\mainL'\in \mainP$ and a set of literals
    $\sideP' \subseteq \sideP$ such that $\sigma(\sideP') = \{\neg \mainL'\} \land \sigma(\sideP \setminus \sideP') \subseteq \mainP \setminus \{\mainL'\}$.
    We can build a model for  $b^\pm_{i,j}$ that satisfies each constraint of $\DirEnc(\sideP, \mainP)$.
    Without loss of generality, let $\mainL' = \mainL_1$.
    For each literal in $\sideL_{i'} \in \sideP'$, we set $b^-_{i',1} = \top$. All other variables $b^\pm_{i',j}$ are set to false. Let $\sideP^* = \sideP \setminus \sideP'$ and $\mainP^* = \mainP \setminus \{\mainL'\}$. If $\sigma(\sideP^*) \subseteq \mainP^*$, then there exists a function $j^*(i^*)$ such that for each literal $\sideL_{i^*}$, we have $\sigma(\sideL_{i^*}) = \mainL_{j^*(i^*)}$. For each literal $\sideL_{i^*}$, we set $b^+_{i^*, j^*(i^*)} = \top$ and all other variables are false. This assignment is indeed a model of $\DirEnc(\sideP, \mainP)$.
\end{proof}
\begin{example}[Subsumption Resolution with $\DirEnc(\sideP, \mainP)$\label{ex:dir-encoding}]
    Consider the following clause pair $\sideP = \sideL_1 \lor \sideL_2 \lor \sideL_3$ and $\mainP = \mainL_1 \lor \mainL_2 \lor \mainL_3$, with
    \begin{center}
        \begin{tabular}{l l}
               $\sideL_1 = p(f(x_1), x_2)$
             & $\mainL_1 = \neg p(f(c), d)$ \\
               $\sideL_2 = \neg p(x_2, x_1)$ 
             & $\mainL_2 = \neg p(d, c)$ \\
               $\sideL_3 = p(f(x_3), x_1)$ 
             & $\mainL_3 = p(f(y_1), c)$\\
        \end{tabular}
    \end{center}
    We build the following match sets:
    \begin{equation*}
        \Sigma^+_{i,j} = \left(
        \begin{tabular}{C{3.3cm} C{3.3cm} C{3.3cm}}
               $\IncSubst$
             & $\IncSubst$ 
             & $\{x_1 \mapsto y_1, x_2 \mapsto c\}$ \\
               $\{x_1 \mapsto d, x_2 \mapsto f(c)\}$
             & $\{x_1 \mapsto c, x_2 \mapsto d\}$
             & $\IncSubst$ \\
               $\IncSubst$ 
             & $\IncSubst$ 
             & $\{x_1 \mapsto c, x_3 \mapsto y_1\}$\\
        \end{tabular}
        \right)
    \end{equation*}

    \begin{equation*}
        \Sigma^-_{i,j} = \left(
        \begin{tabular}{C{3.3cm} C{3.3cm} C{3.3cm}}
               $\{x_1 \mapsto c, x_2 \mapsto d\}$
             & $\IncSubst$ 
             & $\IncSubst$ \\
               $\IncSubst$ 
             & $\IncSubst$
             & $\{x_1 \mapsto c, x_2 \mapsto f(y_1)\}$ \\
               $\{x_1 \mapsto d, x_3 \mapsto c\}$ 
             & $\IncSubst$ 
             & $\IncSubst$\\
        \end{tabular}
        \right)
    \end{equation*}
    \newpage
    We can express the direct subsumption resolution encoding $\DirEnc(\sideP, \mainP)$ as
    \begin{align*}
         & b^+_{1,3} \Rightarrow \{x_1 \mapsto y_1, x_2 \mapsto c\} \subseteq \sigma 
         & \textbf{positive compatibility} \\
         & b^+_{2,1} \Rightarrow \{x_1 \mapsto d, x_2 \mapsto f(c)\} \subseteq \sigma 
         & \textbf{positive compatibility} \\
         & b^+_{2,2} \Rightarrow \{x_1 \mapsto c, x_2 \mapsto d\} \subseteq \sigma 
         & \textbf{positive compatibility} \\
         & b^+_{3,3} \Rightarrow \{x_1 \mapsto c, x_3 \mapsto y_1\} \subseteq \sigma 
         & \textbf{positive compatibility} \\
         & b^-_{1,1} \Rightarrow \{x_1 \mapsto c, x_2 \mapsto d\} \subseteq \sigma 
         & \textbf{negative compatibility} \\
         & b^-_{2,3} \Rightarrow \{x_1 \mapsto c, x_2 \mapsto f(y_1)\} \subseteq \sigma 
         & \textbf{negative compatibility} \\
         & b^-_{3,1} \Rightarrow \{x_1 \mapsto d, x_3 \mapsto c\} \subseteq \sigma 
         & \textbf{negative compatibility} \\
         & b^-_{1,1} \lor b^-_{2,3} \lor b^-_{3,1}
         & \textbf{SAT-based existence} \\
         & \neg b^-_{1,1} \lor \neg b^-_{2,3}
         & \textbf{SAT-based uniqueness} \\
         & \neg b^-_{2,3} \lor \neg b^-_{3,1}
         & \textbf{SAT-based uniqueness} \\
         & b^-_{1,1} \lor b^+_{1,3}
         & \textbf{SAT-based completeness} \\
         & b^+_{2,1} \lor b^+_{2,2} \lor b^-_{2,3}
         & \textbf{SAT-based completeness} \\
         & b^-_{3,1} \lor b^+_{3,3}
         & \textbf{SAT-based completeness} \\
         & \neg b^-_{1,1} \lor \neg b^+_{2,1}
         & \textbf{SAT-based coherence} \\
         & \neg b^-_{3,1} \lor \neg b^+_{2,1}
         & \textbf{SAT-based coherence} \\
         & \neg b^-_{2,3} \lor \neg b^+_{1,3}
         & \textbf{SAT-based coherence} \\
         & \neg b^-_{2,3} \lor \neg b^+_{3,3}
         & \textbf{SAT-based coherence}
    \end{align*}

    Our  SAT solver from Section~\ref{sec:implementation}  derives the model  
    $\{ \lnot b^+_{1,3}, \lnot b^+_{2,1}, b^+_{2,2}, b^+_{3,3}, b^-_{1,1}, \lnot b^-_{2,3}, \lnot b^-_{3,1} \}$
    of $\DirEnc(\sideP, \mainP)$, as detailed in Example~\ref{ex:SATSolvingSRDirect}. The substitution $\sigma$ is correct and is composed of the union of all the substitutions bound to variables assigned true:
    \[
    \sigma =
    \bigcup\Bigl\{\Sigma^+_{i,j} \mathrel{\bigm|} b^+_{i,j} = \top\Bigr\}
    \cup
    \bigcup\Bigl\{\Sigma^-_{i,j} \mathrel{\bigm|} b^-_{i,j} = \top\Bigr\}
    \]
that is,
\begin{align*}
    \sigma &=
    \{x_1 \mapsto c, x_2 \mapsto d\} \cup
    \{x_1 \mapsto c, x_3 \mapsto y_1\} \cup
    \{x_1 \mapsto c, x_2 \mapsto d\} \\
    &= \{x_1 \mapsto c, x_2 \mapsto d, x_3 \mapsto y_1\}
\end{align*}

  The conclusion clause of \SR{} is  built from the model by removing $\mainL_1$ from the main premise because $b^-_{1,1} = \top$. This gives us the resolution clause $\mainP \setminus \{\mainL_1\} = \neg p(d, c) \lor p(f(y_1), c)$, which subsumes $\mainP$.

\end{example}

\subsection{Indirect SAT Encoding of  Subsumption Resolution}
The direct SAT encoding $\DirEnc(\sideP, \mainP)$ of subsumption resolution has a potential inefficiency due to the fact that the \textbf{uniqueness constraint}~\eqref{eq:subsumption-resolution-direct-uniqueness} may create a quartic number of clauses in the worst case. We circumvent this issue by trading off constrains for variables, resulting in an indirect SAT encoding $\IndEnc(\sideP, \mainP)$ of subsumption resolution. Doing so,  we introduce a new set of  propositional variables $c_j$ such that $c_j$ is true iff $\mainL_j$ is the resolution literal of \SR. In other words, $c_j \Leftrightarrow \exists i.\ \sigma(\sideL_i) = \neg \mainL_j$.

We encode the role of $c_j$ with  constraint~\eqref{eq:subsumption-resolution-indirect-encoding} given below:
\begin{align}
  &\textbf{SAT-based structurality } &
  \bigwedge_j \left[\neg c_j \lor \bigvee_i b_{i,j}^- \right] \land \bigwedge_j \bigwedge_i \left(c_j \lor \neg b_{i,j}^-\right)
  \label{eq:subsumption-resolution-indirect-encoding}
\end{align}
Using variables $c_j$, the constraints of Theorem~\ref{thm:subsumption-resolution-constraints} are turned into the following SAT formulas: 

\begin{align}
  &\textbf{SAT-based revised existence  } &
  \bigvee_{j} c_j
  \label{eq:subsumption-resolution-indirect-existence} \\
  &\textbf{SAT-based revised uniqueness  } &
  \AMO(\{c_{j}, j =1,...,|\mainP|\})
  \label{eq:subsumption-resolution-indirect-uniqueness} \\
  &\textbf{SAT-based completeness  } &
  \bigwedge_{i} \bigvee_{j} b^+_{i,j} \lor b^-_{i,j}
  \label{eq:subsumption-resolution-indirect-completeness}\\
  &\textbf{SAT-based revised coherence  } &
  \bigwedge_{j} \bigwedge_{i} \left(\neg c_j \lor \neg b_{i,j}^+\right)
  \label{eq:subsumption-resolution-indirect-coherence}
\end{align}

\begin{definition}[Indirect SAT Encoding of Subsumption Resolution]
    The \emph{indirect SAT encoding for subsumption resolution $\IndEnc(\sideP, \mainP)$} for clauses $\sideP$ and $\mainP$ is the conjunction of  \textbf{positive compatibility}~\eqref{eq:positive-compatibility}, \textbf{negative compatibility}~\eqref{eq:negative-compatibility}, \textbf{structurality}~\eqref{eq:subsumption-resolution-indirect-encoding}, \textbf{revised existence}~\eqref{eq:subsumption-resolution-indirect-existence}, \textbf{revised uniqueness}~\eqref{eq:subsumption-resolution-indirect-uniqueness}, \textbf{completeness}~\eqref{eq:subsumption-resolution-indirect-completeness}, and \textbf{revised coherence}~\eqref{eq:subsumption-resolution-indirect-coherence}.
\end{definition}
With this new indirect encoding $\IndEnc(\sideP, \mainP)$, the number of clauses is only quadratic with respect to the length of the clauses.

\begin{theorem}[Soundness]\label{thm:indirect-encoding-soundness-and-completeness} 
Assume $\mainP$ does not contain duplicate literals, as in~\eqref{eq:no-duplicates}.
    Clauses $\sideP$ and $\mainP$ are  the side and main premise of subsumption resolution iff $\IndEnc(\sideP, \mainP)$ is satisfiable.
\end{theorem}
\begin{proof}
    From Theorem~\ref{thm:subsumption-resolution-backward-compatibility},  if $\eqref{eq:positive-compatibility} \land \eqref{eq:negative-compatibility} \land \eqref{eq:subsumption-resolution-indirect-completeness}$ is satisfiable, then $\forall i, j.\ b^+_{i,j} \Leftrightarrow \sigma(\sideL_i) = \mainL_j$ and $\forall i, j.\ b^-_{i,j} \Leftrightarrow \sigma(\sideL_i) = \neg \mainL_j$. Using~\eqref{eq:subsumption-resolution-indirect-encoding}, we obtain $\forall j.\ c_j \Leftrightarrow \exists i.\ \sigma(\sideL_i) = \neg \mainL_j$. Based on \eqref{eq:subsumption-resolution-indirect-existence}-\eqref{eq:subsumption-resolution-indirect-coherence}, we obtain the subsumption resolution constraints of Theorem~\ref{thm:subsumption-resolution-constraints}. Therefore, if $\IndEnc(\sideP, \mainP)$ is satisfiable, then subsumption resolution can be applied over $(\sideP,\mainP)$.

    For the other direction, assume 
    subsumption resolution can be applied over $(\sideP,\mainP)$. Then, we can build a model that satisfies $\IndEnc(\sideP, \mainP)$, as follows. There exists a substitution $\sigma$, a literal $\mainL'\in \mainP$ and a set of literals $\sideP' \subseteq \sideP$ such that $\sigma(\sideP') = \{\neg \mainL'\} \land \sigma(\sideP \setminus \sideP') \subseteq \mainP \setminus \{\mainL'\}$. 
    Without loss of generality, let $\mainL_1 = \mainL'$ be the resolution literal of \SR. We set $c_1 = \top$ and all the other $c_j$ to false. For each literal in $\sideL_{i'} \in \sideP'$, we set $b^-_{i', 1} = \top$ and $\ b^-_{i', j} = \bot$, for $j\neq 1$; further, $\ b^+_{i', j} = \bot$, for all $j$. Let $\sideP^* = \sideP \setminus \sideP'$. For each literal $\sideL_{i^*} \in \sideP^*$, there exists a literal $\mainL_{j^*} \in \mainP \setminus \{\mainL_1\}$ such that $\sigma(\sideL_{i^*}) = \mainL_{j^*}$. We set $b^+_{i^*, j^*} = \top$; $b^+_{i^*, j} = \bot$, for 
$j\neq j^*$; 
    and $b^-_{i^*, j} = \bot$, for all $j$. This is indeed a model of $\IndEnc(\sideP, \mainP)$.
\end{proof}

We note that, in practice, the number of clauses of the indirect SAT encoding can be greater than the direct SAT encoding, even for large clauses. Indeed, it is not necessary to define variables for literal pairs that we know in advance cannot be matched. If $\Sigma^+_{i,j} = \IncSubst$, we do not define $b^+_{i,j}$ because  the constraint $b^+_{i,j} \Rightarrow \Sigma^+_{i,j} \subseteq \sigma$ will be reduced to $b^+_{i,j} \Rightarrow \bot$ and $b^+_{i,j}$ is always  false. We do not need to add the clauses containing $\neg b^+_{i,j}$,
and we remove the literals $b^+_{i,j}$ where it appears. In practice, most instances of subsumption and subsumption resolution  have a sparse Boolean variable set, and  behave quite well even with the direct SAT encoding. Choosing which encoding to use is discussed in Section~\ref{sec:choosing-encoding}.

\begin{example}[Subsumption Resolution with $\IndEnc(\sideP, \mainP)$]\label{ex:indEncodingSR}
    Consider clauses from Example~\ref{ex:dir-encoding}. Namely, $\sideP = \sideL_1 \lor \sideL_2 \lor \sideL_3$ and $\mainP = \mainL_1 \lor \mainL_2 \lor \mainL_3$, with

    \begin{center}
        \begin{tabular}{l l}
               $\sideL_1 = p(f(x_1), x_2)$
             & $\mainL_1 = \neg p(f(c), d)$ \\
               $\sideL_2 = \neg p(x_2, x_1)$ 
             & $\mainL_2 = \neg p(d, c)$ \\
               $\sideL_3 = p(f(x_3), x_1)$ 
             & $\mainL_3 = p(f(y_1), c)$\\
        \end{tabular}
    \end{center}

    In the indirect SAT encoding $\IndEnc(\sideP, \mainP)$, we introduce two extra variables $c_1$ and $c_3$ such that $c_1$ is true iff $\exists i.\ b^-_{i,1}$, and $c_3$ is true iff $\exists i.\ b^-_{i,3}$. It is not necessary to define $c_2$ since no negative polarity matches exist towards $\mainL_2$, and $c_2$ is set to false.
    The SAT constraints identical to the direct SAT encoding $\DirEnc(\sideP, \mainP)$ are written below  in \gray{light gray} to better highlight the difference between $\DirEnc(\sideP, \mainP)$ and $\IndEnc(\sideP, \mainP)$.
    \begin{align*}
        & \textcolor{gray}{b^+_{1,3} \Rightarrow \{x_1 \mapsto y_1, x_2 \mapsto c\} \subseteq \sigma }
        & \textcolor{gray}{\textbf{positive compatibility}} \\
        & \textcolor{gray}{b^+_{2,1} \Rightarrow \{x_1 \mapsto d, x_2 \mapsto f(c)\} \subseteq \sigma} 
        & \textcolor{gray}{\textbf{positive compatibility}} \\
        & \textcolor{gray}{b^+_{2,2} \Rightarrow \{x_1 \mapsto c, x_2 \mapsto d\} \subseteq \sigma}
        & \textcolor{gray}{\textbf{positive compatibility}} \\
        & \textcolor{gray}{b^+_{3,3} \Rightarrow \{x_1 \mapsto c, x_3 \mapsto y_1\} \subseteq \sigma}
        & \textcolor{gray}{\textbf{positive compatibility}} \\
        & \textcolor{gray}{b^-_{1,1} \Rightarrow \{x_1 \mapsto c, x_2 \mapsto d\} \subseteq \sigma}
        & \textcolor{gray}{\textbf{negative compatibility}} \\
        & \textcolor{gray}{b^-_{2,3} \Rightarrow \{x_1 \mapsto c, x_2 \mapsto f(y_1)\} \subseteq \sigma}
        & \textcolor{gray}{\textbf{negative compatibility}} \\
        & \textcolor{gray}{b^-_{3,1} \Rightarrow \{x_1 \mapsto d, x_3 \mapsto c\} \subseteq \sigma}
        & \textcolor{gray}{\textbf{negative compatibility}} \\
        & \neg c_1 \lor b^-_{1, 1} \lor b^-_{3, 1}
        & \textbf{SAT-based structurality} \\
        & c_1 \lor \neg b^-_{1, 1}
        & \textbf{SAT-based structurality} \\
        & c_1 \lor \neg b^-_{3, 1}
        & \textbf{SAT-based structurality} \\
        & \neg c_3 \lor b^-_{2,3}
        & \textbf{SAT-based structurality} \\
        & c_3 \lor \neg b^-_{2,3}
        & \textbf{SAT-based structurality} \\
        & c_1 \lor c_3
        & \textbf{SAT-based revised existence} \\
        & \AMO(\{c_1, c_3\})
        & \textbf{SAT-based revised uniqueness} \\
        & \textcolor{gray}{b^-_{1,1} \lor b^+_{1,3}}
        & \textcolor{gray}{\textbf{SAT-based completeness}} \\
        & \textcolor{gray}{b^+_{2,1} \lor b^+_{2,2} \lor b^-_{2,3}}
        & \textcolor{gray}{\textbf{SAT-based completeness}} \\
        & \textcolor{gray}{b^-_{3,1} \lor b^+_{3,3}}
        & \textcolor{gray}{\textbf{SAT-based completeness}} \\
        & \neg c_1 \lor \neg b^+_{2,1}
        & \textbf{SAT-based revised coherence} \\
        & \neg c_3 \lor \neg b^+_{1,3}
        & \textbf{SAT-based revised coherence} \\
        & \neg c_3 \lor \neg b^+_{3,3}
        & \textbf{SAT-based revised coherence}
    \end{align*}

    Using the above indirect encoding $\IndEnc(\sideP, \mainP)$, our SAT solver in Section~\ref{sec:implementation} finds the same model (substitution) of subsumption resolution as in Example~\ref{ex:dir-encoding}, with $c_1,c_2$ being assigned true and false respectively.
\end{example}

We remark that the indirect encoding $\IndEnc(\sideP, \mainP)$ does not seem to have much of an advantage on small examples similar to Example~\ref{ex:indEncodingSR}.  Indeed, structurality~\eqref{eq:subsumption-resolution-indirect-encoding} adds a few clauses that are not necessary with the direct encoding $\DirEnc(\sideP, \mainP)$. In Section~\ref{sec:choosing-encoding} we empirically show that the indirect encoding $\IndEnc(\sideP, \mainP)$ of subsumption resolution  performs better on larger clauses.

%% file: 5-implementation.tex
We now describe our approach for solving the SAT-based encodings $\SEnc(\sideP, \mainP)$, $\DirEnc(\sideP, \mainP)$
and $\IndEnc(\sideP, \mainP)$ of Section~\ref{sec:satconstraints} for subsumption and subsumption resolution.
We first introduce our SAT solver adjusted for the efficient handling of  subsumption (resolution) constraints,
important for reasoning about substitution constraints $\Sigma^\pm_{i,j}\subseteq \sigma$ and at-most-one constraints (Section~\ref{sec:sat-solver}).
We then describe pruning-based preprocessing steps of subsumption (resolution) instances (Section~\ref{sec:pruning}),
with the purpose of improving SAT-based solving of subsumption and subsumption resolution.

\paragraph{Lightweight SAT Solving.}
We use the term \emph{lightweight SAT Solving} to
highlight an important engineering aspect
when designing a SAT solver for subsumption and subsumption resolution.
A typical run of a first-order theorem prover
involves a large number of simple subsumption (resolution) tests
and a small number of hard tests.
Even after pruning, most instances that make it to the SAT solver are solved quickly
(see also Section~\ref{sec:cutoff} and Figure~\ref{fig:cutoff}).
As a result, some care must be taken to ensure that setup of the SAT instances is efficient,
because a large overhead may easily outweigh gains in solving efficiency.

\subsection{SAT Solver for Subsumption Encodings\label{sec:sat-solver}}
 Recall that the SAT-based encodings 
 $\SEnc(\sideP, \mainP)$, $\DirEnc(\sideP, \mainP)$ and $\IndEnc(\sideP, \mainP)$ of subsumption and subsumption resolution
use substitution constraints $\Sigma^\pm_{i,j}\subseteq \sigma$ and at-most-one constraints ($\AMO$), 
which are out of scope for standard SAT solvers~\cite{DBLP:conf/sat/EenS03,DBLP:conf/sat/BiereFW23}. 
A na\"ive SAT approach of handling such constraints would be translating  $\Sigma^\pm_{i,j} \subseteq \sigma$ and $\AMO$ formulas into purely propositional clauses.
However, such a translation would either require additional propositional variables to encode $\AMO$ constraints
or would come with a quadratic%
\footnote{%
Quadratic in the size of the $\AMO$ constraint.
}
number of propositional clauses~\cite{FrischGiannaros:2010};
a similar situation also occurs for substitution constraints $\Sigma^\pm_{i,j}\subseteq \sigma$. To ensure efficient solving of subsumption (resolution), 
solving our SAT encodings needs to be \emph{lightweight} in order to be practically feasible during redundancy checking in a first-order theorem prover. 

As a remedy to overcome the increase in propositional variables/clauses in a na\"ive SAT translation approach,
we support substitution constraints as in~$\eqref{eq:positive-compatibility}$ and~$\eqref{eq:negative-compatibility}$,
as well as $\AMO$ constraints as in
\textbf{SAT-based multiplicity conservation}~$\eqref{eq:subsumption-multiplicity-sat}$
and \textbf{SAT-based revised uniqueness}~$\eqref{eq:subsumption-resolution-indirect-uniqueness}$, 
natively in SAT solving.
In particular, we adjust unit propagation and conflict resolution in CDCL-based SAT solving for handling
propositional formulas with substitution constraints and $\AMO$ constraints.

\paragraph{At-most-one constraints.}
Consider the constraint $\AMO(\{b_1, b_2, \dots, b_n\})$, which  is equivalent to the following purely propositional formula:
\begin{equation}
    \label{eq:amo-to-sat}
    \bigwedge_{i} \bigwedge_{j>i} \neg b_i \lor \neg b_j
\end{equation}
To keep our encoding of~$\AMO{}$ constraints lightweight, we combine SAT solving with $\AMO$ constraints in a way similar to SMT solving, as follows.

When the constraint $\AMO(\{b_1, b_2, \dots, b_n\})$ is added
to the SAT solver, each of the variables $b_1, b_2, \dots, b_n$ watches
the constraint.
Whenever one of the variables $b_i$ is assigned true,
all $b_j$ with $j\neq i$ must be false in order not to violate $\AMO(\{b_1, b_2, \dots, b_n\})$; hence $b_j$ are propagated to false.
The reasons of these propagations are exactly
the clauses $\lnot b_i \lor \lnot b_j$
of~\eqref{eq:amo-to-sat}; 
however, these clauses do not need to be explicitly constructed. 
Conflict analysis in SAT solving then behaves as usual, 
without special considerations for $\AMO$ constraints.

\paragraph{Compatibility constraints.}
Similar to $\AMO$ constraints, a compatibility constraint is equivalent to a set of binary clauses, as given in~\eqref{eq:positive-compatibility}-\eqref{eq:negative-compatibility}. 
Let $\Sigma_1 \substincompat \Sigma_2$ denote that
the substitutions $\Sigma_1$ and $\Sigma_2$ are incompatible; based on Definition~\ref{def:Substitution:Compatibility}, there  exists thus a variable $x$ such that $\Sigma_1(x) \neq \Sigma_2(x)$.
Let $F$ be the set of constraints under consideration.
The purely propositional semantics of the compatibility constraints~\eqref{eq:positive-compatibility}-\eqref{eq:negative-compatibility} is the  clause set:
\begin{equation}\label{eq:SAT:CC}
    \left\{
        \neg b \lor \neg b' \mid
        (b \Rightarrow \Sigma\subseteq \sigma)\in F \land (b' \Rightarrow \Sigma' \subseteq \sigma)\in F \land \Sigma \substincompat \Sigma'
    \right\}
\end{equation}

\newcommand{\globalsubst}{\sigma_\tau}
We remark that it is not necessary to generate the clauses~\eqref{eq:SAT:CC} explicitly.
Conceptually, our SAT solver updates a global substitution $\globalsubst$
whenever a Boolean variable~$b$ with associated
substitution constraint $\Sigma\subseteq\sigma$ is assigned true.
Our SAT solver then ensures that the following invariant holds:
\begin{equation}
    \label{eq:sat-solver-invariant}
    \globalsubst =
    \bigcup \bigl\{\Sigma \mathrel{\bigm|} b \Rightarrow (\Sigma \subseteq \sigma)\in \PropF \land b \in \tau \bigr\},
\end{equation}
where $\tau$ is the current set of assigned literals of the SAT solver (i.e., the trail).
Our SAT solver  uses $\globalsubst$ to propagate
any Boolean variables bound to incompatible substitutions to false.

We note that, in practice, it is not necessary to keep $\globalsubst$
explicitly;  instead it suffices to maintain a lookup table that
allows propagating such incompatible substitutions.
\newcommand{\bindingslookup}{\mathit{Bindings}}
Concretely, each first-order variable~$x$
watches the set~$\bindingslookup(x)$ of Boolean variables~$b$ that impose a binding
on~$x$ along with the bound term~$t$:
\[
    \bindingslookup(x) = \{ 
        (b, t) \mid
        b \rightarrow (\Sigma\subseteq\sigma) \in F \land
        t = \Sigma(x)
    \}
\]
When the global substitution $\globalsubst$ is updated
with a variable~$x$ newly mapped to a term~$t$,
our SAT solver uses $\bindingslookup(x)$ to retrieve
all the Boolean variables~$b'$
with an associated substitution constraint $\Sigma'\subseteq\sigma$
such that $\Sigma'(x) \neq t$.
The solver  then propagates~$b'$ to false,
and the propagation reason is the binary clause $\lnot b \lor \lnot b'$,
where~$b$ is the Boolean variable that caused $\globalsubst(x)=t$.

As a result, our SAT solver ensures that $\Sigma' \substincompat \sigma$ implies $\lnot b'$.
We perform this propagation of incompatible substitution constraints
immediately when a Boolean variable is assigned true.
This way, we enforce the  invariant~\eqref{eq:sat-solver-invariant}
and guarantee there can be no conflict due to substitution constraints.
Indeed, if $b\Rightarrow \Sigma\in \sigma$ and $\Sigma\substincompat\globalsubst$, then $b$ would have been assigned false before.

\begin{example}[SAT Solving of the SAT-Based Direct Encoding of Subsumption Resolution]\label{ex:SATSolvingSRDirect}
We illustrate the main steps of our SAT solver using the direct encoding $\DirEnc(\sideP,\mainP)$ of subsumption resolution from Example~\ref{ex:dir-encoding}. 

A potential execution of our SAT solver on 
$\DirEnc(\sideP,\mainP)$ decides $b^+_{1,3} = \top$.
This imposes, among others, the mapping $x_1\mapsto y_1$,
and due to the \textbf{compatibility} constraints
all other Boolean variables are immediately propagated to false.
This leads to conflicts with the \textbf{existence}
and some \textbf{completeness} constraints.
Assume the solver discovers the conflict with $b^-_{3,1} \lor b^+_{3,3}$.
As explained above, the reasons for propagating these literals are
the implicit binary clauses
$\lnot b^+_{1,3} \lor \lnot b^-_{3,1}$
and
$\lnot b^+_{1,3} \lor \lnot b^+_{3,3}$,
and after resolution, the solver will backtrack,
learn the asserting clause $\lnot b^+_{1,3}$,
and propagate $b^+_{1,3} = \bot$.
With \textbf{completeness} $b^-_{1,1} \lor b^+_{1,3}$,
the solver propagates $b^-_{1,1} = \top$,
which imposes the mappings $x_1\mapsto c$ and $x_2\mapsto d$
on $\globalsubst$.
By \textbf{compatibility}, the solver now propagates
$b^+_{2,1} = \bot$,
$b^-_{2,3} = \bot$,
and
$b^-_{3,1} = \bot$.
With the remaining \textbf{completeness} constraints,
the solver now propagates
$b^+_{2,2} = \top$
and  $b^+_{3,3} = \top$.
At this point, all Boolean variables are assigned
and all constraints are satisfied, yielding the model
$\{ \lnot b^+_{1,3}, \lnot b^+_{2,1}, b^+_{2,2}, b^+_{3,3}, b^-_{1,1}, \lnot b^-_{2,3}, \lnot b^-_{3,1} \}$
of $\DirEnc(\sideP, \mainP)$ from Example~\ref{ex:dir-encoding}.
\end{example}

\subsection{Pruning Subsumption Variants for SAT Solving}\label{sec:pruning}
Reducing the number of (trivially unsat) instances of subsumption and subsumption resolution is an important preprocessing step for increasing the effectiveness of our SAT solving engine from Section~\ref{sec:sat-solver}.

\paragraph{Pruning subsumption.}
We prune unsat subsumption instances between $(\sideP,\mainP)$
by checking multiset inclusion between the predicates of the atoms of $\sideP, \mainP$,
together with their polarities.
Intuitively, this pruning step allows to easily determine that there exists no bijective function $j(i)$
such that $\sigma(\sideL_i) = \mainL_{j(i)}$ if the atom cardinalities do not match.

More formally, let~$\mathcal{Pre}(\ell)$ compute the predicate corresponding
to literal~$\ell$ and~$\mathcal{Q}(\ell)$ denote the polarity of~$\ell$.
Our pruning criterion for subsumption is:
\begin{equation}
    \bigr\{ \bigl(\mathcal{P}(\sideL_i), \mathcal{Q}(\sideL_i)\bigr) \bigm| \sideL_i \in \sideP \bigr\}
    \sqsubseteq
    \bigl\{ \bigl(\mathcal{P}(\mainL_j), \mathcal{Q}(\mainL_j)\bigr) \bigm| \mainL_j \in \mainP \bigr\}
    \label{eq:prune-s}
\end{equation}

\begin{theorem}[Pruning Subsumption]
    \label{thm:prune-s}
    If the pruning criterion~\eqref{eq:prune-s} is unsat, then $\sideP$ does not subsume $\mainP$.
\end{theorem}
\begin{proof}
    The multisets $\{ (\mathcal{P}(\sideL_i), \mathcal{Q}(\sideL_i)) \mid \sideL_i \in \sideP \}$
    and $\{ (\mathcal{P}(\mainL_j), \mathcal{Q}(\mainL_j)) \mid \mainL_j \in \mainP \}$
    are projections $\pi$ of the multisets of literals of $\sideP$ and $\mainP$ respectively.
    This projection $\pi$ has the property to make its argument substitution agnostic.
    That is, if there exists $\sigma$ such that $\sigma(\sideL_i) = \mainL_j$, then $\sideL_i$ and $\mainL_j$ are projected on the same location; that is, $(\sigma(\sideL_i) = \mainL_j) \Rightarrow (\pi(\sideL_i) = \pi( \mainL_j))$.
    Therefore, if $\pi(\sideL_i) \neq \pi(\mainL_j)$, then there exist no matching substitution between $\sideL_i$ and $\mainL_j$.
    If formula~\eqref{eq:prune-s} is unsat, then $\pi(\sideP) \nsqsubseteq \pi(\mainP)$, implying that there exists no substitution $\sigma$ such that $\sigma(\sideP) \sqsubseteq \mainP$; as such, subsumption cannot be applied between $(\sideP,\mainP)$.
\end{proof}

\paragraph{Pruning subsumption resolution.}
We similarly prune unsat instances of subsumption resolution,
by using a weaker version of~\eqref{eq:prune-s}.
Namely, for pruning unsat subsumption resolution instances,
we only check  set inclusion between the predicate sets of $\sideP$ and $\mainP$:
\begin{equation}
    \bigl\{\mathcal{P}(\sideL_i) \bigm| \sideL_i \in \sideP \bigr\}
    \subseteq
    \bigl\{\mathcal{P}(\mainL_j) \bigm| \mainL_j \in \mainP \bigr\}
    \label{eq:prune-sr}
\end{equation}

\begin{theorem}\label{thm:prune-s-sr}
    Validity of the subsumption pruning criterion \eqref{eq:prune-s}
    implies validity of the subsumption resolution pruning criterion \eqref{eq:prune-sr}.
\end{theorem}
\begin{proof}
    The sets
    $\{ \mathcal{P}(\sideL_i) \mid \sideL_i \in \sideP \}$ and
    $\{ \mathcal{P}(\mainL_j) \mid \mainL_j \in \mainP \}$
    are obtained by a projection $\pi$ from
    $\{ (\mathcal{P}(\sideL_i), \mathcal{Q}(\sideL_i)) \mid \sideL_i \in \sideP \}$ and
    $\{ (\mathcal{P}(\mainL_j), \mathcal{Q}(\mainL_j)) \mid \mainL_j \in \mainP \}$, respectively.
    Therefore, for each pair of elements
    \(
        (e, e') \in
        \{(\mathcal{P}(\sideL_i), \mathcal{Q}(\sideL_i)) \mid \sideL_i \in \sideP \} \times
        \{(\mathcal{P}(\mainL_j), \mathcal{Q}(\mainL_j)) \mid \mainL_j \in \mainP \}
    \),
    if $e = e'$,
    then $\pi(e) = \pi(e')$,
    and the multiset inclusion is preserved.
    As $\mathcal{S}_1 \sqsubseteq \mathcal{S}_2$ (multiset inclusion)
    implies $\mathcal{S}_1 \subseteq \mathcal{S}_2$ (set inclusion),
    we obtain that~\eqref{eq:prune-s} implies~\eqref{eq:prune-sr}.
\end{proof}

The following is an immediate consequence of Theorems~\ref{thm:prune-s}-\ref{thm:prune-s-sr}. 

\begin{corollary}\label{cor:pruning-sr-weaker}
    If the pruning criterion \eqref{eq:prune-sr} is not satisfied, then $\sideP$ does not subsume $\mainP$.
\end{corollary}

Similarly to Theorem~\ref{thm:prune-s}, we  use the pruning criterion~\eqref{eq:prune-sr} to detect (and delete) unsat subsumption resolution instances between $(\sideP,\mainP)$. 

\begin{theorem}[Pruning Subsumption Resolution]
    If the pruning criterion~\eqref{eq:prune-sr} is unsat, then $\sideP$ and $\mainP$ are not side and main premises of subsumption resolution.
\end{theorem}
\begin{proof}
    Similarly to Theorem~\ref{thm:prune-s}, if criterion~\eqref{eq:prune-sr} is not satisfied, then there exists a literal $\sideL_i \in \sideP$ that cannot be matched with any literal in $\mainP$; as such, the   \textbf{completeness} constraint~\eqref{eq:subsumption-resolution-completeness} of subsumption resolution is violated.
\end{proof}

\paragraph{Fast implementations of pruning.}
To represent the predicate sets used in our pruning criterion, we use an array $\mathcal{A}$ of unsigned integers whose index is the index of the predicate. We first build the multiset with the predicates of the main premise $\mainP$. When a predicate is hit, the value stored in $\mathcal{A}$ is incremented,  we check that $\sideP$ contains a sub-multiset of predicates, and decrement the previously stored value within $A$.

Storing $\mathcal{A}$ entries while applying pruning checks may be memory-expensive.
Resetting the memory before each pruning is also an expensive operation. 
We therefore use a time stamp $t$ such that  $\forall i.\ \mathcal{A}[i] < t + |\mainP|$ holds.
Intuitively, before pruning is applied, $\forall i.\ \mathcal{A}[i] < t$  holds.
Algorithm~\ref{alg:pruning} summarizes  our pruning procedure using time stamps.
\begin{algorithm}[t]
    \caption{Pruning algorithm for subsumption and subsumption resolution}
    \label{alg:pruning}
    \begin{algorithmic}
        \State $N \gets$ number of predicate symbols
        \State $t \gets 0$
        \Comment{\parbox{0.45\textwidth}{time stamp}}
        \State $\mathcal{A} \gets zeros(2\cdot N)$ 
        \Comment{\parbox{0.45\textwidth}{array to hold the (multi-)set}}
        \Function{headerIndex}{$l$}
            \If {$\mathcal{Q}(l)$}
            \Comment{\parbox{0.45\textwidth}{Positive polarity atom}}
                \State \Return{$\mathcal{P}(l)$}
            \EndIf
            \State \Return{$\mathcal{P}(l) + N$}
            \Comment{\parbox{0.45\textwidth}{Negative polarity atom}}
        \EndFunction
        \Procedure{pruneSubsumption}{$\sideP, \mainP$}
            \If {$t + |\mainP| > $ \texttt{UINT\_MAX}}
            \Comment{\parbox{0.45\textwidth}{Reset before arithmetic overflow}}
                \State $\mathcal{A} \gets zeros(2\cdot N)$ 
                \State $t \gets 0$
            \EndIf
            \For {$\mainL \in \mainP$}
                \State $idx \gets \Call{headerIndex}{\mainL}$
                \State $\mathcal{A}[idx] \gets \max (t, \mathcal{A}[idx]) + 1$
            \EndFor
            \For {$\sideL \in \sideP$}
                \State $idx \gets \Call{headerIndex}{\sideL}$
                \If {$\mathcal{A}[idx] \leq t$}
                    \State $t \gets t + |\mainP|$
                    \State \Return{$\top$}
                \EndIf
                \State $\mathcal{A}[idx] \gets \mathcal{A}[idx] - 1$
            \EndFor
            \State $t \gets t + |\mainP|$
            \State \Return{$\bot$}
        \EndProcedure
        \Procedure{pruneSubsumptionResolution}{$\sideP, \mainP$}
            \If {$t + 1 > $ \texttt{UINT\_MAX}}
            \Comment{\parbox{0.45\textwidth}{Reset before arithmetic overflow}}
                \State $\mathcal{A} \gets zeros(2\cdot N)$ 
                \State $t \gets 0$
            \EndIf
            \State $t \gets t + 1$
            \For {$\mainL \in \mainP$}
                \State $\mathcal{A}[\mathcal{P}(\mainL)] \gets t$
            \EndFor
            \For {$\sideL \in \sideP$}
                \If {$\mathcal{A}[\mathcal{P}(\sideL)] \neq t$}
                    \State \Return{$\top$}
                \EndIf
            \EndFor
            \State \Return{$\bot$}
        \EndProcedure
    \end{algorithmic}
\end{algorithm}

\paragraph{Pruning after building match sets.}
While our pruning criteria~\eqref{eq:prune-s}-\eqref{eq:prune-sr} are fast to compute, they do not reason about substitutions needed for subsumption (resolution). However, while building the match sets $\Pi(\sideP,\mainP)$, we may also detect unsat instances of subsumption and subsumption resolutions.
For example, let $\sideL_i = p(f(x))$ be a literal of~$\sideP$. If~$\mainP$ does not contain any literal of the form $p(f(\cdot))$, there is no unifying substitution between $(\sideP,\mainP)$. The non-existence of such substitutions would not necessarily be detected by \eqref{eq:prune-s}-\eqref{eq:prune-sr}, but could be recorded while building the substitution sets.\\
We therefore use the following additional pruning criteria for subsumption:
\begin{equation}
    \label{eq:prune-s-substitution}
    \forall i\exists j.\ \Sigma^+_{i,j} \neq \IncSubst
\end{equation}

\begin{theorem}[Substitution Sets for Pruning Subsumption] \label{thm:pruning-s-substitutions}
    Let
    \(
        \MS(\sideP, \mainP) = \left\{\left(b^\pm_{i,j}, \Sigma^\pm_{i,j}\right)\right\}
    \)
    be the match set of~$\sideP$ and~$\mainP$.
    If~\eqref{eq:prune-s-substitution} is unsat, then~$\sideP$ does not subsume~$\mainP$.
\end{theorem}
\begin{proof}
    Theorem~\ref{thm:subsumption-encoding-soundness-and-completeness} implies that, if $\sideP$ subsumes $\mainP$, then $\forall i \exists j.\ b^+_{i,j}$ and  $\forall i,j.\ b^+_{i,j} \Rightarrow \Sigma^+_{i,j} \subseteq \sigma$. Hence, $\forall i \exists j.\ \Sigma^+_{i,j} \subseteq \sigma$, which is equivalent to \eqref{eq:prune-s-substitution} since $\IncSubst \in \sigma \Rightarrow \bot$. Therefore, if $\sideP$ subsumes $\mainP$, then \eqref{eq:prune-s-substitution} is valid.
\end{proof}

A pruning criterion similar to~\eqref{eq:prune-s-substitution} can  be applied to subsumption resolution:
\begin{equation}
    \label{eq:prune-sr-substitution}
    \forall i\exists j.\ \Sigma^+_{i,j} \neq \IncSubst \lor \Sigma^-_{i,j} \neq \IncSubst
\end{equation}
\begin{theorem}[Substitution Sets for Pruning Subsumption Resolution] \label{thm:pruning-sr-substitutions}
     Let $\MS(\sideP, \mainP) = \left\{\left(b^\pm_{i,j}, \Sigma^\pm_{i,j}\right)\right\}$ be the match set of $\sideP, \mainP$. If  \eqref{eq:prune-sr-substitution} is unsat, then $\sideP$ and $\mainP$ are not side and main premises of subsumption resolution.
\end{theorem}
\begin{proof}
    Similarly to the proof of Theorem~\ref{thm:pruning-s-substitutions}, the \textbf{compatibility} and \textbf{completeness} constraints of $\DirEnc(\sideP, \mainP)$ imply \eqref{eq:prune-sr-substitution}. Based on Theorem~\ref{thm:direct-encoding-soundness-and-completeness}, if $\sideP,\mainP$ are side and main premises of subsumption resolution, then \eqref{eq:prune-sr-substitution} is valid.
\end{proof}

We remark that the pruning criterion \eqref{eq:prune-sr} is a special case of Theorem~\ref{thm:pruning-sr-substitutions}. Therefore, if \eqref{eq:prune-sr} is unsat, then \eqref{eq:prune-s-substitution} is also unsat and no subsumption resolution is possible. 
Furthermore, if there are no negative polarity substitutions, then the \textbf{existence} constraint of  $\DirEnc(\sideP, \mainP)$ does not hold. As such, a further pruning criterion for (unsat) subsumption resolution instances is: 
\begin{equation}
    \label{eq:prune-sr-no-negative-substitution}
    \exists i, j.\ \Sigma^-_{i,j} \neq \IncSubst
\end{equation}

\begin{theorem}[Polarities for Pruning Subsumption Resolution]
    Let  $\MS(\sideP, \mainP) = \left\{\left(b^\pm_{i,j}, \Sigma^\pm_{i,j}\right)\right\}$ be the match set of $\sideP,\mainP$. If~\eqref{eq:prune-sr-no-negative-substitution} is unsat, then $\sideP$ and $\mainP$ are not premises of subsumption resolution.
\end{theorem}
\begin{proof}
    Based on Theorem~\ref{thm:subsumption-encoding-soundness-and-completeness}, if $\DirEnc(\sideP, \mainP)$ is satisfiable, then the \textbf{existence} property $\exists i,j.\ b^-_{i,j}$ is  satisfiable and the  \textbf{compatibility} constraint is satisfied by the same assignment.
    Therefore, if $\DirEnc(\sideP, \mainP)$ is satisfiable, then $\exists i,j.\ \Sigma^-_{i,j} \neq \IncSubst$.
\end{proof}

Finally, if there exist two literals in $\sideP$ such that they do not have positive matches to literals in $\mainP$ and the respective predicates of the literals are different, then subsumption resolution is not possible. This yields our final pruning criterion: 

\begin{equation}\label{eq:different-predicates-no-positive}
    \forall i, i'.\ (i \neq i') \Rightarrow 
             (\mathcal{P}(\sideL_i) = \mathcal{P}(\sideL_{i'})
        \lor \exists j\ \Sigma^+_{i,j} \neq \IncSubst
        \lor \exists j\ \Sigma^+_{i',j} \neq \IncSubst)
\end{equation}

\begin{theorem}[Predicate Matches for Pruning Subsumption Resolution]
    Let  $\MS(\sideP, \mainP) = \left\{\left(b^\pm_{i,j}, \Sigma^\pm_{i,j}\right)\right\}$ be the match set of $\sideP,\mainP$.  If \eqref{eq:different-predicates-no-positive} is unsat, then $\sideP$ and $\mainP$ are not side and main premises of subsumption resolution.
\end{theorem}

\begin{proof}
    By contradiction, assume that subsumption resolution could be applied to $(\sideP,\mainP)$. Then, there exists a unique $\mainL'$ such that $\sigma(\sideP') = \{\mainL'\}$. However, if~\eqref{eq:different-predicates-no-positive} is unsat, there exist two different literals in $\sideP$ that can only be mapped negatively to $\mainL'$ (or not at all). These literals have a different predicate, therefore they cannot be both matched to the same literals ($\forall \sigma, l, l'.\ \mathcal{P} (l) \neq \mathcal{P} (l') \Rightarrow \sigma(l) \neq \sigma(l')$). If one of these literals cannot be matched to the resolution literal of \SR, and has no positive match, then it cannot be matched to any literal in $\mainP$; hence and subsumption resolution cannot be applied. 
\end{proof}

\begin{remark}
    It is easy to see that Algorithm~\ref{alg:pruning} is a very cheap procedure.
    During our experiments (Section~\ref{sec:experiments}, we observed that more than 95\,\% of instances of subsumption are filtered out by the pruning criterion \eqref{eq:prune-s} alone,
    and more than 50\,\% are also pruned by~\eqref{eq:prune-sr}.
    When it comes to subsumption resolution, in our experiments 90\,\% of subsumption resolution instances are pruned by~\eqref{eq:prune-sr}.
    The more restrictive nature of~\eqref{eq:prune-sr-substitution} and~\eqref{eq:prune-sr-no-negative-substitution} prunes an additional 5\,\% of subsumption resolution instances.
    As a result, our experiments show that pruning is indeed an important and cheap preprocessing step.
    Thanks to pruning, in our experiments only 5\,\% of subsumption (resolution) instances need to use more expensive SAT-based computation steps, using our SAT solver from Section~\ref{sec:sat-solver}.
\end{remark}

\section{SAT-Based Subsumption Variants in Saturation}\label{sec:loop-optimization}

In this section, we discuss the direct integration of the SAT solving engine of Section~\ref{sec:implementation} within the saturation loop of first-order theorem proving.
Such an integration greatly improves redundancy checking in theorem proving, without making significant changes to the underlining saturation algorithms of the prover.

\begin{algorithm}[t]
  \centering
  \caption{Forward simplification with SAT-based subsumption resolution}
  \label{alg:forward-new}
  \begin{algorithmic}
      \Procedure{ForwardSimplify}{$\mainP,F$}
        \State $\mainP^* \gets$ \NoSubsumptionResolution
        \For{$\sideP \in F \setminus \{\mainP\}$}
            \Comment{\parbox{0.35\textwidth}{Get candidates from generalisation index.}}
            \If{\Call{subsumption}{$\sideP, \mainP$} is \Subsumption}
            \Comment{\parbox{0.35\textwidth}{using Algorithm~\ref{alg:sat-subsumption}}}
                \State $F \gets F \setminus \{\mainP\}$
                \State \Return $\top$
                \Comment{\parbox{0.35\textwidth}{$\mainP$ is subsumed and removed}}
            \EndIf
            \If{$\mainP^*$ = \NoSubsumptionResolution}
              \State $M^* \gets \Call{subsumptionResolution}{\sideP, \mainP}$
              \Comment{\parbox{0.35\textwidth}{using Algorithm~\ref{alg:sat-subsumption-resolution-new}}}
            \EndIf
        \EndFor
        \If{$\mainP^* \neq \NoSubsumptionResolution$}
          \State $F \gets F \setminus \{\mainP\}\ \cup\ \{\mainP^*\}$ 
          \Comment{\parbox{0.35\textwidth}{$\mainP^*$ is the conclusion of subsumption resolution between $\sideP$ and $\mainP$}}
          \State \Return $\top$
        \EndIf
        \State \Return $\bot$
      \EndProcedure
  \end{algorithmic}
\end{algorithm}

To design a saturation algorithm, one important aspect is to understand how to organise redundancy elimination during proof search.
One common design principle in this respect comes with so-called \emph{given clause algorithms}~\cite{otter},
where inference selection is implemented using clause selection.
At each iteration of the algorithm, a clause from the proof search is selected and inferences are performed between this clause and previously selected clauses.
When a new clause is generated, this clause should only be kept if it is not redundant or it cannot be simplified by another existing clause;
we refer to such redundancy checks over a new clause as \emph{forward} redundancy, implementing \emph{forward simplification}.
On the other hand, a newly generated clause could make existing clauses in the search space redundant;
we call such redundancy checks with a new clause as \emph{backward} redundancy, implementing \emph{backward simplification}.

Using the SAT solver of Section~\ref{sec:implementation} for detecting subsumption (resolution) in saturation needs therefore to (i) address both forward and backward variants of subsumption and subsumption resolution, and (ii) organize proof search with these subsumption variants solved via SAT. In the rest of this section, we mainly focus on forward simplification via subsumption and subsumption resolution, and briefly discuss differences with respect to backward simplification.

\paragraph{Forward simplification.} Intuitively, as subsumption is a stronger inference than subsumption resolution, subsumption  should be performed first. As such, a standard {forward simplification} loop for subsumption (resolution) would be: 

\begin{enumerate}
    \item From a selected clause $\mainP$, search some subsumption candidate clauses $\{\sideP_k \mid k=1,\dots\}$ using a generalisation term index~\cite{handbook-indexing}; 
    \item For each clause in $\{\sideP_k \mid k=1,\dots\}$, check if $\sideP_k$ subsumes $\mainP$. If this is the case, then stop and remove $\mainP$ from the clause set.
    \item For each clause in $\{\sideP_k \mid k=1,\dots\}$, check if $\sideP_k$ can delete a literal from $\mainP$ using subsumption resolution. If it is the case, then replace $\mainP$ by the conclusion of subsumption resolution \SR{} and stop.
\end{enumerate}

\begin{algorithm}[t]
  \caption{SAT-based subsumption in saturation}
  \label{alg:sat-subsumption}
  \begin{algorithmic}
      \State $f_{\texttt{S}} \gets \bot$ 
      \Comment{\parbox{0.45\textwidth}{If $f_{\texttt{S}}$ gets true then subsumption is guaranteed to fail}}
      \State $f_{\SR{}} \gets \bot$
      \Comment{\parbox{0.45\textwidth}{If $f_{\SR{}}$ gets true then subsumption resolution is guaranteed to fail}}
      \Procedure{Subsumption}{$\sideP,\mainP$}
          \State $f_{\SR{}} \gets \Call{pruneSubsumptionResolution}{\sideP,\mainP}$
          \State $f_{\texttt{S}} \gets f_{\SR{}} \lor \Call{pruneSubsumption}{\sideP,\mainP}$
          \State \Comment{\parbox{0.45\textwidth}{Corollary~\ref{cor:pruning-sr-weaker} ensures that $f_{\texttt{SR}} \Rightarrow f_{\texttt{S}}$.}}
          \If {$f_{\texttt{SR}}$}
            \State \Return{\NoSubsumption}
            \State \Comment {\parbox{0.45\textwidth}{If only subsumption fails, we still need to fill the match set.}}
          \EndIf
          \State $\MS \gets \MS(\sideP,\mainP)$
          \State $f_{\texttt{S}} \gets f_{\texttt{S}} \lor \neg \eqref{eq:prune-s-substitution}$
          \Comment{\parbox{0.45\textwidth}{Computed when filling $\MS$}}
          \State $f_{\texttt{SR}} \gets f_{\texttt{SR}}
                                        \lor \neg \eqref{eq:prune-sr-substitution}
                                        \lor \neg \eqref{eq:prune-sr-no-negative-substitution}
                                        \lor \neg \eqref{eq:different-predicates-no-positive}$
          \Comment{\parbox{0.45\textwidth}{Also computed when filling $\MS$}}
          \If {$f_{\texttt{S}}$}
              \State \Return{\NoSubsumption}
          \EndIf
          \State $\SEnc \gets \Call{encodeConstraints}{\MS}$
          \If {$\SEnc \models \bot$} 
          \Comment{\parbox{0.45\textwidth}{SAT solver returns unsatisfiable}}
            \State \Return{\NoSubsumption}
          \Else
            \State \Return{\Subsumption}
          \EndIf
      \EndProcedure
  \end{algorithmic}
\end{algorithm}

In this approach, finding the substitutions of subsumption (resolution) comes with a significant computation burden.
Further, as subsumption checks do not succeed most of the time, the match sets $\Pi(\sideP_k,\mainP)$ must be cached or recomputed. Therefore, when integrating our SAT-based solving of subsumption (resolution) in saturation using Algorithm~\ref{alg:forward-new}, we use  pruning-based preprocessing and  build match sets before checking subsumption and subsumption resolution. Our Algorithm~\ref{alg:forward-new} yields thus a new, SAT-based  forward simplification loop for subsumption (resolution) in saturation.
Algorithm~\ref{alg:forward-new} uses  Algorithm~\ref{alg:sat-subsumption} to possibly prune both subsumption and subsumption resolution and then set up a complete match set. Even though subsumption alone does not require the negative polarity substitutions, these substitutions are  computed for subsumption resolution. Then, Algorithm~\ref{alg:sat-subsumption-resolution-new} benefits from the work done by subsumption, since it only requires to create the propositional clause set.

\begin{remark}
    In Algorithm~\ref{alg:forward-new}, when a subsumption resolution check was successful, no other is performed, but the algorithm still searches for a subsumption. In this case, only a partial match set is necessary and subsumption will not fill negative polarity matches.

    The index used to provide candidate clauses returns clauses on a literal by literal manner. That is, for each literal $\mainL \in \mainP$, the index returns clauses that have at least one literal that is a generalisation of $\mainL$. However, for subsumption resolution, we also get clauses with a generalisation of complemented literals $\neg \mainL_j$. In this case, we do not need to check for subsumption, and only subsumption resolution is performed. Yet, subsumption resolution still sets up the match set.
\end{remark}

\begin{algorithm}[t]
  \caption{SAT-based subsumption resolution in saturation  \\
  {\color{white} for bit empty space only to get}-- with subsumption already set up via Algorithm~\ref{alg:sat-subsumption}}
  \label{alg:sat-subsumption-resolution-new}
  \begin{algorithmic}
    \Procedure{SubsumptionResolution}{$\sideP,\mainP$}
      \State \Comment{\parbox{0.45\textwidth}{upon Algorithm~\ref{alg:sat-subsumption} failing to subsume}}
      \State \Comment{\parbox{0.45\textwidth}{the match set $\MS$ is already set up}}
      \If {$f_{\SR{}}$}
        \State \Return{\NoSubsumptionResolution}
      \EndIf
      \State $enc \gets \Call{chooseEncoding}{\MS, \sideP, \mainP}$ 
      \Comment{\parbox{0.45\textwidth}{choose the best encoding (see Sect~\ref{sec:heuristics})}}
      \State $\SREnc \gets \Call{encodeConstraints}{enc, \MS}$
      \If {$\exists \tau\ldotp\ \tau \models \SREnc$}
      \Comment {\parbox{0.45\textwidth}{$\tau$ is a model of $\SREnc$ found by the solver}}
      \State \Return{$\Call{buildConclusion}{\tau, \mainP}$}\Comment{\parbox{0.45\textwidth}{conclusion of subsumption resolution}}
      \EndIf
      \State \Return{\NoSubsumptionResolution}
    \EndProcedure
  \end{algorithmic}
\end{algorithm}

\paragraph{Backward simplification.} Backward simplifications  use newly generated clauses $\sideP$ to simplify the current clause set $F$.  Given a newly generated clause $\sideP$, backward subsumption (resolution) thus checks whether $\sideP$ subsumes some clauses $\mainP \in F$ (or can remove a literal from $\mainP$). In this case, performing subsumption resolution right after subsumption is almost free. Indeed, since {backward simplifications} do not stop after simplifying one clause, the only cost of performing subsumption resolution right after subsumption is to setup the full match set, rather than simply setting up the positive polarity matches.

\smallskip 

\paragraph{Extensions of subsumption variants in saturation}\label{sec:SAT-saturation}\label{sec:symmetry-of-equality}
Our SAT-based approach  for solving subsumption (resolution) in saturation is very flexible. Indeed, the SAT solver can handle different types of matches to the same literal pair, yielding further extensions of the standard subsumption and subsumption resolution framework.

In the case of \emph{symmetric predicates},
such as equality, two different substitutions are possible. Consider the literals $\sideL_i := x = y$ and $\mainL_j := c = f(c)$. To match these two literals, one can either use the substitution $\{x\mapsto c, y\mapsto f(c)\}$ or $\{x\mapsto f(c), y\mapsto c\}$.
In this case, both substitutions would be added to the match set $\Pi(\sideL_i,\mainL_j)$ of $\sideL_i,\mainL_j$. That is, the matches $(\{x\mapsto c, y\mapsto f(c)\}, +, b^+_{i,j})$ and $(\{x\mapsto f(c), y\mapsto c\}, +, b'^+_{i,j})$ are added to $\Pi(\sideL_i,\mainL_j)$. In our implementation of the match set, it provides a list of matches $(b^\pm_{i,j}, \Sigma^\pm_{i,j})$ with either $i$ or $j$ fixed. When enumerating over this list to build the clauses, we ignore the second index. If several variables have the same index $(i,j)$, the system will not be broken. Therefore, even when adding more than one match to the same literal pair, the SAT encoding remains the same.
In addition, both substitutions are distinct, since otherwise one of the literals of $\sideL_i$ or $\mainL_j$ is a tautology and the respective clause would be removed. Handling of symmetric predicates brings great practical improvements, see Remark~\ref{rem:sparsity}.

In the case of subsumption resolution, one may use the \emph{most general unifier} on the resolution literal $\mainL'$,
if the variable set of $\mainL'$ is disjoint from the variables in $\mainP \setminus \{\mainL'\}$.
However, within the  splitting approach of the AVATAR framework~\cite{Voronkov:2014:Avatar} of first-order theorem proving, the prover would split upon the main premise $\mainP$;
hence, using most general unifiers on the literal $\mainL'$ of $\mainP$ would not be triggered.

%% file: 6-heuristics.tex
Section~\ref{sec:loop-optimization} introduced efficient algorithms for integrating SAT-based subsumption reasoning in saturation.
In this section, we  further improve our methods from
Section~\ref{sec:loop-optimization} by identifying and fine-tuning the key parameters of our SAT-based subsumption algorithm in saturation.
Doing so, we  (i) impose a solving timeout on particularly difficult subsumption and subsumption resolution instances (Section~\ref{sec:cutoff}),
and (ii) devise a framework for choosing the best SAT encodings for subsumption resolution (Section~\ref{sec:choosing-encoding}).

\subsection{Cutting off the SAT Search} \label{sec:cutoff}
We present  how to fine-tune a timeout strategy for our  SAT solver from Section~\ref{sec:implementation},
in order to prevent getting it stuck on unnecessary/difficult subsumption instances, while solving still as many positive instances as possible.

\subsubsection{Measuring SAT Solver Progress}
In general, the solver behaviour should be as deterministic as possible to ensure results are consistent and reproducible.
Elapsed wall-clock time depends on many factors such as the type of machine and current load,
and elapsed CPU time and number of CPU instructions easily change when refactoring code.
As such, these measures are unsuitable when a deterministic solver behavior, and respective progress measure, is expected.

For evaluating our SAT solving approach in saturation, we therefore follow the \textsc{Kissat} methodology~\cite{kissat}: we count the number of elapsed \emph{ticks},
which is a rough approximation of the number of memory cache lines accessed during
unit propagation and conflict analysis.

\begin{figure}[b]
    \begin{subfigure}{\textwidth}
        \includegraphics[width=\textwidth]{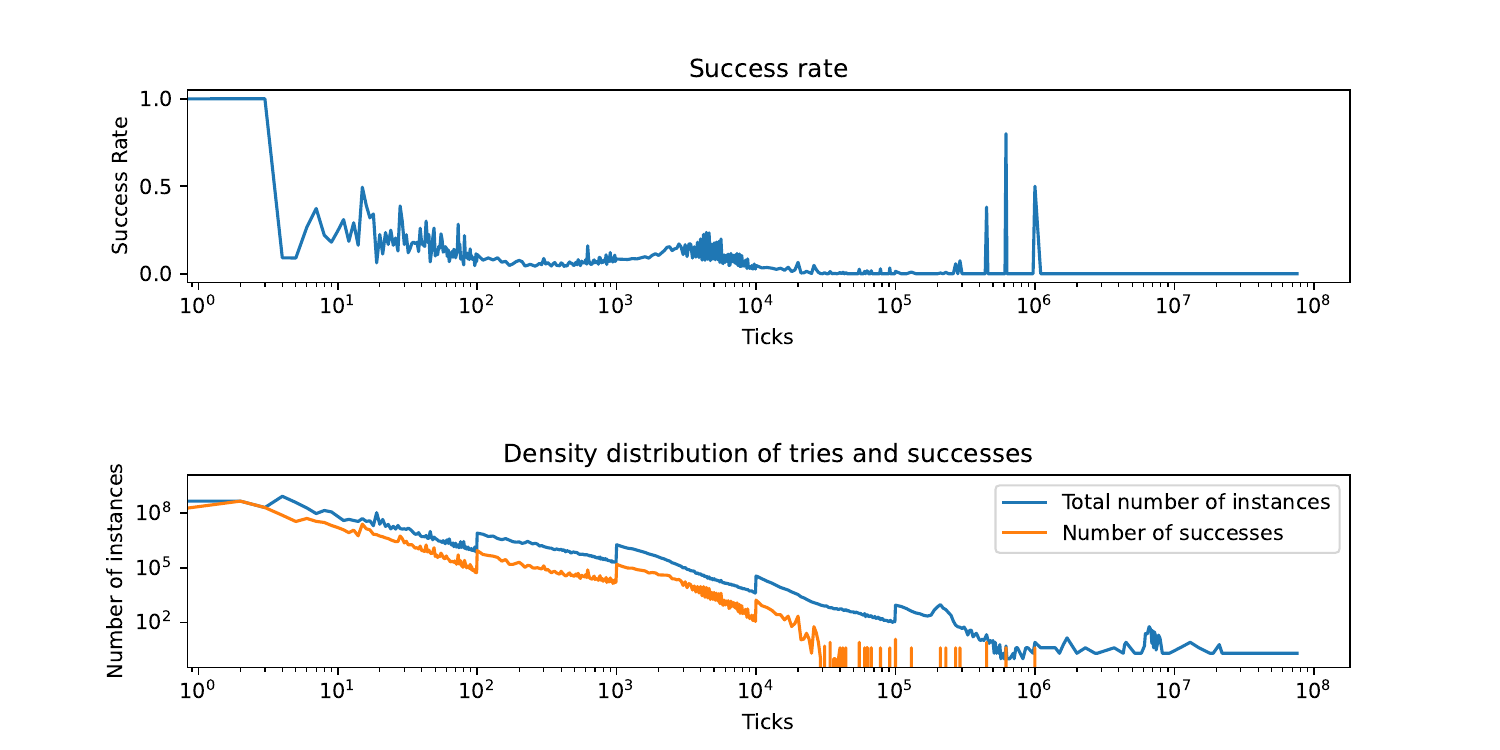}
        \caption{Solving subsumption instances.}
        \label{fig:subsumption-success-rate}
    \end{subfigure}
    \begin{subfigure}{\textwidth}
        \includegraphics[width=\textwidth]{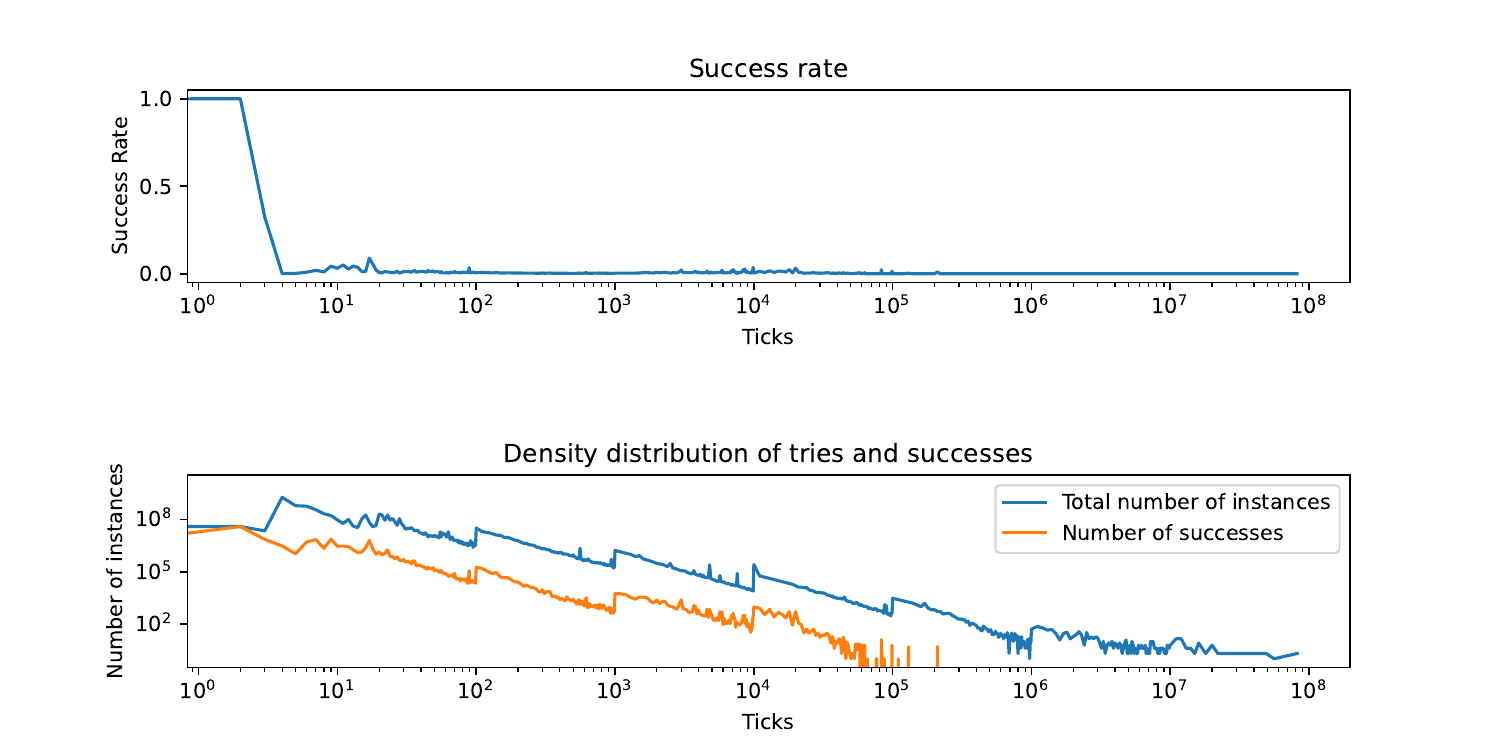}
        \caption{Solving subsumption resolution instances using a direct SAT encoding}
        \label{fig:sr-direct-success-rate}
    \end{subfigure}
    \caption{Success rates of the SAT solver depending on the number of ticks (ticks are displayed on the horizontal axes). The problems taking the longest time are less likely to succeed.\label{fig:cutoff}}
\end{figure}

\subsubsection{Empirical Observations}
 In our experiments (see Section~\ref{sec:experiments}), we evaluated our approach using the TPTP problem library~\cite{Sutcliffe:2017:TPTP}. Here, we logged the number of ticks the SAT solver performs on each problem and whether its search was successful. 
 Figure~\ref{fig:subsumption-success-rate} shows how the success rate of subsumption drops close to zero when our SAT solver runs longer. This effect is even more noticeable with subsumption resolution, as can be seen on Figure~\ref{fig:sr-direct-success-rate}.
We note that the performance jumps of Figures~\ref{fig:subsumption-success-rate}--\ref{fig:sr-direct-success-rate}
when crossing $10^k$ ticks are due to the non-linear scale used when aggregating data. We keep two significant digits to reduce the size of the files. Therefore, when jumping from $9.9\cdot 10^{k-1}$ to $1.0\cdot 10^{k}$, the size of the interval is multiplied by 10, hence a greater number of instances are gathered, and the line is discontinuous.

For improved solving progress, we aim to  estimate a good trade-off between losing solutions by stopping the search early and the number of ticks saved. 
To do so, (i)  we compute the number of ticks that the SAT solver has performed on instances that would be timed out;  (ii)  subtract the number of ticks ran before the timeout; and (iii) divide the result by the total number of ticks.
Figure~\ref{fig:lost-saved-tradoff} shows that, when using a cutoff of 150, less than 1\,\% of the successful instances are lost, while around 50\,\% of ticks are saved for subsumption and 35\,\% for subsumption resolution. Interestingly, when using a  cutoff of 5000, we loose less than 0.01\,\% of problems while still saving 10\,\% of ticks. 

\begin{figure}[t]
    \begin{subfigure}{0.30\textwidth}
        \centering
        \includegraphics[width=\textwidth]{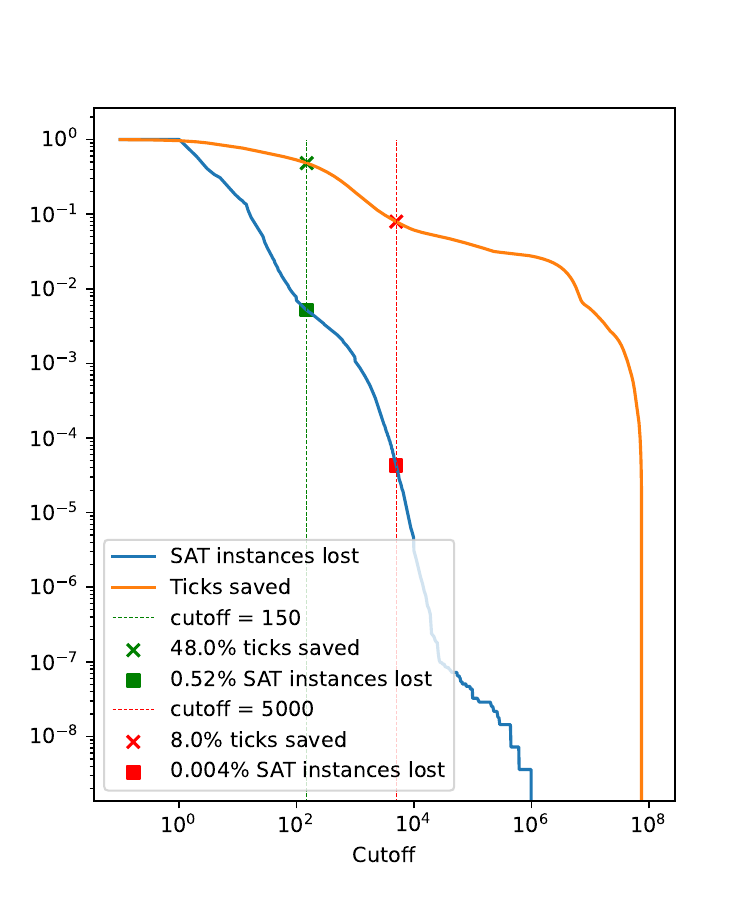}
        \caption{Subsumption}
    \end{subfigure}
    \begin{subfigure}{0.30\textwidth}
        \centering
        \includegraphics[width=\textwidth]{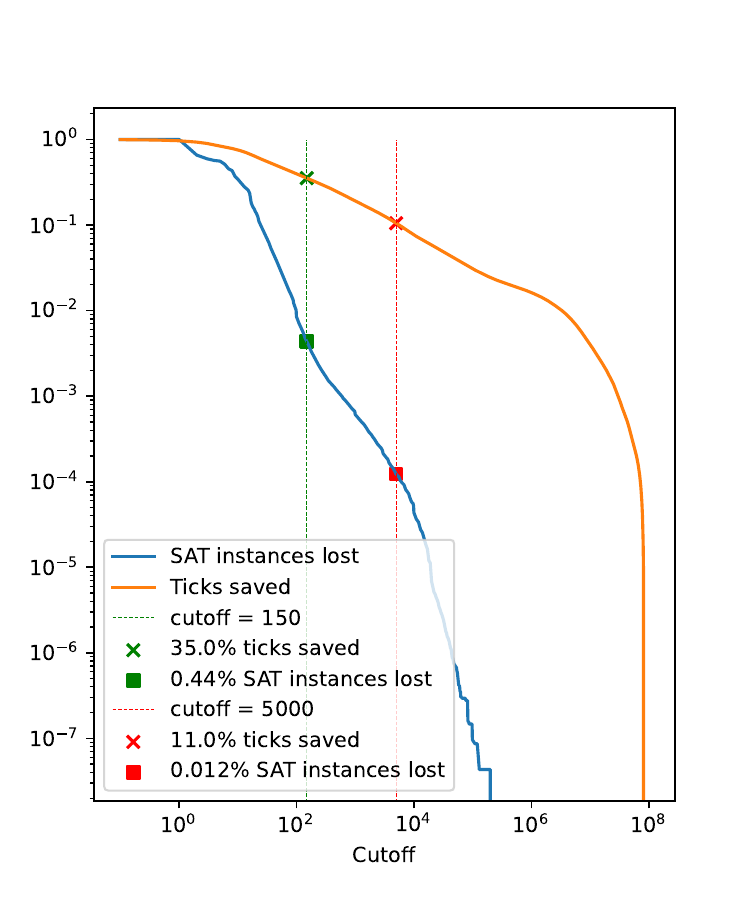}
        \caption{\SR{} direct encoding}
    \end{subfigure}
    \begin{subfigure}{0.30\textwidth}
        \centering
        \includegraphics[width=\textwidth]{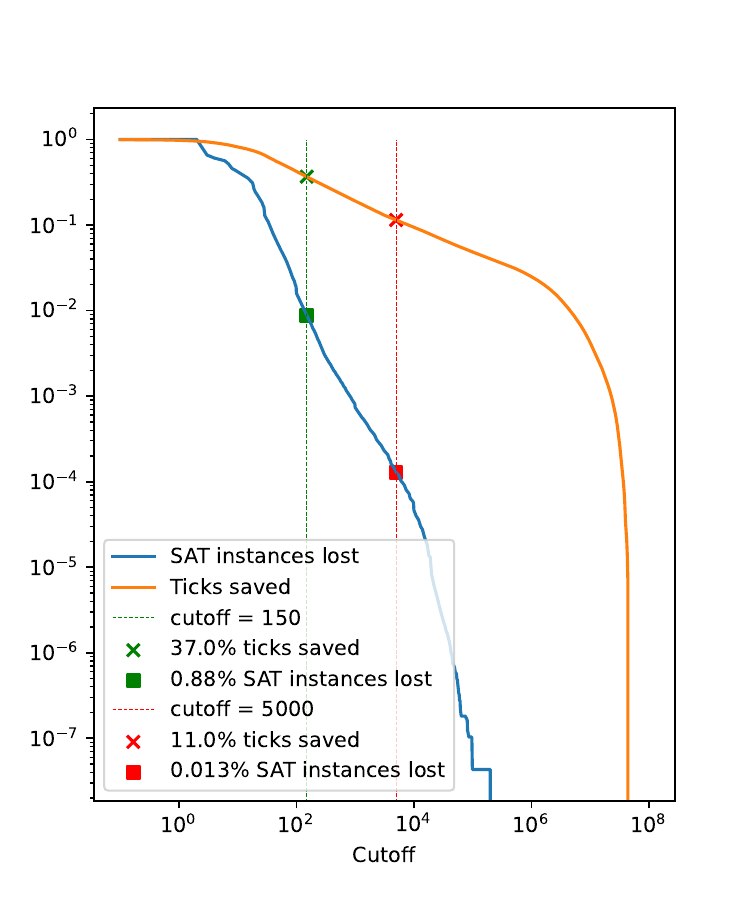}
        \caption{\SR{} indirect encoding}
    \end{subfigure}
    \caption{Trade-off between positive instances lost when cutting off, and the number of ticks saved.}
    \label{fig:lost-saved-tradoff}
\end{figure}

\subsection{Choosing SAT Encodings for Subsumption Resolution}\label{sec:choosing-encoding}
Section~\ref{sec:satconstraints} introduced two different encodings for subsumption resolution over $(\sideP,\mainP)$.
The direct encoding $\DirEnc(\sideP,\mainP)$ has $O(|\sideP|^2\cdot|\mainP|^2)$ clauses and  $O(|\sideP|\cdot|\mainP|)$ variables,  while the indirect encoding $\IndEnc(\sideP,\mainP)$ contains $O(|\sideP|\cdot|\mainP|)$ clauses with $O(|\sideP|\cdot(|\mainP|+1))$ variables. Intuitively, the direct encoding $\DirEnc(\sideP,\mainP)$ should to be more light weight and faster for smaller instances of subsumption resolution, whereas the  indirect encoding $\IndEnc(\sideP,\mainP)$ should scale better on harder instances. In this section, we present a procedure to choose which encoding to use for a given instance of subsumption resolution.

\subsubsection{Problem Setup}
\label{sec:formal-description}
We focus on the problem of \textbf{choosing SAT encodings of subsumption resolution}.
We approximate this problem via a random distribution $\mathcal{D}(y | x)$, where
\begin{compactitem}[$\bullet$\leftmargin=0em]
\item  input $x$, drawn from another distribution $\mathcal{X}$, is a vector of features $x_1, \dots, x_n$;
\item output $y$ is a pair of values $(y_0, y_1)$, where $y_0$ is the encoding and SAT solving time of the direct encoding $\DirEnc(\sideP,\mainP)$ and $y_1$ is the encoding and SAT solving time of the indirect encoding $\IndEnc(\sideP,\mainP)$.
\end{compactitem}

\paragraph {Objective function.} Let a function family $\mathcal{F}$ be a set of functions $f: \mathbb{R}^n \rightarrow \{0, 1\}$. We define our objective function over $\mathcal{D}(y | x)$ and $\mathcal{X}$ as follows:
\begin{equation}
    \label{eq:objective-function}
    \arg \min_{f \in \mathcal{F}} \E\limits_{\substack{x \sim \mathcal{X}\\ (y_0, y_1)\sim \mathcal{D}(\cdot | x)}} \left[ y_{f(x)}\right]
\end{equation}
Intuitively, our objective function~\eqref{eq:objective-function} computes a classifier $f$ whose choice, given a set of features, minimises the expected run time of the respective SAT encoding and solving of subsumption resolution.

\paragraph {Features.} For any classification problem, identifying relevant features is important.
We chose the following features for our  classifier $f$ computed by~\eqref{eq:objective-function}:
\begin{enumerate}
  \item the number $n$ of literals of the main premise $\mainP$;
  \item the number $k$ of literals of the side premise $\sideP$;
  \item the ``sparsity'' of the match set $\MS(\sideP,\mainP)$, computed as : $\frac{|\MS|}{k \cdot n}$, where $|\MS|$ denotes the size of the match set $\MS(\sideP,\mainP)$.
  \end{enumerate}
  
The relevance of the respective lengths $k, n$ of the premises $\sideP, \mainP$ is fairly self-explanatory, as the numbers of clauses of both SAT encodings grow differently with the number of literals of $\sideP,\mainP$. The sparsity of the match set  $\MS(\sideP,\mainP)$ is a measure of how many matches are found between literals of the main and side premises $\mainP, \sideP$. Sparsity of the match set is  a good indicator of the difficulty of the subsumption resolution problem. Indeed, if the match set  $\MS(\sideP,\mainP)$ is very sparse, then the subsumption resolution problem is easy: there are few matches to consider and the purely propositional clauses are already very constrained. On the other hand, if the match set $\MS(\sideP,\mainP)$ is dense, then the subsumption resolution problem is hard.

\begin{remark}\label{rem:sparsity}
  The sparsity of the match set may be greater than $1$. Indeed, in practice, we perform matching modulo the symmetry of equality (see Section~\ref{sec:symmetry-of-equality}). In such cases, one could use more than one match for a given literal pair.
\end{remark}

\subsubsection{Model Architecture}
\label{sec:choosing-model-architecture}
The problem described in Section~\ref{sec:formal-description} is formalized as a classification problem in~\eqref{eq:objective-function}.
Indeed, given a set of features $x$, we classify our problems sample into one of two classes: using the direct encoding $\DirEnc(\sideP,\mainP)$ (class value $0$) or using the indirect encoding $\IndEnc(\sideP,\mainP)$ (class value $1$). For solving the problem of
Section~\ref{sec:formal-description}, we select the SAT encoding of subsumption resolution over $(\sideP, \mainP)$ that is likely to be solved the fastest way.
Our classification procedure should thus be fast to compute at runtime.
We therefore use a decision tree as our classifier, where  our decision tree is a set of \texttt{if then else} expressions.
We used the library \texttt{scikit-learn}~\cite{scikit-learn} to train our decision tree.

\subsubsection {Building the Dataset}
\label{sec:building-dataset}
\paragraph {Sampling.} We construct the set of samples $(x, y)$ from the distributions $\mathcal{X}$, $\mathcal{D}_0$ and $\mathcal{D}_1$, respectively corresponding to the subsumption resolution input $(\sideP,\mainP)$ drawn from the distribution $\mathcal{X}$ and their direct $\DirEnc(\sideP,\mainP)$ and indirect encodings $\IndEnc(\sideP,\mainP)$ modeling the distribution $\mathcal{D}_0(y_0\ |\ x)$ and $\mathcal{D}_1(y_1\ |\ x)$ respectively.
To do so, we recorded the saturation  running time of any subsumption resolution inference that reaches the SAT solving procedure.
Indeed, if the subsumption instance  is pruned, both encoding will behave exactly the same and the sample is irrelevant.
We also recorded the features of the subsumption resolution check, that is, the length of the premises, and the sparsity.
Each problem is run twice, once with the direct encoding~$\DirEnc(\sideP,\mainP)$, and another with the indirect encoding~$\IndEnc(\sideP,\mainP)$. As a result, we obtain  two sets of samples samples $(x, y)$, one for each encoding of subsumption resolution. We  pair these samples to form $(x, y_0, y_1)$.

\paragraph{Condensing the dataset.}
Decision trees cannot be trained online, nor with mini-batches.
Traditionally, when facing a large dataset, the classical method is to segment it into small batches, and train the model on randomly sampled batches~\cite{DBLP:books/daglib/0040158}.
However, this approach is not supported within the decision trees of the \texttt{scikit-learn} library.
We therefore build a new dataset by summing the run times of all the samples that have the same features.
That is, we build a new dataset $\mathcal{S}$ of  $(x, \hat{y}_0, \hat{y}_1)$ samples,
where $x$ describes the feature and $\hat{y}_0$ and $\hat{y}_1$ are the respective sums of the run times of all the samples that have the same features $x$.

\paragraph {Modified objective function.} We adjust our objective function to our new dataset $\mathcal{S}$, as follows: 
\begin{equation}
    \label{eq:objective-function-condensed}
    \arg \min_{f \in \mathcal{F}} \sum_{(x, \hat{y}_0, \hat{y}_1) \in \mathcal{S}} \left[|\hat{y}_0 - \hat{y}_1| * \left(f(x) - H(\hat{y}_0 - \hat{y}_1)\right)^2\right]
\end{equation}
where $H$ is the step function, i.e., $H(a) = 1$ if $a \geq 0$, and $H(a) = 0$ otherwise. 

The optimisation problem~\eqref{eq:objective-function-condensed} is an empirical version of~\eqref{eq:objective-function}. 
Intuitively,~\eqref{eq:objective-function-condensed} introduces more weight to samples with a large difference of efficiency between both SAT encodings $|\hat{y}_0 - \hat{y}_1|$. A choice of a SAT encoding of subsumption resolution is considered ``wrong" if  (i) $f(x)$ predicted $0$  and the indirect encoding $\IndEnc(\sideP,\mainP)$ is faster than the direct encoding $\DirEnc(\sideP,\mainP)$;
or (ii) $f(x)$ predicted $1$  and the indirect encoding $\IndEnc(\sideP,\mainP)$ is slower than the direct encoding $\DirEnc(\sideP,\mainP)$. That is,  $f(x) - H(\hat{y}_0 - \hat{y}_1)$ is $1$ or $-1$ on wrong choices of SAT encodings, and $0$ on correct choices.

\paragraph {Evaluating the model.} We introduce a metric called the \emph{advantage} of the model over a function to evaluate the performance of our classifier $f$ from~\eqref{eq:objective-function-condensed}. We introduce three baseline classifiers to compare our model to:
\begin{enumerate}
  \item the \emph{direct encoding} $d(x) = 0$ always chooses the direct encoding $\DirEnc(\sideP,\mainP)$ for sample $x$;
  \item the \emph{indirect encoding} $i(x) = 1$  always chooses the indirect encoding $\IndEnc(\sideP,\mainP)$ for  sample $x$; 
  \item the \emph{perfect model} $p_\mathcal{S}(x)$ always chooses the fastest encoding for  sample $x$, being defined as:
  \begin{equation}
    p_\mathcal{S}(x) = \begin{cases}
      0 & \text{if } \exists (x, \hat{y}_0, \hat{y}_1) \in \mathcal{S} \land \hat{y}_0 < \hat{y}_1 \\
      1 & \text{otherwise}
    \end{cases}
  \end{equation}
\end{enumerate}

We then set the \emph{advantage} of the model $f$ over a function $g$ on a dataset $\mathcal{S}$ as: 
\begin{equation}
  Adv(f, g, \mathcal{S}) = \frac
  {\sum_{(x, \hat{y}_0, \hat{y}_1) \in \mathcal{S}} \left[\hat{y}_{g(x)}\right]}
  {\sum_{(x, \hat{y}_0, \hat{y}_1) \in \mathcal{S}} \left[\hat{y}_{f(x)}\right]}
\end{equation}

Naturally, the higher the advantage $Adv(f, g, \mathcal{S})$ is, the better the model $f$ performs. Note that  advantage over the perfect model is always less than or equal to $1$.

\begin{figure}[b]
  \centering
  \includegraphics[width=0.8\textwidth]{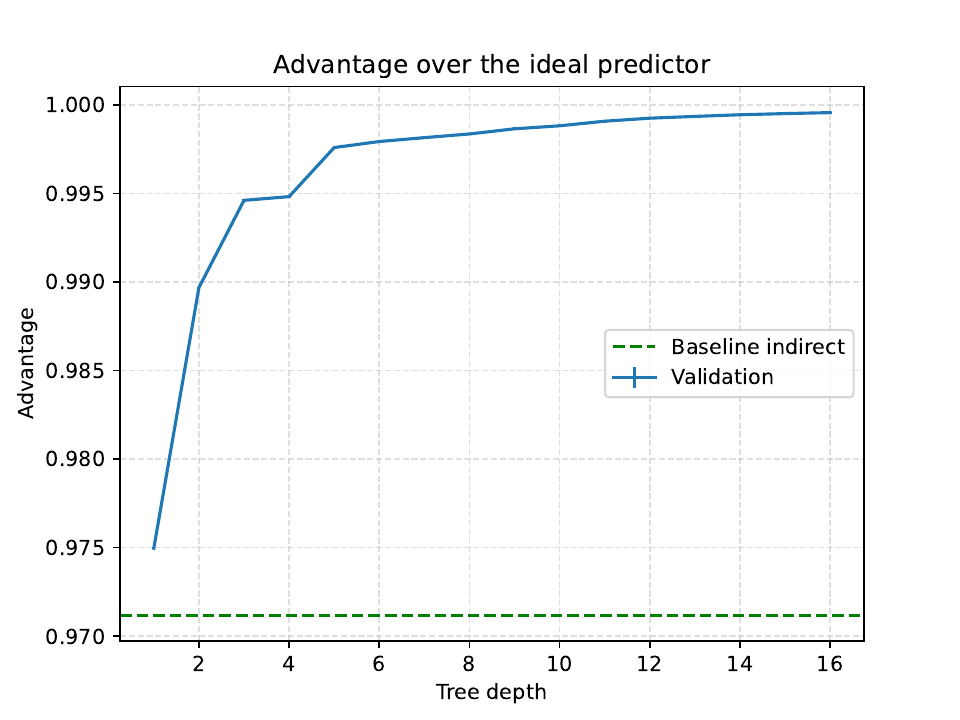}
  \caption{Advantage of the model over the perfect model for different depths. The green dashed line shows the baseline advantage of the indirect SAT encoding over the perfect model.}
  \label{fig:depth}
\end{figure}

\subsubsection{Choosing the Depth of the Decision Tree}%
\label{sec:choosing-depth}

\paragraph{Training, validation and test sets.}
We divided our dataset into a test set and a set of pairs of training and validation sets.
More precisely, we chose to segment our dataset $\mathcal{S}$ into $11$ segments,
namely $\mathcal{S}_0, \dots, \mathcal{S}_{10}$.
Here, $\mathcal{S}_0$ is kept for the final testing phase while
the remaining 10 segments of $\mathcal{S}$ are used
to generate pairs $(\mathcal{S}_i, \Union_{j \neq i} \mathcal{S}_j)$ for $i = 1, \dots, 10$.

\paragraph{Choosing the right depth.}
Decision trees have the ability to match arbitrary functions if they are deep enough and the training set is sufficiently large.
However, this is not desirable for two reasons: (i) the deeper the tree is, the more code will have to be added;
and (ii) the deeper the tree is, the more susceptible it is to overfitting.
We therefore need to find a proper depth for our decision tree.
To do so, for each $i=1,\dots,10$, we train a decision tree for each depth $d = 1, \dots, 15$ on the set $\Union_{j \neq i} \mathcal{S}_j$ and evaluate the performance on the validation set $\mathcal{S}_i$.
Figure~\ref{fig:depth} shows that the performance gains are mostly achieved by trees of depth lower than $3$. As such, we empirically chose to use a decision tree of depth 3 in our framework.

\begin{figure}[!p]
  \begin{subfigure}{0.40\textwidth}
    \centering
    \includegraphics[width=0.95\textwidth,trim={1.5cm, 2.25cm, 1.5cm, 2.0cm},clip]{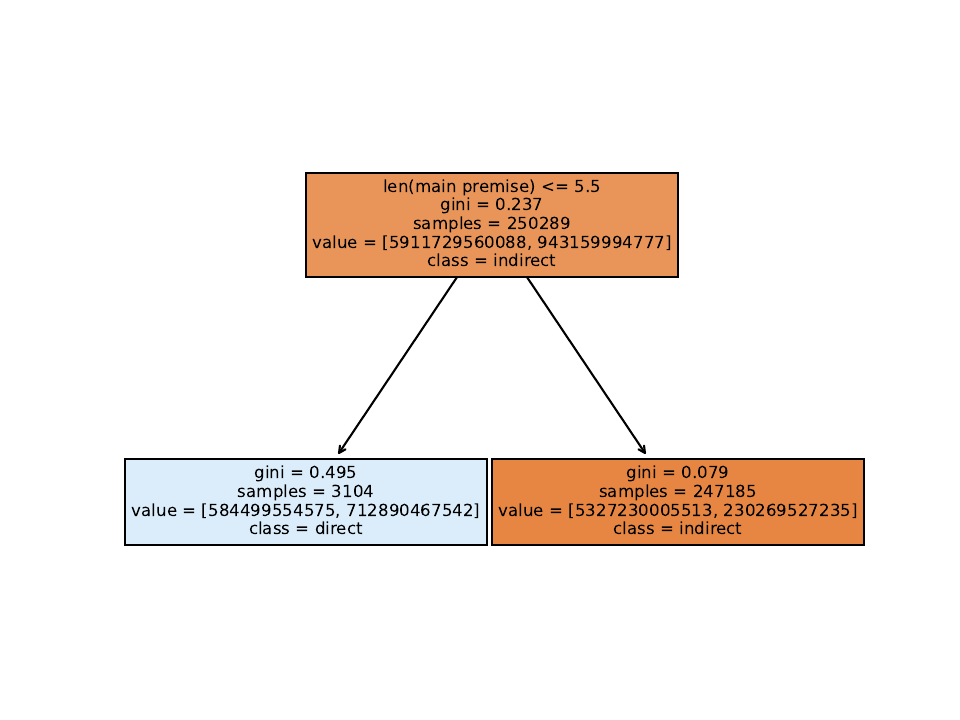}%
    \caption{Depth 1}
    \label{fig:tree-depth-1}
  \end{subfigure}
  \begin{subfigure}{0.60\textwidth}
    \centering
    \includegraphics[width=0.95\textwidth,trim={1.5cm, 2.25cm, 1.5cm, 2.0cm},clip]{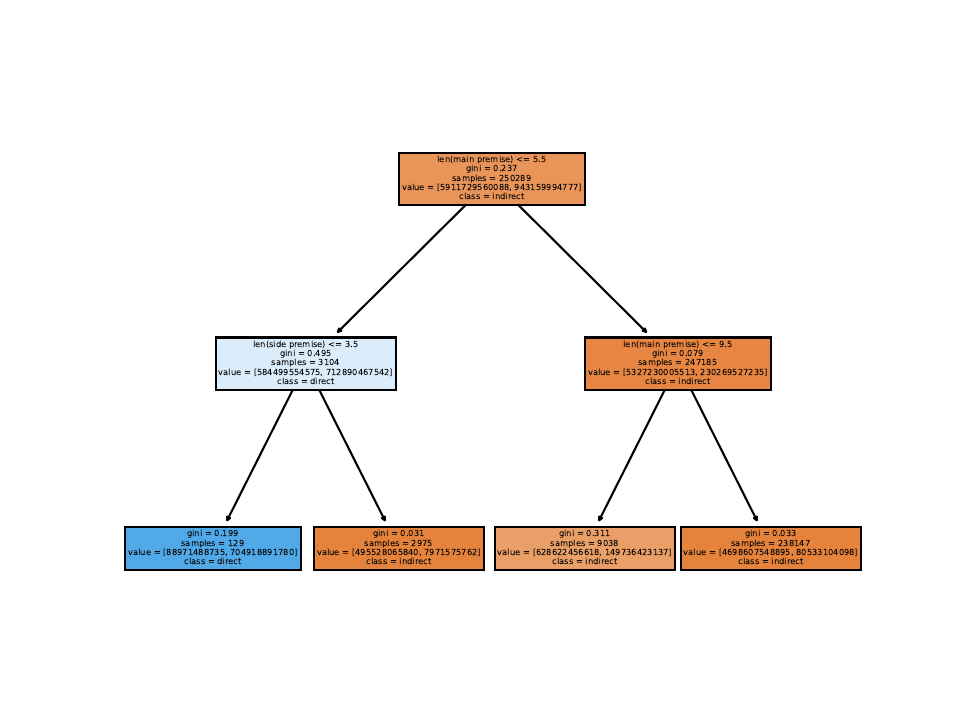}%
    \caption{Depth 2}
    \label{fig:tree-depth-2}
  \end{subfigure}
  \begin{subfigure}{\textwidth}
    \centering
    \includegraphics[width=0.8\textwidth,trim={1.5cm, 2.25cm, 1.5cm, 2.0cm},clip]{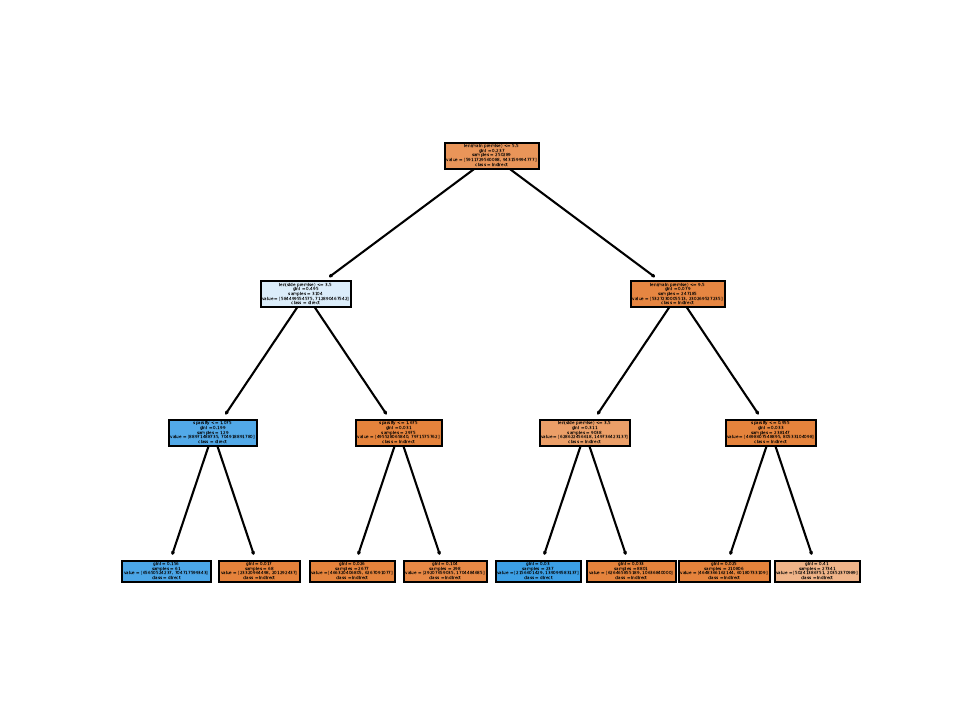}%
    \caption{Depth 3}
    \label{fig:tree-depth-3}
  \end{subfigure}
  \caption{Decision trees of different depths. The orange nodes choose the indirect encoding, and the blue nodes choose the direct encoding.}
\end{figure}

Figure~\ref{fig:tree-depth-3} summarizes the  decision tree resulting from the training. 
As only two leaves prefer the direct encoding, this tree can be summarised and optimised into the following pseudo-code on Figure~\ref{fig:tree-code}.
\begin{figure}[t]
    \begin{lstlisting}[language=Python]
      if (|S| <= 3):
        if (|M| <= 5 and  sparsity <= 1.075):
          return direct
        elif (|M| <= 9):
          return direct
      return indirect
    \end{lstlisting}
    \caption{Decision tree from Figure~\ref{fig:tree-depth-3} in pseudo-code.}
    \label{fig:tree-code}
\end{figure}

\paragraph{Evaluating the model.}
Once our decision  tree is trained, we can evaluate its performance on the test
set. Because of the very large dataset available and the limited number of
features, no overfitting was observed. The predictor has an advantage of 1.024
over the indirect encoding, and 0.995 over the ideal predictor. This method
could be further improved by adding new encodings, or by increasing the feature
space.

%% file: 7-experiments.tex
We implemented our SAT-based framework for solving subsumption and subsumption resolution
in the \vampire{} theorem prover~\cite{kovacs2013first}. We next discuss the evaluation and results of our approach. 

\paragraph{Benchmarks.}
We use the TPTP library~\cite{Sutcliffe:2017:TPTP} (version 8.1.2) as the benchmark source for our experiments. 
This version of the TPTP library contains altogether~25,257~problems in various languages. Out of these examples, 24,973~problems have been included in our evaluation of
SAT-based subsumption and subsumption resolution in \vampire{}.
The remaining TPTP problems that we did not use for our experiments
requires features that \vampire{} currently does not support (e.g.,
higher-order logic with theories).

\paragraph{Experimental Setup.}
All our experiments were carried out on a cluster at TU Wien,
where each compute node contains two AMD Epyc 7502 processors,
each of which has 32~CPU cores running at 2.5\,GHz.
Each compute node is equipped with 1008\,GiB of physical memory
that is split into eight memory nodes of 126\,GiB each,
with eight logical CPUs assigned to each node.
We used the \texttt{runexec} engine from the
benchmarking framework~\textsc{BenchExec}~\cite{BeyerLoweWendler:2017:benchexec}
to assign each benchmark process to a different CPU core and its corresponding memory node,
aiming to balance the load across memory nodes.
Further, we used \textsc{GNU Parallel}~\cite{Tange:2018:GNUParallel}
to schedule 32~benchmark processes in parallel.

\paragraph{Ensuring consistent progress.}
For several of the subsequent experiments,
we perform relatively expensive computation and/or logging
in addition to the measured solving process.
While this instrumentation does not affect the measurements per se,
it will reduce the progress the solver can make in the saturation algorithm
within a fixed duration of wall-clock time.
To avoid this effect,
we first performed a run of \vampire{} without any expensive instrumentation
and a time limit of 60~seconds,
and report for each TPTP problem the number of times
the forward simplification loop has been called.
For all subsequent \vampire{} runs that involve instrumentation,
we do not impose a time limit, but instead terminate after performing the
previously reported number of forward simplification loops.

\subsection{Measuring Speed Improvements for Subsumption\label{sec:slog}}
We  first measured the cost of subsumption checks in isolation.
A similar evaluation has previously been done for indexing techniques in
first-order provers, see~\cite{voronkov-evaluation-indexing}.

\paragraph{Methods Considered.} We first ran \vampire{} with a timeout of 60~seconds
on each TPTP problem, while logging each subsumption check into a file.
Each of these files then contains a sequence of subsumption checks,
which we call the \emph{subsumption log} for a problem.
This preparatory step led to a large number of benchmarks that are representative
for subsumption checks that appear during actual proof search.
These benchmarks occupy 1.79\,TiB of disk space in compressed form,
and contain about 278~billion subsumption checks in total.
About 0.6\,\% of these subsumption checks are satisfiable (1.7~billion),
while the rest is unsatisfiable.
We note that we removed 5530~TPTP problems from this experiment,
because \vampire{} was unable to parse back the output it generated during the
logging phase.
However, the successfully replayed subsumptions amount to about 258~billion
subsumption checks (93\,\% of the collected).

Next, we executed the checks listed in each subsumption log and measured the total running times,
once for the existing backtracking-based subsumption algorithm of~\vampire{},
and once for our SAT-based subsumption approach in~\vampire{}.

\paragraph{Results.}
The results of this experiment are given
in Figure~\ref{fig:slog-scatter}
and Table~\ref{tab:slog-replay}.
Each mark in Figure~\ref{fig:slog-scatter} represents
one subsumption log from a TPTP problem,
and compares the total running time of executing
all subsumption checks contained in the log
with the old backtracking-based algorithm
vs. the new SAT-based algorithm.
The dashed line indicates equal runtime,
hence, our SAT-based approach was faster for marks below the line.
In Table~\ref{tab:slog-replay},
we give the cumulative time used for subsumption.
For the six TPTP problems
\texttt{LCL673+1.015},
\texttt{LCL673+1.020},
\texttt{NLP023+1},
\texttt{NLP023-1},
\texttt{NLP024+1},
and
\texttt{NLP024-1},
the old backtracking-based subsumption algorithm of \vampire{} 
did not terminate within a time limit of 1200\,s;
these problems are not included in the cumulative sum.

Overall, our results show a clear improvement
of the running time of subsumption in \vampire{}, yielding an improvement  
by a factor of 2.5.

\newfloatcommand{capbtabbox}{table}[][\FBwidth]

\begin{figure}
    \begin{floatrow}
        \capbtabbox{%
  \begin{tabular}{l | c l }
    \hline
    \textbf{Prover} & \textbf{Subsumption}  & \textbf{Boost} \\
    \hline
    \upstrut
    \textsc{Vampire}$_\textsc{M}$            & $35.86\,h$ & 
     \\
    \textsc{Vampire}$_\textsc{Sat}$ & $13.68\,h$ & 2.62\,x 
    \downstrut \\
    \hline
  \end{tabular}
        }{%
        \caption{%
        \label{tab:slog-replay}
    Total time spent on subsumption checks, summed over 19437~TPTP~problems.
    Note that \textsc{Vampire}$_M$ timed out on 6 problems during subsumption replay;
    these have not been included in the total.
        }%
        }
        \ffigbox{%
  \includegraphics[width=0.5\textwidth]{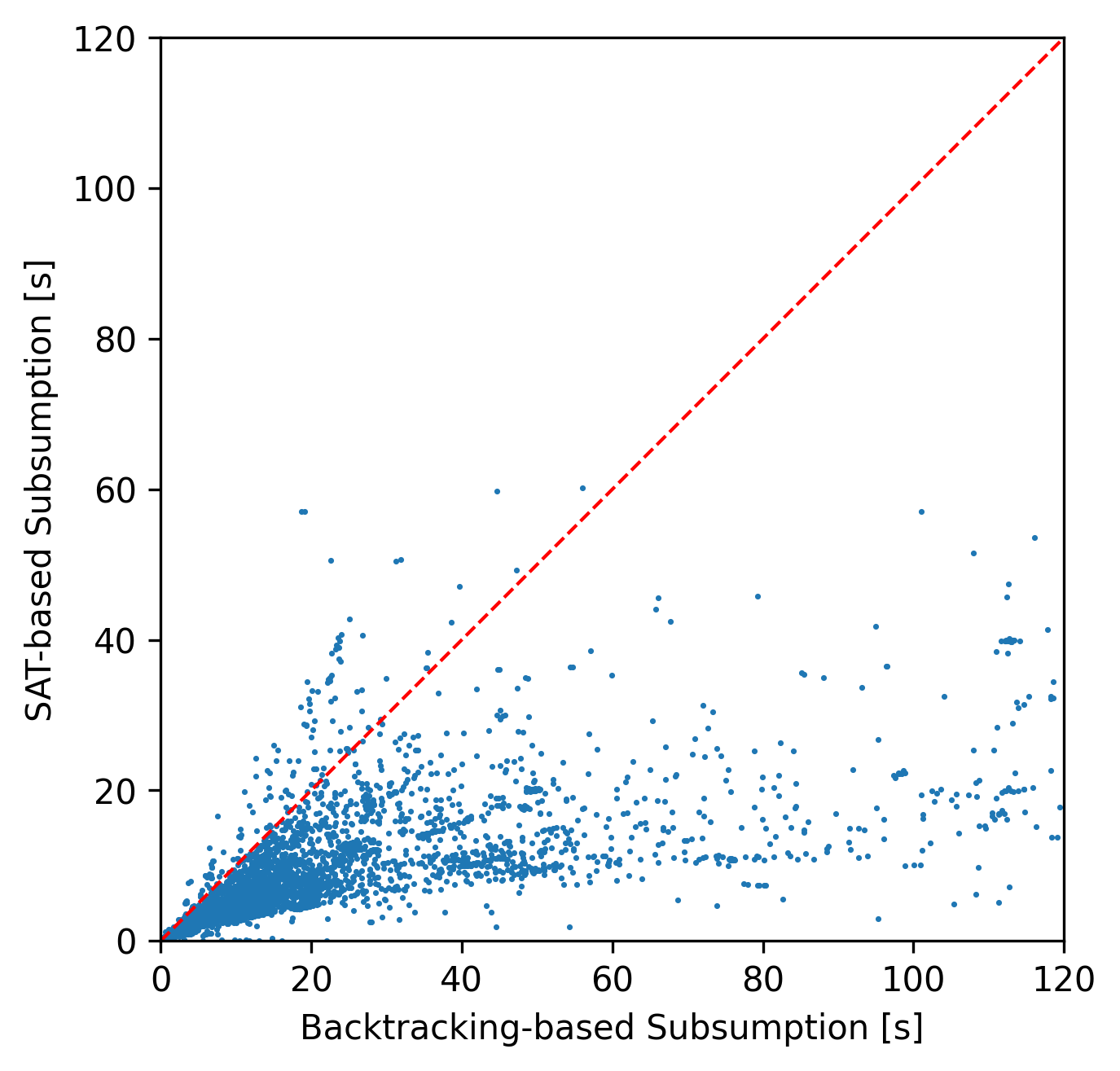}%
        }{%
        \caption{%
            \label{fig:slog-scatter}
            Total running time (in seconds)
            of backtracking-based vs. SAT-based subsumption,
            where each mark represents a TPTP problem.
            For marks below the dashed line,
            our SAT-based approach was faster.
        }%
        }
    \end{floatrow}
\end{figure}

\subsection{Measuring Speed Improvements for Subsumption Resolution}
Whereas subsumption instances can be separated from the rest of the execution, efficient subsumption resolution cannot. As explained in Section~\ref{sec:implementation}, subsumption resolution is applied in a simplification loop that optimised the setup between subsumption and subsumption resolution. It would thus be an unfair comparison to isolate subsumption resolution from subsumption. This is why we decided to measure the runtime of subsumption and subsumption resolution together in a forward simplification loop. 

Our experimental procedure  is summarized in Algorithm~\ref{alg:evaluation} and commented below. 

\begin{algorithm}[t]
  \caption{Evaluation of SAT-based subsumption resolution}
  \label{alg:evaluation}
  \begin{algorithmic}
    \Procedure{ForwardSimplifyWrapper}{$M, F$}
      \State $s \gets \mathrm{StartTimer}()$
      \State $r \gets \Call{ForwardSimplify}{M, F}$ \Comment{\parbox{0.45\textwidth}{Benchmarked method}}
      \State \Comment{\parbox{0.45\textwidth}{Prevent modification of $F$}}
      \State $e \gets \mathrm{EndTimer}()$
      \State $\mathrm{WriteToFile}(e - s)$
      \State $r' \gets \Call{Oracle}{M, F}$
      \State $\Call{CheckCoherence}{r, r'}$  \Comment{\parbox{0.45\textwidth}{Empiric check}}
      \State \Return{$r'$}
    \EndProcedure
  \end{algorithmic}
\end{algorithm}

\begin{compactitem}[$\bullet$\leftmargin=0em]
  \item The conclusion clause of the subsumption resolution rule \SR{} is not necessarily unique. Therefore,  different versions of subsumption resolution, including our work based on direct and indirect SAT encodings, may not return the same conclusion clause of \SR{}. Hence, applying different versions of subsumption resolution over the same clauses  may change the saturation  process.
  \item Saturation with our SAT-based subsumption resolution takes advantage of subsumption checking (see Algorithms~\ref{alg:forward-new}--\ref{alg:sat-subsumption-resolution-new}). Therefore, only checking subsumption resolution  on pairs of clauses is not a fair nor viable comparison, as isolating subsumption checks from subsumption resolution is not what we aimed for (due to efficiency). 
  \item CPU cache influences results. For example, two consecutive runs of Algorithm~\ref{alg:forward-new} may be up to 25\,\% faster on second execution, due to cache effects.
  \item $\Call{CheckCoherence}{r, r'}$ is an empiric check that ensures that the result of the oracle is compatible with the result of the benchmarked method.
\end{compactitem}

\paragraph{Oracle.}
The oracle used in our experiments is the fastest method overall. The motivation of this choice is to maximise the number of sample points compared to the total computation time. Indeed, a slower oracle will prevent \vampire{} to progress faster. The oracle therefore runs the dynamic encoding (heuristic encoding selection) with loop optimisation.

\paragraph{Methods considered.}
We compared the following versions of \vampire:
\begin{itemize}
    \item \vampire$_M$ -- the master branch of \vampire{} (commit a47e1dca9), without SAT-based subsumption and subsumption resolution;
    \item \vampire$_D$ --  the SAT-based version of \vampire{} using the direct encoding $\DirEnc$;
    \item \vampire$_I$ -- the SAT-based version of \vampire{} using the indirect encoding $\IndEnc$;
    \item \vampire$_H$ -- the SAT-based version of \vampire{} using the heuristic discussed in Section~\ref{sec:choosing-encoding};
    \item \vampire$^*_\mathcal{E}$ -- using  the loop optimisation discussed in Section~\ref{sec:loop-optimization} with the encoding~$\mathcal{E}$ (note that the loop optimisation does not apply to the non-SAT version);
    \item \vampire-cutoff-$n$ -- uses a cutoff at $n$ ticks, as discussed in Section~\ref{sec:cutoff}.
\end{itemize}

\paragraph{Results.}
In Table~\ref{tab:results}, the average and standard deviation of the runtime of the forward simplification loop have been logged for the considered versions. The column \textbf{Boost} is the ratio between the average runtime of \vampire{}$_M$ and the method considered. From the table, we can see that the simplest version of our algorithm, that is, the direct encoding without loop optimisation, already performs better than the old backtracking-based algorithm.
Introducing the indirect encoding creates a large drop in variance, indicating that $\IndEnc$ is more stable and scalable.
The loop optimisation further improved performance by sharing work in the encoding setup.
Finally, choosing the encoding based on the heuristic discussed in Section~\ref{sec:choosing-encoding} brings another small improvement boost.
Overall, we obtained an improvement in performance by a factor of 1.36
on the simplification loop.

When considering these results with
our previous analysis of subsumption alone (Section~\ref{sec:slog}),
it is worth mentioning that they are not comparable.
While the evaluation method in Section~\ref{sec:slog}
allows a direct comparison of the backtracking-based subsumption implementation to the SAT-based approach,
such an evaluation is not suitable for subsumption resolution,
especially when considering the optimized simplification loop (Section~\ref{sec:loop-optimization}).
Indeed, our second benchmarking technique (this section) includes all components of the simplification loop, including obtaining candidate clauses.
In particular, this means we also measure improvements obtained by optimizing the simplification loop itself.

\begin{remark}
    In \cite{DBLP:conf/cade/CoutelierKRR23}, we observed a large drop in variance from the standard to the optimised simplification loop. Improving the memory usage of the pruning algorithms in Section~\ref{sec:pruning} greatly reduced this unexpected  behaviour. In~\cite{DBLP:conf/cade/CoutelierKRR23}, we used a standard C++ vector that was cleared between pruning runs. However, some problems in the TPTP library have very large signatures. On these instances, the subsumption execution time has been greatly impacted by the calls to our simpler pruning algorithm from~\cite{DBLP:conf/cade/CoutelierKRR23}. Namely, in the standard saturation loop,  pruning was executed once for subsumption and once for subsumption resolution. As discussed in Section~\ref{sec:implementation}, a large proportion of subsumption pruning checks are unnecessary if the subsumption resolution pruning criterion fails first.  Our fast implementations of pruning from Section~\ref{sec:pruning} greatly reduced this effect from~\cite{DBLP:conf/cade/CoutelierKRR23}.
\end{remark}

Figure~\ref{fig:cummulative-graph}
shows the cumulative number of forward simplification loops performed in less then $t\ \mu s$ for some methods. We can visually see that our method performs better than the previous implementation even for the easier instances, and further increasing its advantage on harder instances. The loop optimisation shows most its strength in the 10 to 20 $\mu s$ region before almost getting caught up by the non-optimised loop. The reason is that the optimisation only improves the polynomial setup of the algorithm, that becomes less relevant as the exponential nature of the problem takes over.

\begin{table}[t]
  \setlength{\tabcolsep}{6pt}
  \centering
  \begin{tabular}{l | c r | r}
    \hline
    \textbf{Prover} & \textbf{Average} & \textbf{Std.\ Dev.} & \textbf{Boost}\\
    \hline
    \upstrut
    \vampire$_M$ & 33.63 $\mu s$ & 1839.25 $\mu s$ \quad & 1.00 \quad \\
    \vampire$_D$ & 28.74 $\mu s$ & 1245.88 $\mu s$ \quad & 1.17 \quad \\
    \vampire$_I$ & 28.36 $\mu s$ & 243.38 $\mu s$ \quad & 1.19 \quad \\
    \vampire$_H$ & 28.16 $\mu s$ & 233.87 $\mu s$ \quad & 1.19 \quad \\
    \vampire$^*_D$ & 25.38 $\mu s$ & 1241.86 $\mu s$ \quad & 1.32 \quad \\
    \vampire$^*_I$ & 24.93 $\mu s$ & 196.38 $\mu s$ \quad & 1.35 \quad \\
    \vampire$^*_H$ & 24.73 $\mu s$ & 190.69 $\mu s$ \quad & 1.36
    \downstrut\quad \\
    \hline
  \end{tabular}\vspace*{0.5em}
  \caption{%
    Average time spent in the forward simplify loop.
    \textsc{Vampire}$^*_{H}$ is the fastest method, closely followed by the \textsc{Vampire}$^*_{I}$.
    The versions \vampire$^*_\mathcal{E}$ integrate the loop optimisation discussed in Section~\ref{sec:loop-optimization} into \vampire$_\mathcal{E}$.
  }
  \label{tab:results}
\end{table}

\begin{figure}[t]
  \centering
  \includegraphics[width=\textwidth, trim={2.25cm, 0cm, 2.25cm, 0cm}]{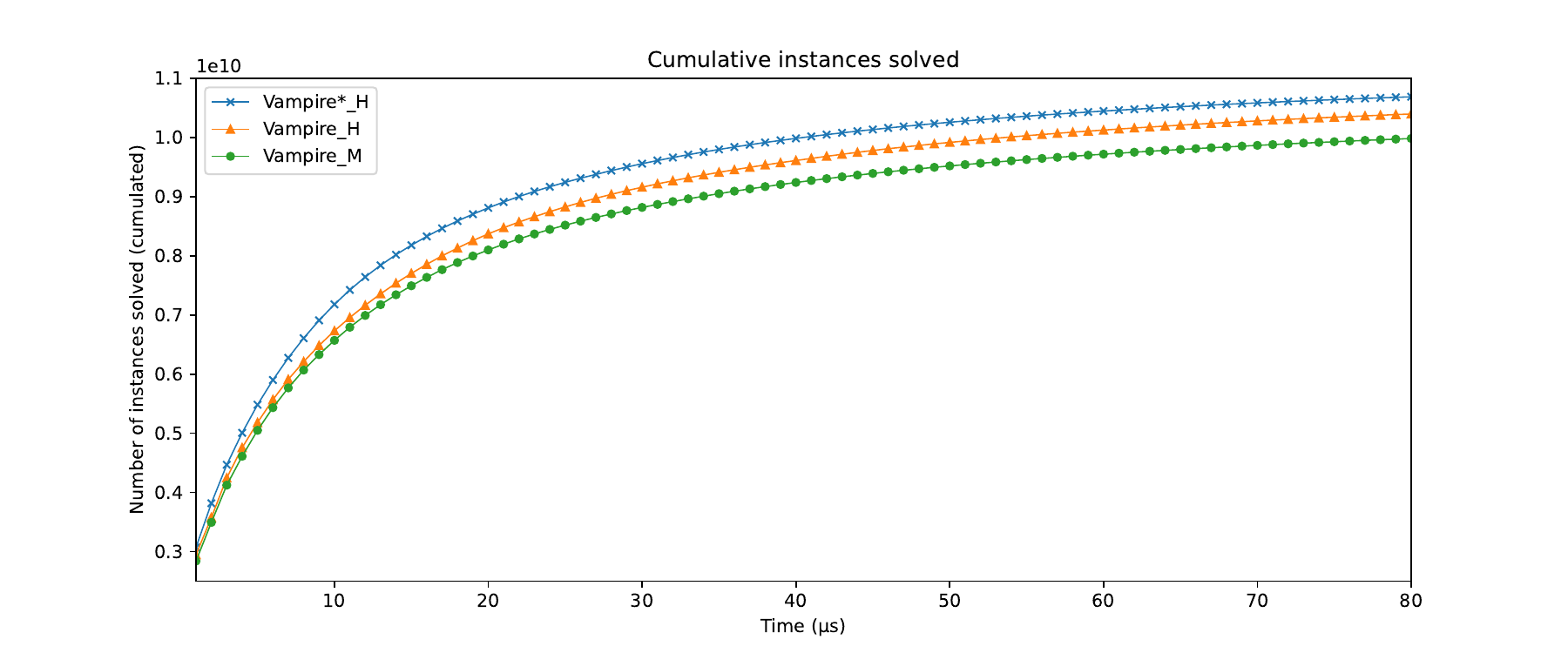}
  \caption{%
    Cumulative instances of applying subsumption resolution, using  the TPTP examples.
    A point $(t, n)$ on the graph means that $n$
    forward simplify loops were executed in less than $t~\mu s$.
    The higher the curve, the faster the \textsc{Vampire} version is. The difference between the different encoding being small relative to the difference the optimisation brings, we only displayed the dynamic encoding to avoid superposition of plot lines.
    }
  \label{fig:cummulative-graph}
\end{figure}

\subsection{Overall \vampire{} Runs} 
We finally analysed the 
 the number of problems \vampire{} solves depending on various implementations of the subsumption and subsumption resolution procedure.
 Table~\ref{tab:performances-vampire-vanilla} summarizes our findings and we draw the following conclusions.

\begin{compactitem}[$\bullet$\leftmargin=0em]
    \item
        Each SAT-based configuration of subsumption solves more problems than
        the previous, backtracking-based  implementation of subsumption,
        showing the superiority of our method in solving subsumption and
        subsumption resolution.
    \item
        Our heuristic approach using decision trees of
        Section~\ref{sec:heuristics} solves slightly more problems than the
        other SAT-based only methods of Section~\ref{sec:loop-optimization}.
        We remark that we trained our decision trees on a dataset built from
        the exact problems we are testing our methods against, with the purpose
        of maximizing the number of solved problems.
        We note that our methodology might suffer from (minimial) overfitting:
        we used a very rigid classification algorithm with a very low potential
        for overfitting. It is unlikely that a decision tree with such a low
        depth and few features will learn how to solve specific problems, but
        not learn general trends.
    \item
        Our cutoff method from Section~\ref{sec:cutoff} did not bring great
        improvements. While this result may sound discouraging, we believe it
        actually strengthens our  contributions from Section~\ref{sec:loop-optimization}.
        Indeed, it shows that only finding the simple subsumption and
        subsumption resolution instances is not an effective strategy. While
        our methods from Section~\ref{sec:loop-optimization} might not be the
        fastest for small clauses, they scale well with the complexity of the
        problem.
    \item
        The saturation loop optimisation techniques, e.g., forward simplifcation
        from Section~\ref{sec:loop-optimization}, bring the largest increase in
        number of problems solved. This follows our intuition built from
        Table~\ref{tab:results}. We however note that our loop optimisation
        techniques may  lose slightly more problems than their un-optimised
        loop versions. This is because our loop optimisation methods may
        perform  some unnecessary and potentially hard subsumption resolutions,
        slightly increasing the likelihood of being stuck on difficult
        combinatorial problems.
\end{compactitem}

\begin{table}
    \setlength{\tabcolsep}{6pt}
    \centering
    \begin{tabular}{l|cc}
        \hline
        \textbf{Prover} & \textbf{Total Solved} & \textbf{Gain/Loss}
        \\
        \hline
        \vampire$_M$ &
            10\,728 & baseline
            \\
        \vampire$_{D}$ &
            10\,762 & ($+62$, $-28$)
            \\
        \vampire$_{I}$ &
            10\,760 & ($+63$, $-31$)
            \\
        \vampire$_{H}$ &
            10\,764 & ($+64$, $-28$)
            \\
        \vampire--cutoff-5000$_{H}$ &
            10\,766 & ($+65$, $-27$)
            \\
        \vampire--cutoff-150$_{H}$ &
            10\,739 & ($+56$, $-45$)
            \\
        \vampire$^*_{D}$ &
            10\,791 & ($+94$, $-31$)
            \\
        \vampire$^*_{I}$ &
            10\,785 & ($+92$, $-35$)
            \\
        \vampire$^*_{H}$ &
            10\,794 & ($+97$, $-31$)
            \\
        \vampire--cutoff-5000$^*_{H}$ &
            10\,790 & ($+97$, $-35$)
            \\
        \vampire--cutoff-150$^*_{H}$ &
            10\,768 & ($+93$, $-53$)
            \\
        \hline
    \end{tabular}\vspace*{0.5em}
    \caption{%
        Number of TPTP problems solved by the considered versions of \textsc{Vampire}.
        The run was made using the options \texttt{-sa otter -av off} with a timeout of~60\,s.
        The \textbf{Gain/Loss} column reports the difference of solved instances compared to \vampire{}$_M$.
        The versions \vampire$^*_\mathcal{E}$ integrate the loop optimisation discussed in Section~\ref{sec:loop-optimization} into \vampire$_\mathcal{E}$.
    }
    \label{tab:performances-vampire-vanilla}
\end{table}

%% file: 8-related.tex
Subsumption and subsumption resolution are some of the
most powerful and frequently used redundancy criteria in saturation-based first-order theorem proving. 

\paragraph{Subsumption.} While efficient literal- and clause-indexing techniques have been proposed~\cite{Tammet:1998,schulz2013simple},
optimising the matching step among multisets of literals, and hence clauses,
has so far not been addressed.
Our work shows that SAT solving methods can provide efficient
solutions in this respect,
further improving first-order theorem proving.

A related approach that integrates multi-literal matching
into indexing is given in~\cite{handbook-indexing}, using code trees. 
Code trees organise potentially subsuming clauses into a tree-like data structure
with the aim of sharing some matching effort for similar clauses.
However, the underlying matching algorithm uses a fixed branching order and does not learn from conflicts,
and will thus run into the same issues on hard subsumption instances as the standard backtracking-based matching.

The specialised subsumption algorithm DC~\cite{GottlobLeitsch:1985a}
is based on the idea of separating the clause~$\sideP$ into variable-disjoint components
and testing subsumption for each component separately.
However, the notion of subsumption considered in that work is defined using subset inclusion,
rather than multiset inclusion.
For subsumption based on multiset inclusion,
the subsumption test for one variable-disjoint component is no longer independent of the other components.

An improved version of~\cite{GottlobLeitsch:1985a} comes with IDC~\cite{GottlobLeitsch:1985b}, whereupon each recursion level is checked whether each literal of~$\sideP$ by itself subsumes~$\mainP$ under the current partial substitution,
which is a necessary condition for subsumption.
The backtracking-based subsumption algorithm of \vampire{} uses this optimisation as well,
and our SAT-based approach also implements it as propagation over substitution constraints.

By combining subsumption and resolution into one simplification rules, subsumption resolution is  supported as {contextual literal cutting} in~\cite{E19},
along with efficient approaches for detecting multiset inclusions
among clauses~\cite{DBLP:conf/cav/KovacsV13,Spass09,schulz2013simple}. 
Special cases of {unit deletion} as a by-product of subsumption tests are also proposed in~\cite{Tammet:1998}, with further  refinements of term indexing to drastically reduce the set of candidate clauses checked for subsumption (resolution).

SAT- and SMT-based techniques have previously been applied to the setting of
first-order saturation-based proof search, e.g.
in the form of the \textsc{Avatar} architecture~\cite{Voronkov:2014:Avatar}.
These techniques are, however, independent from our work,
as they apply the SAT- or SMT-solver over an abstraction of the input problem,
while in our work
we use a SAT solver to speed up certain inferences.

Some solvers, such as the pseudo-boolean solver \textsc{MiniCard}~\cite{LiffitonMaglalang:2012}
and the ASP solver \textsc{Clasp}~\cite{Gebser2009a}, support cardinality constraints natively,
in a similar way to our handling of at-most-one  constraints.
Our encoding, however, requires only at-most-one  constraints
instead of arbitrary cardinality constraints,
thus simplifying the implementation.

Note that 
clausal subsumption can also be seen as a constraint satisfaction problem (CSP).
In this view, the boolean variables $b^+_{ij}$ in our subsumption encoding
represent the different choices of a non-boolean CSP variable,
corresponding to the so-called \emph{direct encoding} of a CSP variable~\cite{Walsh:2000}.
A well-known heuristic in CSP solving is the {minimum remaining values} heuristic:
always assign the CSP variable that has the fewest possible choices remaining.
We adapted this heuristic to our embedded SAT solver and used it to
solve subsumption instances~\cite{rath2022first};
however, it does not fit the subsumption resolution encodings well,
especially the indirect encoding.
Moreover, the
advantage over the well-known variable-move-to-front (VMTF) heuristic~\cite{ryan2004vmtf}
is minor even for subsumption,
which is why we now always use VMTF for variable selection in the SAT solver.

We finally remark that redundancy is  explored in SAT-based equivalence checking~\cite{DBLP:journals/jar/HeuleKB20}, by using first-order and QBF reasoning for subsumption checks~\cite{DBLP:conf/sat/Biere04}.
In particular, first-order backward subsumption~\cite{handbook-indexing} has become a key preprocessing techniques in SAT solving, in particular in bounded variable elimination~\cite{DBLP:conf/sat/EenB05,DBLP:series/faia/BiereJK21}. Our work complements this line of research by showcasing that SAT solving also improves solving variants of first-order subsumption, not just the other way around.

%% file: 9-conclusion.tex
We promote tailored  SAT solving
to solve clausal subsumption and subsumption resolution in first-order theorem proving.
We introduce substitution constraints to encode subsumption constraints as
SAT instance. For solving such instances, we adjust unit propagation
and conflict resolution in SAT solving towards a specialized treatment of
substitution constraints and at-most-one constraints.
Crucially,
our encoding together with our SAT solver
enables efficient setup of subsumption and subsumption resolution instances.
We show that the resulting SAT solver can directly be integrated within the saturation loop of first-order theorem proving, solving both subsumption and subsumption resolution.
Our experimental results indicate that SAT-based subsumption and subsumption resolution significantly
improves the performance of first-order proving.
Extending subsumption with theory reasoning with equality,
possibly in the presence of (arithmetic) first-order theories, is an
interesting task for future work.
We believe this would open up potentially new venues for using
SMT solving instead of SAT solving for subsumption reasoning.

%% file: main.bbl

\begin{thebibliography}{55}
\ifx \bisbn   \undefined \def \bisbn  #1{ISBN #1}\fi
\ifx \binits  \undefined \def \binits#1{#1}\fi
\ifx \bauthor  \undefined \def \bauthor#1{#1}\fi
\ifx \batitle  \undefined \def \batitle#1{#1}\fi
\ifx \bjtitle  \undefined \def \bjtitle#1{#1}\fi
\ifx \bvolume  \undefined \def \bvolume#1{\textbf{#1}}\fi
\ifx \byear  \undefined \def \byear#1{#1}\fi
\ifx \bissue  \undefined \def \bissue#1{#1}\fi
\ifx \bfpage  \undefined \def \bfpage#1{#1}\fi
\ifx \blpage  \undefined \def \blpage #1{#1}\fi
\ifx \burl  \undefined \def \burl#1{\textsf{#1}}\fi
\ifx \doiurl  \undefined \def \doiurl#1{\url{https://doi.org/#1}}\fi
\ifx \betal  \undefined \def \betal{\textit{et al.}}\fi
\ifx \binstitute  \undefined \def \binstitute#1{#1}\fi
\ifx \binstitutionaled  \undefined \def \binstitutionaled#1{#1}\fi
\ifx \bctitle  \undefined \def \bctitle#1{#1}\fi
\ifx \beditor  \undefined \def \beditor#1{#1}\fi
\ifx \bpublisher  \undefined \def \bpublisher#1{#1}\fi
\ifx \bbtitle  \undefined \def \bbtitle#1{#1}\fi
\ifx \bedition  \undefined \def \bedition#1{#1}\fi
\ifx \bseriesno  \undefined \def \bseriesno#1{#1}\fi
\ifx \blocation  \undefined \def \blocation#1{#1}\fi
\ifx \bsertitle  \undefined \def \bsertitle#1{#1}\fi
\ifx \bsnm \undefined \def \bsnm#1{#1}\fi
\ifx \bsuffix \undefined \def \bsuffix#1{#1}\fi
\ifx \bparticle \undefined \def \bparticle#1{#1}\fi
\ifx \barticle \undefined \def \barticle#1{#1}\fi
\bibcommenthead
\ifx \bconfdate \undefined \def \bconfdate #1{#1}\fi
\ifx \botherref \undefined \def \botherref #1{#1}\fi
\ifx \url \undefined \def \url#1{\textsf{#1}}\fi
\ifx \bchapter \undefined \def \bchapter#1{#1}\fi
\ifx \bbook \undefined \def \bbook#1{#1}\fi
\ifx \bcomment \undefined \def \bcomment#1{#1}\fi
\ifx \oauthor \undefined \def \oauthor#1{#1}\fi
\ifx \citeauthoryear \undefined \def \citeauthoryear#1{#1}\fi
\ifx \endbibitem  \undefined \def \endbibitem {}\fi
\ifx \bconflocation  \undefined \def \bconflocation#1{#1}\fi
\ifx \arxivurl  \undefined \def \arxivurl#1{\textsf{#1}}\fi
\csname PreBibitemsHook\endcsname

\bibitem[\protect\citeauthoryear{Leino}{2017}]{DBLP:journals/software/Leino17}
\begin{barticle}
\bauthor{\bsnm{Leino}, \binits{K.R.M.}}:
\batitle{{Accessible Software Verification with Dafny}}.
\bjtitle{{IEEE} Softw.}
\bvolume{34}(\bissue{6}),
\bfpage{94}--\blpage{97}
(\byear{2017})
\end{barticle}
\endbibitem

\bibitem[\protect\citeauthoryear{Clochard et~al.}{2020}]{DBLP:journals/pacmpl/ClochardMP20}
\begin{botherref}
\oauthor{\bsnm{Clochard}, \binits{M.}},
\oauthor{\bsnm{March{\'{e}}}, \binits{C.}},
\oauthor{\bsnm{Paskevich}, \binits{A.}}:
{Deductive Verification with Ghost Monitors}.
Proc. of POPL,
2--1226
(2020)
\end{botherref}
\endbibitem

\bibitem[\protect\citeauthoryear{Georgiou et~al.}{2020}]{DBLP:conf/fmcad/GeorgiouGK20}
\begin{bchapter}
\bauthor{\bsnm{Georgiou}, \binits{P.}},
\bauthor{\bsnm{Gleiss}, \binits{B.}},
\bauthor{\bsnm{Kov{\'{a}}cs}, \binits{L.}}:
\bctitle{{Trace Logic for Inductive Loop Reasoning}}.
In: \bbtitle{Proc. of FMCAD},
pp. \bfpage{255}--\blpage{263}
(\byear{2020})
\end{bchapter}
\endbibitem

\bibitem[\protect\citeauthoryear{Komuravelli et~al.}{2016}]{SPACER16}
\begin{barticle}
\bauthor{\bsnm{Komuravelli}, \binits{A.}},
\bauthor{\bsnm{Gurfinkel}, \binits{A.}},
\bauthor{\bsnm{Chaki}, \binits{S.}}:
\batitle{{SMT-Based Model Checking for Recursive Programs}}.
\bjtitle{Formal Methods Syst. Des.}
\bvolume{48}(\bissue{3}),
\bfpage{175}--\blpage{205}
(\byear{2016})
\end{barticle}
\endbibitem

\bibitem[\protect\citeauthoryear{Padon et~al.}{2016}]{IVY16}
\begin{bchapter}
\bauthor{\bsnm{Padon}, \binits{O.}},
\bauthor{\bsnm{McMillan}, \binits{K.L.}},
\bauthor{\bsnm{Panda}, \binits{A.}},
\bauthor{\bsnm{Sagiv}, \binits{M.}},
\bauthor{\bsnm{Shoham}, \binits{S.}}:
\bctitle{{Ivy: Safety Verification by Interactive Generalization}}.
In: \bbtitle{Proc. of PLDI},
pp. \bfpage{614}--\blpage{630}
(\byear{2016})
\end{bchapter}
\endbibitem

\bibitem[\protect\citeauthoryear{Asadi et~al.}{2020}]{DBLP:conf/fmcad/AsadiBHFS20}
\begin{bchapter}
\bauthor{\bsnm{Asadi}, \binits{S.}},
\bauthor{\bsnm{Blicha}, \binits{M.}},
\bauthor{\bsnm{Hyv{\"{a}}rinen}, \binits{A.E.J.}},
\bauthor{\bsnm{Fedyukovich}, \binits{G.}},
\bauthor{\bsnm{Sharygina}, \binits{N.}}:
\bctitle{{Incremental Verification by SMT-based Summary Repair}}.
In: \bbtitle{Proc. of FMCAD},
pp. \bfpage{77}--\blpage{82}
(\byear{2020})
\end{bchapter}
\endbibitem

\bibitem[\protect\citeauthoryear{Garcia{-}Contreras et~al.}{2023}]{DBLP:conf/cav/GarciaContrerasKSG23}
\begin{bchapter}
\bauthor{\bsnm{Garcia{-}Contreras}, \binits{I.}},
\bauthor{\bsnm{K.}, \binits{H.G.V.}},
\bauthor{\bsnm{Shoham}, \binits{S.}},
\bauthor{\bsnm{Gurfinkel}, \binits{A.}}:
\bctitle{{Fast Approximations of Quantifier Elimination}}.
In: \bbtitle{Proc. of {CAV}},
pp. \bfpage{64}--\blpage{86}
(\byear{2023}).
\doiurl{10.1007/978-3-031-37703-7\_4}
\end{bchapter}
\endbibitem

\bibitem[\protect\citeauthoryear{Pick et~al.}{2020}]{DBLP:conf/fmcad/PickFG20}
\begin{bchapter}
\bauthor{\bsnm{Pick}, \binits{L.}},
\bauthor{\bsnm{Fedyukovich}, \binits{G.}},
\bauthor{\bsnm{Gupta}, \binits{A.}}:
\bctitle{{Automating Modular Verification of Secure Information Flow}}.
In: \bbtitle{Proc. of FMCAD},
pp. \bfpage{158}--\blpage{168}
(\byear{2020})
\end{bchapter}
\endbibitem

\bibitem[\protect\citeauthoryear{Mart{\'{\i}}nez et~al.}{2019}]{DBLP:conf/esop/MartinezADGHHNP19}
\begin{bchapter}
\bauthor{\bsnm{Mart{\'{\i}}nez}, \binits{G.}},
\bauthor{\bsnm{Ahman}, \binits{D.}},
\bauthor{\bsnm{Dumitrescu}, \binits{V.}},
\bauthor{\bsnm{Giannarakis}, \binits{N.}},
\bauthor{\bsnm{Hawblitzel}, \binits{C.}},
\bauthor{\bsnm{Hritcu}, \binits{C.}},
\bauthor{\bsnm{Narasimhamurthy}, \binits{M.}},
\bauthor{\bsnm{Paraskevopoulou}, \binits{Z.}},
\bauthor{\bsnm{Pit{-}Claudel}, \binits{C.}},
\bauthor{\bsnm{Protzenko}, \binits{J.}},
\bauthor{\bsnm{Ramananandro}, \binits{T.}},
\bauthor{\bsnm{Rastogi}, \binits{A.}},
\bauthor{\bsnm{Swamy}, \binits{N.}}:
\bctitle{{Meta-F\(^\star\): Proof Automation with SMT, Tactics, and Metaprograms}}.
In: \bbtitle{Proc. of ESOP},
pp. \bfpage{30}--\blpage{59}
(\byear{2019})
\end{bchapter}
\endbibitem

\bibitem[\protect\citeauthoryear{Veronese et~al.}{2023}]{DBLP:conf/sp/VeroneseFBTSM23}
\begin{bchapter}
\bauthor{\bsnm{Veronese}, \binits{L.}},
\bauthor{\bsnm{Farinier}, \binits{B.}},
\bauthor{\bsnm{Bernardo}, \binits{P.}},
\bauthor{\bsnm{Tempesta}, \binits{M.}},
\bauthor{\bsnm{Squarcina}, \binits{M.}},
\bauthor{\bsnm{Maffei}, \binits{M.}}:
\bctitle{{WebSpec: Towards Machine-Checked Analysis of Browser Security Mechanisms}}.
In: \bbtitle{SP},
pp. \bfpage{2761}--\blpage{2779}
(\byear{2023}).
\doiurl{10.1109/SP46215.2023.10179465}
\end{bchapter}
\endbibitem

\bibitem[\protect\citeauthoryear{Brugger et~al.}{2023}]{DBLP:conf/ccs/BruggerKKR023}
\begin{bchapter}
\bauthor{\bsnm{Brugger}, \binits{L.S.}},
\bauthor{\bsnm{Kov{\'{a}}cs}, \binits{L.}},
\bauthor{\bsnm{Komel}, \binits{A.P.}},
\bauthor{\bsnm{Rain}, \binits{S.}},
\bauthor{\bsnm{Rawson}, \binits{M.}}:
\bctitle{{CheckMate: Automated Game-Theoretic Security Reasoning}}.
In: \bbtitle{CCS},
pp. \bfpage{1407}--\blpage{1421}
(\byear{2023}).
\doiurl{10.1145/3576915.3623183}
\end{bchapter}
\endbibitem

\bibitem[\protect\citeauthoryear{Biere}{2008}]{Picosat}
\begin{barticle}
\bauthor{\bsnm{Biere}, \binits{A.}}:
\batitle{{PicoSAT Essentials}}.
\bjtitle{J. Satisf. Boolean Model. Comput.}
\bvolume{4}(\bissue{2-4}),
\bfpage{75}--\blpage{97}
(\byear{2008})
\end{barticle}
\endbibitem

\bibitem[\protect\citeauthoryear{De~Moura and Bj{\o}rner}{2008}]{Z3}
\begin{bchapter}
\bauthor{\bsnm{De~Moura}, \binits{L.}},
\bauthor{\bsnm{Bj{\o}rner}, \binits{N.}}:
\bctitle{{Z3: An Efficient SMT Solver}}.
In: \bbtitle{Proc. of TACAS},
pp. \bfpage{337}--\blpage{340}
(\byear{2008})
\end{bchapter}
\endbibitem

\bibitem[\protect\citeauthoryear{Barbosa et~al.}{2022}]{CVC5}
\begin{bchapter}
\bauthor{\bsnm{Barbosa}, \binits{H.}},
\bauthor{\bsnm{Barrett}, \binits{C.W.}},
\bauthor{\bsnm{Brain}, \binits{M.}},
\bauthor{\bsnm{Kremer}, \binits{G.}},
\bauthor{\bsnm{Lachnitt}, \binits{H.}},
\bauthor{\bsnm{Mann}, \binits{M.}},
\bauthor{\bsnm{Mohamed}, \binits{A.}},
\bauthor{\bsnm{Mohamed}, \binits{M.}},
\bauthor{\bsnm{Niemetz}, \binits{A.}},
\bauthor{\bsnm{N{\"{o}}tzli}, \binits{A.}},
\bauthor{\bsnm{Ozdemir}, \binits{A.}},
\bauthor{\bsnm{Preiner}, \binits{M.}},
\bauthor{\bsnm{Reynolds}, \binits{A.}},
\bauthor{\bsnm{Sheng}, \binits{Y.}},
\bauthor{\bsnm{Tinelli}, \binits{C.}},
\bauthor{\bsnm{Zohar}, \binits{Y.}}:
\bctitle{{CVC5: {A} Versatile and Industrial-Strength {SMT} Solver}}.
In: \bbtitle{Proc. of TACAS},
pp. \bfpage{415}--\blpage{442}
(\byear{2022})
\end{bchapter}
\endbibitem

\bibitem[\protect\citeauthoryear{Weidenbach et~al.}{2009}]{Spass09}
\begin{bchapter}
\bauthor{\bsnm{Weidenbach}, \binits{C.}},
\bauthor{\bsnm{Dimova}, \binits{D.}},
\bauthor{\bsnm{Fietzke}, \binits{A.}},
\bauthor{\bsnm{Kumar}, \binits{R.}},
\bauthor{\bsnm{Suda}, \binits{M.}},
\bauthor{\bsnm{Wischnewski}, \binits{P.}}:
\bctitle{{SPASS} version 3.5}.
In: \bbtitle{Proc. of CADE},
pp. \bfpage{140}--\blpage{145}
(\byear{2009})
\end{bchapter}
\endbibitem

\bibitem[\protect\citeauthoryear{Kov{\'a}cs and Voronkov}{2013}]{kovacs2013first}
\begin{bchapter}
\bauthor{\bsnm{Kov{\'a}cs}, \binits{L.}},
\bauthor{\bsnm{Voronkov}, \binits{A.}}:
\bctitle{{First-Order Theorem Proving and Vampire}}.
In: \bbtitle{CAV},
pp. \bfpage{1}--\blpage{35}
(\byear{2013})
\end{bchapter}
\endbibitem

\bibitem[\protect\citeauthoryear{Schulz et~al.}{2019}]{E19}
\begin{bchapter}
\bauthor{\bsnm{Schulz}, \binits{S.}},
\bauthor{\bsnm{Cruanes}, \binits{S.}},
\bauthor{\bsnm{Vukmirovic}, \binits{P.}}:
\bctitle{{Faster, Higher, Stronger: {E} 2.3}}.
In: \bbtitle{Proc. of CADE},
pp. \bfpage{495}--\blpage{507}
(\byear{2019})
\end{bchapter}
\endbibitem

\bibitem[\protect\citeauthoryear{Cruanes}{2017}]{zipperposition}
\begin{bchapter}
\bauthor{\bsnm{Cruanes}, \binits{S.}}:
\bctitle{{Superposition with Structural Induction}}.
In: \bbtitle{Proc. of FroCoS},
pp. \bfpage{172}--\blpage{188}
(\byear{2017})
\end{bchapter}
\endbibitem

\bibitem[\protect\citeauthoryear{Buchberger}{2006}]{Buchberger06a}
\begin{barticle}
\bauthor{\bsnm{Buchberger}, \binits{B.}}:
\batitle{{Bruno Buchberger's PhD thesis 1965: An Algorithm for Finding the Basis Elements of the Residue Class Ring of a Zero Dimensional Polynomial Ideal}}.
\bjtitle{J. Symb. Comput.}
\bvolume{41}(\bissue{3-4}),
\bfpage{475}--\blpage{511}
(\byear{2006})
\doiurl{10.1016/j.jsc.2005.09.007}
\end{barticle}
\endbibitem

\bibitem[\protect\citeauthoryear{Nieuwenhuis and Rubio}{2001}]{Rubio01}
\begin{bchapter}
\bauthor{\bsnm{Nieuwenhuis}, \binits{R.}},
\bauthor{\bsnm{Rubio}, \binits{A.}}:
\bctitle{{Paramodulation-Based Theorem Proving}}.
In: \bbtitle{Handbook of Automated Reasoning},
pp. \bfpage{371}--\blpage{443}
(\byear{2001}).
\doiurl{10.1016/b978-044450813-3/50009-6} .
\burl{https://doi.org/10.1016/b978-044450813-3/50009-6}
\end{bchapter}
\endbibitem

\bibitem[\protect\citeauthoryear{Robinson}{1965}]{Robinson65}
\begin{barticle}
\bauthor{\bsnm{Robinson}, \binits{J.A.}}:
\batitle{{A Machine-Oriented Logic Based on the Resolution Principle}}.
\bjtitle{J. {ACM}}
\bvolume{12}(\bissue{1}),
\bfpage{23}--\blpage{41}
(\byear{1965})
\doiurl{10.1145/321250.321253}
\end{barticle}
\endbibitem

\bibitem[\protect\citeauthoryear{Bachmair and Ganzinger}{1994}]{BG94}
\begin{barticle}
\bauthor{\bsnm{Bachmair}, \binits{L.}},
\bauthor{\bsnm{Ganzinger}, \binits{H.}}:
\batitle{{Rewrite-Based Equational Theorem Proving with Selection and Simplification}}.
\bjtitle{J. Log. Comput.}
\bvolume{4}(\bissue{3}),
\bfpage{217}--\blpage{247}
(\byear{1994})
\end{barticle}
\endbibitem

\bibitem[\protect\citeauthoryear{Biere}{2004}]{DBLP:conf/sat/Biere04}
\begin{bchapter}
\bauthor{\bsnm{Biere}, \binits{A.}}:
\bctitle{Resolve and expand}.
In: \bbtitle{Proc. of {SAT}}
(\byear{2004}).
\doiurl{10.1007/11527695\_5}
\end{bchapter}
\endbibitem

\bibitem[\protect\citeauthoryear{Sekar et~al.}{2001}]{handbook-indexing}
\begin{bchapter}
\bauthor{\bsnm{Sekar}, \binits{R.}},
\bauthor{\bsnm{Ramakrishnan}, \binits{I.V.}},
\bauthor{\bsnm{Voronkov}, \binits{A.}}:
\bctitle{{Term Indexing}}.
In: \bbtitle{Handbook of Automated Reasoning},
pp. \bfpage{1853}--\blpage{1964}
(\byear{2001})
\end{bchapter}
\endbibitem

\bibitem[\protect\citeauthoryear{Nieuwenhuis et~al.}{2001}]{voronkov-evaluation-indexing}
\begin{bchapter}
\bauthor{\bsnm{Nieuwenhuis}, \binits{R.}},
\bauthor{\bsnm{Hillenbrand}, \binits{T.}},
\bauthor{\bsnm{Riazanov}, \binits{A.}},
\bauthor{\bsnm{Voronkov}, \binits{A.}}:
\bctitle{{On the Evaluation of Indexing Techniques for Theorem Proving}}.
In: \bbtitle{Proc. of IJCAR},
pp. \bfpage{257}--\blpage{271}
(\byear{2001})
\end{bchapter}
\endbibitem

\bibitem[\protect\citeauthoryear{Schulz}{2013}]{schulz2013simple}
\begin{bchapter}
\bauthor{\bsnm{Schulz}, \binits{S.}}:
\bctitle{{Simple and Efficient Clause Subsumption with Feature Vector Indexing}}.
In: \bbtitle{Automated Reasoning and Mathematics - Essays in Memory of William W. McCune},
pp. \bfpage{45}--\blpage{67}
(\byear{2013})
\end{bchapter}
\endbibitem

\bibitem[\protect\citeauthoryear{Kapur and Narendran}{1986}]{matching-np-complete}
\begin{bchapter}
\bauthor{\bsnm{Kapur}, \binits{D.}},
\bauthor{\bsnm{Narendran}, \binits{P.}}:
\bctitle{{{NP}-Completeness of the Set Unification and Matching Problems}}.
In: \bbtitle{Proc. of CADE},
pp. \bfpage{489}--\blpage{495}
(\byear{1986})
\end{bchapter}
\endbibitem

\bibitem[\protect\citeauthoryear{Rath et~al.}{2022}]{rath2022first}
\begin{bchapter}
\bauthor{\bsnm{Rath}, \binits{J.}},
\bauthor{\bsnm{Biere}, \binits{A.}},
\bauthor{\bsnm{Kov{\'a}cs}, \binits{L.}}:
\bctitle{{First-Order Subsumption via {SAT} Solving}}.
In: \bbtitle{FMCAD},
p. \bfpage{160}
(\byear{2022})
\end{bchapter}
\endbibitem

\bibitem[\protect\citeauthoryear{Coutelier et~al.}{2023}]{DBLP:conf/cade/CoutelierKRR23}
\begin{bchapter}
\bauthor{\bsnm{Coutelier}, \binits{R.}},
\bauthor{\bsnm{Kov{\'{a}}cs}, \binits{L.}},
\bauthor{\bsnm{Rawson}, \binits{M.}},
\bauthor{\bsnm{Rath}, \binits{J.}}:
\bctitle{{SAT-Based Subsumption Resolution}}.
In: \bbtitle{Proc. of {CADE}},
pp. \bfpage{190}--\blpage{206}
(\byear{2023}).
\doiurl{10.1007/978-3-031-38499-8\_11}
\end{bchapter}
\endbibitem

\bibitem[\protect\citeauthoryear{Gleiss et~al.}{2020}]{gleiss-sd}
\begin{bchapter}
\bauthor{\bsnm{Gleiss}, \binits{B.}},
\bauthor{\bsnm{Kov{\'{a}}cs}, \binits{L.}},
\bauthor{\bsnm{Rath}, \binits{J.}}:
\bctitle{{Subsumption Demodulation in First-Order Theorem Proving}}.
In: \bbtitle{Proc. of IJCAR},
pp. \bfpage{297}--\blpage{315}
(\byear{2020})
\end{bchapter}
\endbibitem

\bibitem[\protect\citeauthoryear{E{\'{e}}n and S{\"{o}}rensson}{2003}]{DBLP:conf/sat/EenS03}
\begin{bchapter}
\bauthor{\bsnm{E{\'{e}}n}, \binits{N.}},
\bauthor{\bsnm{S{\"{o}}rensson}, \binits{N.}}:
\bctitle{An extensible {SAT}-solver}.
In: \bbtitle{Proc. of {SAT}},
pp. \bfpage{502}--\blpage{518}
(\byear{2003}).
\doiurl{10.1007/978-3-540-24605-3\_37}
\end{bchapter}
\endbibitem

\bibitem[\protect\citeauthoryear{Biere et~al.}{}]{DBLP:conf/sat/BiereFW23}
\begin{botherref}
\oauthor{\bsnm{Biere}, \binits{A.}},
\oauthor{\bsnm{Froleyks}, \binits{N.}},
\oauthor{\bsnm{Wang}, \binits{W.}}:
{CadiBack: Extracting Backbones with CaDiCaL}.
In: Proc. of {SAT},
pp. 3--1312.
\doiurl{10.4230/LIPICS.SAT.2023.3}
\end{botherref}
\endbibitem

\bibitem[\protect\citeauthoryear{Fleury and Biere}{2022}]{DBLP:journals/fmsd/FleuryB22}
\begin{barticle}
\bauthor{\bsnm{Fleury}, \binits{M.}},
\bauthor{\bsnm{Biere}, \binits{A.}}:
\batitle{Mining definitions in kissat with kittens}.
\bjtitle{Formal Methods Syst. Des.}
\bvolume{60}(\bissue{3}),
\bfpage{381}--\blpage{404}
(\byear{2022})
\doiurl{10.1007/S10703-023-00421-2}
\end{barticle}
\endbibitem

\bibitem[\protect\citeauthoryear{Marques-Silva et~al.}{2021}]{Marques-SilvaLynceMalik:2021}
\begin{bchapter}
\bauthor{\bsnm{Marques-Silva}, \binits{J.}},
\bauthor{\bsnm{Lynce}, \binits{I.}},
\bauthor{\bsnm{Malik}, \binits{S.}}:
\bctitle{{Conflict-Driven Clause Learning {SAT} Solvers}}.
In: \bbtitle{Handbook of Satisfiability}.
\bsertitle{Frontiers in Artificial Intelligence and Applications},
vol. \bseriesno{336},
pp. \bfpage{133}--\blpage{182}
(\byear{2021}).
\bcomment{Chap. 4}
\end{bchapter}
\endbibitem

\bibitem[\protect\citeauthoryear{Moskewicz et~al.}{2001}]{MoskewiczMadiganZhaoZhangMalik:2001:Chaff}
\begin{bchapter}
\bauthor{\bsnm{Moskewicz}, \binits{M.W.}},
\bauthor{\bsnm{Madigan}, \binits{C.F.}},
\bauthor{\bsnm{Zhao}, \binits{Y.}},
\bauthor{\bsnm{Zhang}, \binits{L.}},
\bauthor{\bsnm{Malik}, \binits{S.}}:
\bctitle{{{Chaff}: Engineering an Efficient {SAT} Solver}}.
In: \bbtitle{Proc. of DAC},
pp. \bfpage{530}--\blpage{535}
(\byear{2001})
\end{bchapter}
\endbibitem

\bibitem[\protect\citeauthoryear{Frisch and Giannaros}{2010}]{FrischGiannaros:2010}
\begin{bchapter}
\bauthor{\bsnm{Frisch}, \binits{A.M.}},
\bauthor{\bsnm{Giannaros}, \binits{P.A.}}:
\bctitle{{{SAT} Encodings of the At-Most-k Constraint. Some Old, Some New, Some Fast, Some Slow}}.
In: \bbtitle{Proc. of WS on Constraint Modelling and Reformulation}
(\byear{2010})
\end{bchapter}
\endbibitem

\bibitem[\protect\citeauthoryear{McCune and Wos}{1997}]{otter}
\begin{barticle}
\bauthor{\bsnm{McCune}, \binits{W.}},
\bauthor{\bsnm{Wos}, \binits{L.}}:
\batitle{\textsc{Otter}--- the {CADE-13} competition incarnations}.
\bjtitle{Journal of Automated Reasoning}
\bvolume{18},
\bfpage{211}--\blpage{220}
(\byear{1997})
\end{barticle}
\endbibitem

\bibitem[\protect\citeauthoryear{Voronkov}{2014}]{Voronkov:2014:Avatar}
\begin{bchapter}
\bauthor{\bsnm{Voronkov}, \binits{A.}}:
\bctitle{{{AVATAR}: The Architecture for First-Order Theorem Provers}}.
In: \bbtitle{Proc. of CAV},
pp. \bfpage{696}--\blpage{710}
(\byear{2014}).
\doiurl{10.1007/978-3-319-08867-9\_46}
\end{bchapter}
\endbibitem

\bibitem[\protect\citeauthoryear{Biere et~al.}{2020}]{kissat}
\begin{bchapter}
\bauthor{\bsnm{Biere}, \binits{A.}},
\bauthor{\bsnm{Fazekas}, \binits{K.}},
\bauthor{\bsnm{Fleury}, \binits{M.}},
\bauthor{\bsnm{Heisinger}, \binits{M.}}:
\bctitle{\textsc{CaDiCaL}, \textsc{Kissat}, \textsc{Paracooba}, \textsc{Plingeling} and \textsc{Treengeling} entering the {SAT} competition 2020}.
In: \bbtitle{Proc. of {SAT} Competition 2020: Solver and Benchmark Descriptions},
pp. \bfpage{50}--\blpage{53}
(\byear{2020}).
\burl{http://hdl.handle.net/10138/318450}
\end{bchapter}
\endbibitem

\bibitem[\protect\citeauthoryear{Sutcliffe}{2017}]{Sutcliffe:2017:TPTP}
\begin{barticle}
\bauthor{\bsnm{Sutcliffe}, \binits{G.}}:
\batitle{{The TPTP Problem Library and Associated Infrastructure. From CNF to TH0, TPTP v6.4.0}}.
\bjtitle{J. of Automated Reasoning}
\bvolume{59}(\bissue{4}),
\bfpage{483}--\blpage{502}
(\byear{2017})
\end{barticle}
\endbibitem

\bibitem[\protect\citeauthoryear{Pedregosa et~al.}{2011}]{scikit-learn}
\begin{barticle}
\bauthor{\bsnm{Pedregosa}, \binits{F.}},
\bauthor{\bsnm{Varoquaux}, \binits{G.}},
\bauthor{\bsnm{Gramfort}, \binits{A.}},
\bauthor{\bsnm{Michel}, \binits{V.}},
\bauthor{\bsnm{Thirion}, \binits{B.}},
\bauthor{\bsnm{Grisel}, \binits{O.}},
\bauthor{\bsnm{Blondel}, \binits{M.}},
\bauthor{\bsnm{Prettenhofer}, \binits{P.}},
\bauthor{\bsnm{Weiss}, \binits{R.}},
\bauthor{\bsnm{Dubourg}, \binits{V.}},
\bauthor{\bsnm{Vanderplas}, \binits{J.}},
\bauthor{\bsnm{Passos}, \binits{A.}},
\bauthor{\bsnm{Cournapeau}, \binits{D.}},
\bauthor{\bsnm{Brucher}, \binits{M.}},
\bauthor{\bsnm{Perrot}, \binits{M.}},
\bauthor{\bsnm{Duchesnay}, \binits{E.}}:
\batitle{Scikit-learn: Machine learning in {P}ython}.
\bjtitle{Journal of Machine Learning Research}
\bvolume{12},
\bfpage{2825}--\blpage{2830}
(\byear{2011})
\end{barticle}
\endbibitem

\bibitem[\protect\citeauthoryear{Goodfellow et~al.}{2016}]{DBLP:books/daglib/0040158}
\begin{bbook}
\bauthor{\bsnm{Goodfellow}, \binits{I.J.}},
\bauthor{\bsnm{Bengio}, \binits{Y.}},
\bauthor{\bsnm{Courville}, \binits{A.C.}}:
\bbtitle{Deep Learning}.
\bsertitle{Adaptive computation and machine learning},
(\byear{2016}).
\burl{http://www.deeplearningbook.org/}
\end{bbook}
\endbibitem

\bibitem[\protect\citeauthoryear{Beyer et~al.}{2017}]{BeyerLoweWendler:2017:benchexec}
\begin{barticle}
\bauthor{\bsnm{Beyer}, \binits{D.}},
\bauthor{\bsnm{L{\"o}we}, \binits{S.}},
\bauthor{\bsnm{Wendler}, \binits{P.}}:
\batitle{{Reliable Benchmarking: Requirements and Solutions}}.
\bjtitle{J. on Software Tools for Technology Transfer}
\bvolume{21}(\bissue{1}),
\bfpage{1}--\blpage{29}
(\byear{2017})
\end{barticle}
\endbibitem

\bibitem[\protect\citeauthoryear{Tange}{2018}]{Tange:2018:GNUParallel}
\begin{bbook}
\bauthor{\bsnm{Tange}, \binits{O.}}:
\bbtitle{{GNU} Parallel 2018},
(\byear{2018})
\end{bbook}
\endbibitem

\bibitem[\protect\citeauthoryear{Tammet}{1998}]{Tammet:1998}
\begin{bchapter}
\bauthor{\bsnm{Tammet}, \binits{T.}}:
\bctitle{{Towards Efficient Subsumption}}.
In: \bbtitle{Proc. of CADE},
pp. \bfpage{427}--\blpage{441}
(\byear{1998})
\end{bchapter}
\endbibitem

\bibitem[\protect\citeauthoryear{Gottlob and Leitsch}{1985a}]{GottlobLeitsch:1985a}
\begin{barticle}
\bauthor{\bsnm{Gottlob}, \binits{G.}},
\bauthor{\bsnm{Leitsch}, \binits{A.}}:
\batitle{{On the Efficiency of Subsumption Algorithms}}.
\bjtitle{J. of the {ACM}}
\bvolume{32}(\bissue{2}),
\bfpage{280}--\blpage{295}
(\byear{1985})
\end{barticle}
\endbibitem

\bibitem[\protect\citeauthoryear{Gottlob and Leitsch}{1985b}]{GottlobLeitsch:1985b}
\begin{bchapter}
\bauthor{\bsnm{Gottlob}, \binits{G.}},
\bauthor{\bsnm{Leitsch}, \binits{A.}}:
\bctitle{{Fast Subsumption Algorithms}}.
In: \bbtitle{Proc. of {EUROCAL} {\textquotesingle}85},
pp. \bfpage{64}--\blpage{77}
(\byear{1985})
\end{bchapter}
\endbibitem

\bibitem[\protect\citeauthoryear{Kov{\'{a}}cs and Voronkov}{2013}]{DBLP:conf/cav/KovacsV13}
\begin{bchapter}
\bauthor{\bsnm{Kov{\'{a}}cs}, \binits{L.}},
\bauthor{\bsnm{Voronkov}, \binits{A.}}:
\bctitle{{First-Order Theorem Proving and Vampire}}.
In: \bbtitle{CAV},
pp. \bfpage{1}--\blpage{35}
(\byear{2013}).
\doiurl{10.1007/978-3-642-39799-8\_1} .
\burl{https://doi.org/10.1007/978-3-642-39799-8\_1}
\end{bchapter}
\endbibitem

\bibitem[\protect\citeauthoryear{Liffiton and Maglalang}{2012}]{LiffitonMaglalang:2012}
\begin{bchapter}
\bauthor{\bsnm{Liffiton}, \binits{M.H.}},
\bauthor{\bsnm{Maglalang}, \binits{J.C.}}:
\bctitle{{A Cardinality Solver: More Expressive Constraints for Free}}.
In: \bbtitle{Proc. of {SAT}},
pp. \bfpage{485}--\blpage{486}
(\byear{2012})
\end{bchapter}
\endbibitem

\bibitem[\protect\citeauthoryear{Gebser et~al.}{2009}]{Gebser2009a}
\begin{bchapter}
\bauthor{\bsnm{Gebser}, \binits{M.}},
\bauthor{\bsnm{Kaminski}, \binits{R.}},
\bauthor{\bsnm{Kaufmann}, \binits{B.}},
\bauthor{\bsnm{Schaub}, \binits{T.}}:
\bctitle{{On the Implementation of Weight Constraint Rules in Conflict-Driven ASP Solvers}}.
In: \bbtitle{Proc. of {ICLP}},
pp. \bfpage{250}--\blpage{264}
(\byear{2009})
\end{bchapter}
\endbibitem

\bibitem[\protect\citeauthoryear{Walsh}{2000}]{Walsh:2000}
\begin{bchapter}
\bauthor{\bsnm{Walsh}, \binits{T.}}:
\bctitle{{SAT} v {CSP}}.
In: \bbtitle{Proc. of CP},
pp. \bfpage{441}--\blpage{456}
(\byear{2000})
\end{bchapter}
\endbibitem

\bibitem[\protect\citeauthoryear{Ryan}{2004}]{ryan2004vmtf}
\begin{botherref}
\oauthor{\bsnm{Ryan}, \binits{L.}}:
Efficient algorithms for clause-learning {SAT} solvers.
Master's thesis,
Simon Fraser University
(2004)
\end{botherref}
\endbibitem

\bibitem[\protect\citeauthoryear{Heule et~al.}{2020}]{DBLP:journals/jar/HeuleKB20}
\begin{barticle}
\bauthor{\bsnm{Heule}, \binits{M.J.H.}},
\bauthor{\bsnm{Kiesl}, \binits{B.}},
\bauthor{\bsnm{Biere}, \binits{A.}}:
\batitle{{Strong Extension-Free Proof Systems}}.
\bjtitle{J. Autom. Reason.}
\bvolume{64}(\bissue{3}),
\bfpage{533}--\blpage{554}
(\byear{2020})
\doiurl{10.1007/S10817-019-09516-0}
\end{barticle}
\endbibitem

\bibitem[\protect\citeauthoryear{E{\'{e}}n and Biere}{2005}]{DBLP:conf/sat/EenB05}
\begin{bchapter}
\bauthor{\bsnm{E{\'{e}}n}, \binits{N.}},
\bauthor{\bsnm{Biere}, \binits{A.}}:
\bctitle{{Effective Preprocessing in {SAT} Through Variable and Clause Elimination}}.
In: \bbtitle{Proc. of {SAT}},
vol. \bseriesno{3569},
pp. \bfpage{61}--\blpage{75}
(\byear{2005}).
\doiurl{10.1007/11499107\_5}
\end{bchapter}
\endbibitem

\bibitem[\protect\citeauthoryear{Biere et~al.}{2021}]{DBLP:series/faia/BiereJK21}
\begin{bchapter}
\bauthor{\bsnm{Biere}, \binits{A.}},
\bauthor{\bsnm{J{\"{a}}rvisalo}, \binits{M.}},
\bauthor{\bsnm{Kiesl}, \binits{B.}}:
\bctitle{Preprocessing in {SAT} solving}.
In: \bbtitle{Handbook of Satisfiability - Second Edition}.
\bsertitle{Frontiers in Artificial Intelligence and Applications},
vol. \bseriesno{336},
pp. \bfpage{391}--\blpage{435}
(\byear{2021}).
\doiurl{10.3233/FAIA200992}
\end{bchapter}
\endbibitem

\end{thebibliography}
